\documentclass{article}

\usepackage{arxiv}

\usepackage[utf8]{inputenc} 
\usepackage[T1]{fontenc}    
\usepackage{hyperref}       
\usepackage{url}            
\usepackage{booktabs}       
\usepackage{array}
\usepackage{amsfonts}       
\usepackage{nicefrac}       
\usepackage{graphicx}
\usepackage{subcaption}
\usepackage{adjustbox}
\usepackage[numbers]{natbib}
\usepackage{doi}
\usepackage{amsmath}
\usepackage{amssymb}
\usepackage{algorithm}
\usepackage{algpseudocode}
\usepackage{xcolor}
\usepackage{fancyhdr}
\usepackage{amsthm}
\usepackage{cleveref}
\usepackage{bm}
\usepackage[mathscr]{euscript}
\usepackage{textcomp}

\DeclareMathAlphabet{\pazocal}{OMS}{zplm}{m}{n}

\usepackage{lscape}		
\usepackage{fancyvrb}

\def \X  {\pmb{\pazocal{X}}}

\newtheoremstyle{remarkstyle}  
  {5pt}                        
  {5pt}                        
  {}                           
  {}                           
  {\bfseries}                  
  {.}                          
  { }                          
  {}                           
\newtheorem{theorem}{Theorem}[section]

\newtheorem{lemma}[theorem]{Lemma}

\theoremstyle{remarkstyle}
\newtheorem{remark}[theorem]{Remark}

\crefname{figure}{Fig.}{Figs.}
\Crefname{figure}{Figure}{Figures}
\crefname{algorithm}{Algorithm}{Algorithms}
\crefname{equation}{Eq.}{Eqs.}

\title{Neural Chaos:\\ A Spectral Stochastic Neural Operator}

\author{Bahador Bahmani\thanks{Corresponding authors.}\\
	Hopkins Extreme Materials Institute\\
        Dept. of Civil and Systems Engineering\\
	Johns Hopkins University\\
	Baltimore, USA\\
	\texttt{bbahman2@jh.edu} \\
	\And
        Ioannis G. Kevrekidis\\
        Dept. of Chemical and Biomolecular Engineering\\
        Dept. of Applied Mathematics and Statistics\\
	Johns Hopkins University\\
	Baltimore, USA\\
	\texttt{yannisk@jhu.edu} \\
	\And
        Michael D. Shields\textsuperscript{*}\\
        Dept. of Civil and Systems Engineering\\
	Johns Hopkins University\\
	Baltimore, USA\\
	\texttt{michael.shields@jhu.edu} \\
}

\renewcommand{\shorttitle}{Neural Chaos: A Spectral Stochastic Neural Operator}

\hypersetup{
    colorlinks=true,      
    linkcolor=green,      
    urlcolor=blue,        
    citecolor=blue        
}

\begin{document}
\maketitle

\begin{abstract}
    Building surrogate models with uncertainty quantification capabilities is essential for many engineering applications where randomness—such as variability in material properties, boundary conditions, and initial conditions—is unavoidable. Polynomial Chaos Expansion (PCE) is widely recognized as a to-go method for constructing stochastic solutions in both intrusive and non-intrusive ways. Its application becomes challenging, however, with complex or high-dimensional processes, as achieving accuracy requires higher-order polynomials, which can increase computational demands and or the risk of overfitting. Furthermore, PCE requires specialized treatments to manage random variables that are not independent, and these treatments may be problem-dependent or may fail with increasing complexity. In this work, we aim to adopt the same formalism as the spectral expansion used in PCE; however, we replace the classical polynomial basis functions with neural network (NN) basis functions to leverage their expressivity. To achieve this, we propose an algorithm that identifies NN-parameterized basis functions in a purely data-driven manner, without any prior assumptions about the joint distribution of the random variables involved, whether independent or dependent, or about their marginal distributions. The proposed algorithm identifies each NN-parameterized basis function sequentially, ensuring they are orthogonal with respect to the data distribution. The basis functions are constructed directly on the joint stochastic variables without requiring a tensor product structure or assuming independence of the random variables. This approach may offer greater flexibility for complex stochastic models, while simplifying implementation compared to the tensor product structures typically used in PCE to handle random vectors. This is particularly advantageous given the current state of open-source packages, where building and training neural networks can be done with just a few lines of code and extensive community support. We demonstrate the effectiveness of the proposed scheme through several numerical examples of varying complexity and provide comparisons with classical PCE.
\end{abstract}


\section{Introduction}
\label{sec:intro}
Uncertainty quantification (UQ) is a cornerstone of scientific computing, encompassing the identification, propagation, and management of uncertainties throughout the computational modeling process. These uncertainties can arise from various sources, including observational data, modeling assumptions, numerical approximations, and stochastic variability in system inputs. Propagating uncertainty from the stochasticity in input information to the final model predictions is a critical subset of UQ, with significant implications in engineering disciplines. For instance, UQ enables the estimation of confidence intervals for predictions, which is crucial for informed decision-making under uncertainty. It also guides engineers and modelers in designing new experiments or refining models to reduce uncertainties, thereby improving reliability and robustness. These capabilities are vital in high-consequence engineering applications, such as aerospace, nuclear energy, and structural safety, where even small errors in prediction can lead to catastrophic outcomes.

UQ methods can be broadly categorized into two approaches of \textit{intrusive} and \textit{non-intrusive}. Intrusive methods require direct modifications to the continuous form of governing equations or their discrete approximations to incorporate uncertainty directly into the computational model. On the other hand, non-intrusive methods do not modify the computational model itself. Instead, they treat the model as a black-box, using techniques like sampling to explore the input space and analyze the resulting outputs. 
Since intrusive methods require direct access to the governing equations and involve modifying the computational model accordingly, they present significant challenges for problems where the governing equations are unavailable or the computational model is highly complex (e.g., nonlinear finite element methods). 

As a result, there is considerable interest in non-intrusive methods, which only require sampling the real physical system, computational model, or legacy codes without making any modifications to them. This flexibility makes non-intrusive methods particularly appealing for practical applications in a \text{plug-and-play} manner.
Non-intrusive methods may rely on sampling techniques to conduct experiments on the physical system and directly quantify the output prediction uncertainty; or alternatively they may use the samples to build a stochastic surrogate model of the underlying system. The latter approach is often more appealing, as it not only quantifies uncertainty but also provides a differentiable replica of the physical system. This surrogate model can then be leveraged for sensitivity analysis, reliability assessments, and design optimization.

Among the various hypothesis classes of functions used to build stochastic surrogate models, Polynomial Chaos Expansions (PCE) have gained significant attention in engineering applications \cite{hosder2006non,lüthen2021sparse,novak2024physics,sharma2024physics,giovanis2024polynomial}. This is due to their robustness in uncertainty quantification and strong mathematical rigor due to \textbf{spectral expansion} formalism \cite{ghanem2003stochastic,mercer1909xvi}, offering a systematic framework for representing stochastic processes and accurately propagating uncertainties in complex systems. However, certain aspects of their mathematical construction can impose limitations. First, PCEs heavily rely on a tensor product structure to simplify statistical moment calculations over a random vector by assuming its marginals are statistically independent. Consequently, when modeling (linearly or nonlinearly) correlated random variables, special treatments, such as transformations or dependency modeling through copulas, are required to account for these correlations effectively. These treatments are problem-specific and may not be effective in general setups \citep{sklar1959fonctions,joe2014dependence,emile1960distributions}. Secondly, PCE requires knowledge of the marginal distributions of the random variables to construct appropriate mutually orthogonal basis functions. This reliance on predefined distributions can limit the PCE applicability when the true distributions are unknown or when the random variables exhibit complex or non-standard distributions. Notably, generalized Polynomial Chaos Expansions (gPCE) \cite{xiu2002wiener} and arbitrary PCE (aPCE) \cite{oladyshkin2012data} can handle arbitrary marginal distributions, provided the random variables are independent. However, this raises concerns regarding the first limitation mentioned: the need to assume independence of random variables to approximate multivariate basis functions as tensor products of univariate ones, where both the independence assumption and the tensor product structure may limit modeling flexibility.

This work introduces Neural Chaos to address these limitations, allowing the retention of all the desirable properties of the spectral expansion formalism while constructing multivariate basis functions in a \textbf{data-driven} manner. These basis functions are designed to operate on the joint distribution of random variables without requiring tensor product decomposition, prior knowledge of marginal distributions, or the independence assumption of the input random variables. Moreover, we establish the connection between the classical spectral expansion, in particular PCEs, for describing stochastic processes with modern operator learning frameworks, in particular Deep Operator Network (DeepONet) \citep{lu2021learning,chen1995universal}.

The remainder of this paper is organized as follows. First, we present the spectral stochastic formulation, one of the main building blocks of Neural Chaos, and discuss how it connects to operator learning methods. In \cref{sec:form}, we provide the details of the Neural Chaos formulation and introduce two algorithms for learning the NN-parametrized basis functions utilized in Neural Chaos. In \cref{sec:exam}, we validate and demonstrate the proposed algorithms using five numerical examples. Finally, we conclude in \cref{sec:conclusion} by summarizing the major findings of this work.


\subsection{Spectral Stochastic Methods}
\label{sec:OL}


Consider the general stochastic partial differential equation given by
\begin{equation}
    \begin{aligned}
        & \mathcal{G}(\bm{x},t,\X(\omega); u(\bm{x},t,\X(\omega))) = f(\bm{x},t,\X(\omega)), & \forall \bm{x}\in \mathcal{D}_x, t \in \mathcal{T}, \omega \in \Omega
        \label{eqn:PDE_UQ}
    \end{aligned}
\end{equation}
where $\mathcal{T}\subset \mathbb{R}$, $\mathcal{D}_x\subset \mathbb{R}^3$, 
$\mathcal{G}$ is a differential operator, 
$u(\cdot)$ is the response/solution of the system, $f$ is an external force/source term, and $\X(\omega)\in \mathbb{R}^d$ is a $d$-dimensional random vector having sample space $\Omega$.

The classical spectral stochastic approach to solve this system, first introduced by Ghanem and Spanos~\cite{ghanem2003stochastic}, expands the solution $u(\bm{x},t,\X(\omega))$ as a stochastic process using spectral methods according to the Wiener-Askey polynomial chaos expansion (PCE)~\cite{xiu2002wiener} as
\begin{equation}
  u(\boldsymbol{x},t, \boldsymbol{\xi})
  =
  \sum_{p=0}^{\infty} \Phi_p(\boldsymbol{x},t) \Psi_p(\boldsymbol{\xi})
  \approx
  u^{(P)}(\boldsymbol{x},t, \boldsymbol{\xi})
  \triangleq
  \sum_{p=0}^P \Phi_p(\boldsymbol{x},t) \Psi_p(\boldsymbol{\xi}),
  \label{eq:spec-expan}
\end{equation}
where the random variables $\X(\omega)$ are transformed to standardized space according to $\boldsymbol{\xi}=T(\X)$ having probability density $p_{\boldsymbol{\xi}}(\boldsymbol{\xi})$, $\Psi_p(\boldsymbol{\xi})$ are orthogonal polynomials defined according to the Askey scheme~\cite{askey1985some}, and $\Phi_p(\boldsymbol{x},t)$ are deterministic spatio-temporal expansion coefficients that must be determined. Note the random variables, $\boldsymbol{\xi}$, are typically considered to be \textit{independent} such that the multi-variate orthogonal polynomials $\Psi_p(\boldsymbol{\xi})$ are expressed as a tensor product of \textit{univariate} orthogonal polynomials $\Psi_p(\boldsymbol{\xi})=\prod_{m=1}^d \psi_{pm}(\xi_m)$ and the joint distribution can be expressed as the product of the marginal distributions $p_{\boldsymbol{\xi}}(\boldsymbol{\xi}) = \prod_{m=1}^d p_{\xi_m}(\xi_m)$. Note also that the expansion in Eq.~\eqref{eq:spec-expan} is truncated to have a total number of terms given by
\begin{equation}
    P = \begin{pmatrix}
    d + n\\
    n
    \end{pmatrix}
    = \dfrac{(d+n)!}{d!n!},
    \label{eqn:num_terms}
\end{equation}
where $d$ is the dimension of the random variable $\boldsymbol{\xi}$ and $n$ is the highest order of the polynomials $\{\psi_p\}$. 

The expansion coefficients $\Phi_p(\boldsymbol{x},t)$ can be determined by substituting the expansion in \cref{eq:spec-expan} into the differential equation in \cref{eqn:PDE_UQ} and performing a Galerkin projection onto the orthogonal polynomial basis functions $\{g_p\}$ as
\begin{equation}
    \left\langle \mathcal{G}\left(\bm{x},t,\X(\omega); \sum_{p=0}^P \Phi_p(\boldsymbol{x},t) \Psi_p(\boldsymbol{\xi})\right), \Psi_k \right\rangle = \langle f, \Psi_k \rangle, \quad k = 0,1,\dots, P.
    \label{eqn:galerkin}
\end{equation}
This defines a coupled set of $(P+1)$ equations that can be solved for $\Phi_p(\boldsymbol{x},t)$ using classical numerical solvers, e.g., finite elements, by discretizing over the space $\bm{x}$ and time $t$. This class of spectral stochastic methods has been explored in great detail in the literature~\cite{ghanem2003stochastic,babuska2004galerkin,xiu2009fast, xiu2010numerical} and has become a benchmark in solving stochastic systems~\cite{smith2024uncertainty}. 

Alternatively, when $u$ is a low-dimensional or scalar-valued quantity of interest (QoI) -- e.g., when it is a scalar performance metric/value extracted from the full spatio-temporal solution $u(\bm{x},t)$ -- and a set of $N$ sample QoIs are available the coefficients $\Phi_p$ can be determined through regression. That is, given a set of $N$ samples of $\boldsymbol{\xi}_i, i=1,\dots, N$ and corresponding QoIs $\boldsymbol{u}_i$, the coefficients $\Phi_p$ can be determined by solving the following ordinary least squares problem~\cite{berveiller2006stochastic} 
\begin{equation}
    \begin{aligned}
        & \min_{\boldsymbol{\Phi}} \sum_{i=1}^{N} \left[\boldsymbol{u}_i - \boldsymbol{u}^{(P)}(\boldsymbol{\xi}_i) \right]^2=  \min_{\boldsymbol{F}} \lVert \boldsymbol{U}  - \boldsymbol{\Phi}^\intercal \boldsymbol{\Psi} \rVert^2,
    \end{aligned}
    \label{Eq. PC2 determ definition}
\end{equation}
where $\boldsymbol{\Phi} = \{\boldsymbol{\Phi}_p\}$, $\boldsymbol{\Psi} = \{\Psi_{pi}\}$, $\boldsymbol{U}=\{\boldsymbol{u}_i\}$, and $\Psi_{pi}$ is the $p$-th basis function evaluated at the $i$-th sample.

This regression-based approach has become increasingly popular in recent years for solving engineering problems in uncertainty quantification. This is because the regression approach is non-intrusive. That it, it does not require the development of custom numerical solvers and allows for surrogate model development in the form of \cref{eq:spec-expan} from a finite number $N$ of samples obtained by solving \cref{eqn:PDE_UQ} deterministically at fixed values of $\boldsymbol{\xi}_i, i=1,\dots,N$, often using commercially available numerical solvers. This has led to the rapid growth of advanced regression method using e.g., sparse regression and hyperbolic truncation schemes~\cite{blatman2011adaptive}, as well as the recognition that PCE can be posed in a purely data-driven sense as a machine learning regression problem~\cite{torre2019data}. This is the view that we take in this work. 

\subsection{The Importance of Orthogonality}
\label{sec:ortho}

The essential element of spectral stochastic methods is that the basis functions $\Psi_p(\boldsymbol{\xi})$ must be orthogonal with respect to the probability distribution $p_{\boldsymbol{\xi}}(\boldsymbol{\xi})$. That is,
\begin{equation}
    \langle \Psi_p, \Psi_k \rangle \triangleq \int_\Omega \Psi_p(\boldsymbol{\xi}) \Psi_k(\boldsymbol{\xi}) p_{\boldsymbol{\xi}}(\boldsymbol{\xi}) d\boldsymbol{\xi} = \| \Psi_p\|^2 \delta_{pk}.
    \label{eqn:orthogonality}
\end{equation}
Employing orthonormal basis functions, we have $\| \Psi_p \|^2 = 1$, where $\|\cdot \|$ denotes the norm associated with the specified inner product. This enables the Galerkin projection in \cref{eqn:galerkin}, ensuring that the error is orthogonal to the function space spanned by $\Psi_p(\boldsymbol{\xi})$. Moreover, this orthogonality yields convenient properties of the PCE in Eq.~\eqref{eq:spec-expan} for uncertainty quantification. Specifically, moments of the solution $u(\boldsymbol{x},t,\boldsymbol{\xi})$ can be estimated directly from the coefficients $\Phi_p(\boldsymbol{x}, t)$ where the mean function is given by $\Phi_0(\boldsymbol{x},t)$ and the second moment is given by
\begin{align}
    \label{Eq:SecondRawMoment}
    \mathbb{E}_{\boldsymbol{\xi}}[{{u^2}(\boldsymbol{x},t, \boldsymbol{\xi})}]
    &=
    \int
    {\left[ {\sum\limits_{p=0}^P
        {{\Phi_p(\bm{x},t)}{\Psi_{p}}\left( \boldsymbol{\xi}  \right)} } \right]} ^2
        p_{\boldsymbol{\xi}}(\boldsymbol{\xi})   \;
        \mathrm{d}  \boldsymbol{\xi} =
    \sum\limits_{p=0}^P
    \sum\limits_{k=0}^P
    \Phi_p(\bm{x},t)
    \Phi_k(\bm{x},t)
    \int
    {\Psi_{p}}\left( \boldsymbol{\xi}  \right)
    {\Psi_{k}}\left( \boldsymbol{\xi}  \right)
    p_{\boldsymbol{\xi}}(\boldsymbol{\xi})   \;
        \mathrm{d}  \boldsymbol{\xi}
    \\   \nonumber
    & =
    \sum\limits_{p=0}^P {\Phi_p^2(\bm{x},t)} {\int {\Psi_{p}^2}\left( \boldsymbol{\xi}  \right) }p_{\boldsymbol{\xi}}(\boldsymbol{\xi})
     \;
    \mathrm{d}  \boldsymbol{\xi}
    = \sum\limits_{p=0}^P {\Phi_p^2(\bm{x},t)} \|g_p\|^2
    = \sum\limits_{p=0}^P {\Phi_p^2(\bm{x},t)}.
\end{align}
Additionally, higher-order moments~\cite{novak2022distribution}, fractional moments~\cite{novak2024fractional}, and sensitivity indices~\cite{sudret2008global,novak2022distribution} can be derived directly from the coefficients. This makes the PCE very powerful for uncertainty quantification.

\subsection{Challenges of Spectral Stochastic Methods}
In spectral stochastic methods, orthogonality is guaranteed by selecting the basis functions $\Psi_p(\boldsymbol{\xi})$ \textit{according to the Askey scheme of orthogonal polynomials}. This is the primary strength of spectral stochastic methods, but also leads to some drawbacks and associated challenges.  Polynomial functions, while inherently smooth and differentiable, are not optimally expressive and may fail to produce representations that generalize well. The consequence of this is that PCE methods may require high-order polynomial terms to fit non-linear functions, but these high-order polynomials lead to instabilities (e.g., Runge's phenomenon, bad conditioned design matrix, etc.) that result in over-fitting \cite{kontolati2023influence}. Hence, there is a trade-off between generalization (requiring low-order polynomials) and expressivity (requiring high-order polynomials).  

%
A related issue is the polynomial growth of the number of terms in a PC expansion; $P$ in \cref{eqn:num_terms} scales as $O(n^d)$ for fixed dimensions $d$ and as $O(d^n)$ for a fixed polynomial order $n$.
This results in an explosion of the number of coefficients that need to be solved for. For stochastic Galerkin methods, the consequence is a very large system of equations that must be solved using numerical methods -- effectively an explosion in the number of degrees of freedom in the numerical model. As a result, spectral stochastic methods typical employ only low-order polynomials and struggle with strongly non-linear problems. For regression methods, again the number of coefficients grows extremely large, which either necessitates very large training data sets (that come at high computational expense) or requires sophisticated sparse regression such as methods like Least Angle Regression~\cite{blatman2011adaptive}, Bayesian compressive sensing~\cite{hampton2015compressive, tsilifis2019compressive}, partial least squares~\cite{papaioannou2019pls}, and many other methods have become standard practice. A comprehensive review can be found in~\cite{lüthen2021sparse}. Even with these methods, regression-based PCE may require the determination of hundreds to thousands of coefficients while balancing the trade-off between generalization and over-fitting~\cite{kontolati2023influence}. 

The PCE is also generally constrained by the class of polynomials that are available for the analytical construction of the polynomial chaos expansion in Eq.~\eqref{eq:spec-expan}. These classes constrain the types on uncertainties that can be included to those with common and well-known distributions -- e.g. Gaussian, uniform -- and independent marginal distributions resulting in tensor product basis construction. For distributions that do not follow the Askey scheme or involve dependent variables, analytical spectral expansions are scarce and typically valid only under special conditions~\cite{rahman2018polynomial}. Recent advances do allow for the construction of polynomial chaos expansions for random variables with aribitrary distributions, termed arbitrary PCE~\cite{soize2004physical,oladyshkin2012data}. 
These approaches build orthogonal polynomial basis functions with coefficients that can be solved directly from statistical moments estimated from the given dataset. These methods have been shown to perform well in a data-driven setting where the form of the distribution $p_{\boldsymbol{\xi}}(\boldsymbol{\xi})$ is not known analytically~\cite{soize2004physical,oladyshkin2012data}. 

A final, and important drawback of regression-based PCE is usually used to define a surrogate model for a low-dimensional (often scalar-valued) quantity of interest derived from the solution $u$. Regression-based PCE methods generally cannot learn the coefficient functions $\Phi_p(\bm{x},t)$ needed to express the full spatio-temporal solution $u(\bm{x},t)$. Rather, they generally learn scalar coefficients $\Phi_p$ that, as previously mentioned, can be used to build surrogate models for scalar response quantities of interest derived from the full solution $u(\bm{x},t)$. This severely limits the PCE as a modeling framework for operator learning, which is of interest in this work and is briefly discussed next.

\subsection{PCE for Operator Learning}
A few recent works have begun to explore PCE in a context that can be considered operator learning -- although these methods are not formally stated or presented as operator learning methods. The formal problem statement for operator learning will be presented in the following section. For the present purposes, we will state the (deterministic) operator learning problem informally as any method that attempts to learn the operator $\mathcal{G}(\cdot)$ through a functional $\mathcal{G}_{\theta}(\cdot)$ with parameters $\theta$ 
such that $u(\bm{x},t, \bm{\xi})\approx \mathcal{G}_{\theta}(\boldsymbol{x}, t, \boldsymbol{\xi}, f)$.

To solve this operator learning problem, Kontolati et al.~\cite{kontolati2022manifold,kontolati2022survey} developed the manifold PCE (mPCE) method that can be view as an operator learning method. The mPCE method is a regression-based approach in which realizations of, for example, $f(\bm{x},t,\boldsymbol{\xi}_i), i=1,\dots,N$ and corresponding solutions $u(\bm{x},t,\boldsymbol{\xi}_i)$ are projected onto a low-dimensional manifold. A PCE model is then fit between the projected low-dimensional input and the corresponding low-dimensional solution. For any new input, $f(\bm{x},t,\boldsymbol{\xi}^*)$ is projected onto the manifold, the low-dimensional solution is predicted with the PCE model, and the low-dimensional solution is then lifted back to reconstruct the approximate output function $\tilde{u}(\bm{x},t,\boldsymbol{\xi}^*)$ by inverting the solution projection operator. This procedure is generalizable and amenable to various forms of dimension reduction ranging from linear methods such as principal component analysis to nonlinear manifold learning methods such as diffusion maps and autoencoders. Kontolati et al.~\cite{kontolati2022survey} performed a comprehensive study of the mPCE method for different classes of unsupervised dimension reduction methods and Kontolati et al.~\cite{kontolati2023influence} then compared this approach with neural network-based operator learning using the deep operator network (DeepONet).

Very recently, Novak et al.~\cite{novak2024physics} and Sharma et al.~\cite{sharma2024physics} introduced the physics-constrained polynomial chaos expansion (PC$^2$). This approach, which introduces physics-informed learning in the context of PCE, defines orthogonal polynomial basis functions over the spatial and temporal variables such that the coefficient functions $\Phi_i(\bm{x},t)$ are, themselves, expressed through a polynomial chaos expansion. Doing so allows the least squares solution to be constrained (by physical constraints, i.e., PDEs, or arbitrary constraints such as inequalities), which is the objective of these works. But, by formulating the coefficient functions in this way, the authors also extend the classical regression-based PCE to an operator learning setting where the full spatial-temporal solution can be predicted.

\section{Formulation}
\label{sec:form}

In this section, we break from the classical spectral stochastic methods by formulating the general operator learning problem for stochastic problems and demonstrating how this operator learning problem can be solved using a new class of \textit{spectral stochastic neural operators}. Rather than employing orthogonal polynomials, these spectral stochastic neural operators learn a dictionary of orthogonal functions, \textit{constructed as neural networks}, that serve as a compact basis for the spectral expansion. Importantly, these learned basis functions maintain the orthogonality properties introduced in Section~\ref{sec:ortho}, which endows them with the beneficial properties of existing spectral stochastic methods (e.g., easier moment estimation). They can be learned for arbitrary distributions and are not constrained by a tensor product structure meaning they can uphold orthogonality to joint distributions with arbitrary dependence structure. However, if a tensor product serves as a strong inductive bias or is computationally justified, one can still use the proposed scheme while parametrizing a subset or all basis functions in a tensor product fashion \cite{cho2022separable}.

\subsection{Operator Learning Between Random Function Spaces}


Let's begin by framing the stochastic operator learning problem in the classic setting of Chen and Chen~\citep{chen1995universal} adopted most notably in the Deep Operator Network (DeepONet)~\citep{lu2021learning} by considering that we aim to approximate the solution operator for a given input at a specific point $\bm{y}=\{\bm{x}, t, \boldsymbol{\xi}\}$ conditioned on the input function $g$ (e.g., operator source term) where $\bm{y}$ concatenates the spatial and temporal locations of the prediction with the random vector $\boldsymbol{\xi}$ and $\boldsymbol{g}$ concatenates the parts of the PDE operator that we have access to as input function information. In this setting, the neural operator can be formulated as
\begin{equation}
\begin{aligned}
    \mathcal{G}(g)(\bm{y})
    &\approx
    \sum_{s=1}^S
    \underbrace{
    \sum_{i=1}^n \bar{\beta}_{i}^s \sigma\left(\sum_{j=1}^m \hat{\beta}_{ij}^s g_j + \tilde{\beta}_i^s\right)
    }_{\bar{\Phi}_s(\boldsymbol{g}; \beta_s)}
    \underbrace{
    \sigma(\bar{\theta}_s \bm{y} + \tilde{\theta}_s)
    }_{\bar{\Psi}_s(\bm{y}; \theta_s)}\\
    &=
    \sum_{s=1}^S
    \bar{\Phi}_s(\boldsymbol{g}; \beta_s)\bar{\Psi}_s(\bm{y}; \theta_s) ,
\end{aligned}
\label{eqn:Stochastic_DeepONet}
\end{equation}
where $\sigma$ is a nonlinear activation function (e.g., tanh), $\bar{\beta}_i^s$, $\hat{\beta}_i^s$, $\tilde{\beta}_i^s$, $\bar{\theta}_i^s$, $\tilde{\theta}_i^s$, 
are tunable parameters whose their collections are denoted by $\beta_s =\{ \bar{\beta}_i^s, \hat{\beta}_i^s, \tilde{\beta}_i^s\}$ and $\theta_s =\{ \bar{\theta}_i^s, \tilde{\theta}_i^s\}$, respectively. 
According to \cref{eqn:Stochastic_DeepONet}, $\bar{\Phi}_s(\boldsymbol{g}; \beta_s)$ defines the branch network and $\bar{\Psi}_s(\bm{y}; \theta_s)$ is a \textit{global} trunk network that spans the output over the combined spatio-temporal and \textit{stochastic} spaces. We refer to this naive implementation as the Stochastic DeepONet.

The Stochastic DeepONet in \cref{eqn:Stochastic_DeepONet} is a straightforward extension of the DeepONet with predictive ability for stochastic problems at arbitrary $\bm{x}$, $t$, and $\boldsymbol{\xi}$. However, it lacks the properties of the spectral expansion that, as argued above, are desirable for stochastic problems. In particular, the basis functions $\bar{\Psi}_s(\bm{y}; \theta_s)$ should be orthogonal with respect to $p_{\boldsymbol{\xi}}(\boldsymbol{\xi})$. According to the formulation in \cref{eqn:Stochastic_DeepONet}, this can only be achieved by integrating the conditions in \cref{eqn:orthogonality} into the loss function. Although this is perhaps possible, it may require a separable tensor product architecture~\cite{cho2022separable,mandl2024separable, yu2024separable} (which assumes independent random variables). This approach can pose significant computational challenges for optimization, particularly when evaluating and enforcing all orthogonality integrals during each training iteration. This cost and the associated challenges may compound rapidly as the number of dimensions, $d$, increases.

The function $\bar{\Psi}_s(\boldsymbol{y})$ can be universally approximated as $\bar{\Psi}_s(\boldsymbol{y}) = \sum_{q=1}^{\infty} \hat{\Phi}_{sq}(\boldsymbol{x}, t) \hat{\Psi}_{sq}(\boldsymbol{\xi})$, which follows the spectral stochastic expansion in Eq.~\eqref{eq:spec-expan}. Substituting this into \cref{eqn:Stochastic_DeepONet}, we obtain:
\begin{equation}
\begin{aligned}
    \mathcal{G}(g)(\bm{y})
    &\approx
    \sum_{s=1}^S
    \bar{\Phi}_s(\boldsymbol{g})
    \sum_{q=1}^{\infty} \hat{\Phi}_{sq}(\boldsymbol{x}, t) \hat{\Psi}_{sq}(\boldsymbol{\xi})\\
    &=
    \sum_{q=1}^{\infty}
    \sum_{s=1}^{S}
    \bar{\Phi}_s(\boldsymbol{g})
     \hat{\Phi}_{sq}(\boldsymbol{x}, t) \hat{\Psi}_{sq}(\boldsymbol{\xi})\\
    &=
    \sum_{p=1}^{\infty}
    \underbrace{
    \bar{\Phi}_p(\boldsymbol{g})
     \hat{\Phi}_{p}(\boldsymbol{x}, t)}_{{\Phi}_p(\boldsymbol{g}, \boldsymbol{x}, t)}
     {\Psi}_{p}(\boldsymbol{\xi})\\
    &=
    \sum_{p=1}^{\infty}
    \Phi_p(\boldsymbol{g}, \boldsymbol{x}, t)
     {\Psi}_{p}(\boldsymbol{\xi})\\
    &\approx
    \sum_{p=1}^{P}
    \Phi_p(\boldsymbol{g}, \boldsymbol{x}, t)
     {\Psi}_{p}(\boldsymbol{\xi}). 
\end{aligned}
\label{eqn:Spectral_Stochastic_NO}
\end{equation}
The proposed formulation in \cref{eqn:Spectral_Stochastic_NO}, which we refer to as the \textit{Spectral 
Stochastic Neural Operator} (or Neural Chaos in homage to the aforementioned polynomial chaos), differs subtly in construction from the Stochastic DeepONet but now clearly takes the desirable form of a spectral expansion (see Eq.~\eqref{eq:spec-expan}) if we can ensure that $\bar{\Psi}_p(\boldsymbol{\xi}; \beta_p)$ satisfy the orthogonality conditions in \cref{eqn:orthogonality}. This is the primary challenge addressed herein.

\remark{
In this work, we lump all sources of uncertainty, such as errors regarding boundary conditions, initial conditions, domain boundary, PDE operator, measurements, etc., into random variables $\boldsymbol{\xi}$, while in some applications, it might be important to identify and separate different sources of uncertainty. We leave such an extension as a promising direction for future studies.
}

\subsection{Building Spectral Stochastic Neural Network Basis Functions}
A naive way to build the Spectral Stochastic Neural Operator would, again, be to enforce orthogonality through the loss function. But this may run into the same huge computational challenges discussed above.
%
Instead, we propose an iterative algorithm that learns the stochastic basis functions $g_p(\boldsymbol{\xi}; \beta_p)$ and the associated spatio-temporal coefficients $f_p(\bm{x}; \theta_p)$ together as a sequence of neural networks that are (approximately) orthogonal by construction. This proposed architecture is presented in Figure \ref{fig:nc-method}.

\begin{figure}[!ht]
  \centering
    \includegraphics[width=.6\textwidth]{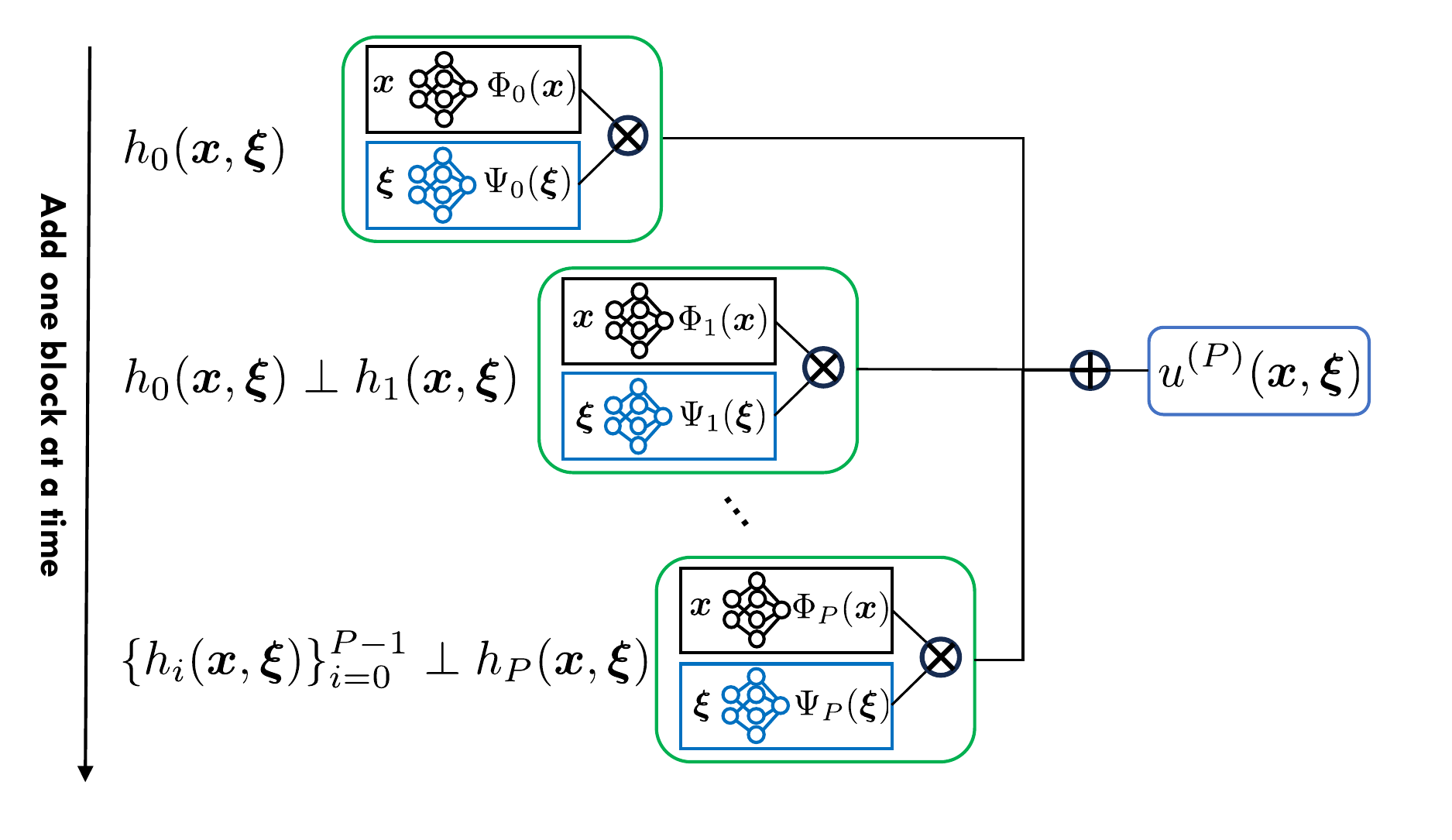}
  \caption{
  Neural Chaos architecture, which mirrors the structure of a spectral expansion. Spectral terms are added one by one such that each new term is orthogonal to the previous ones. Each spectral term has a multiplicative structure with $h_i(\bm{x},\boldsymbol{\xi}) = \Phi_i(\bm{x})\Psi_i(\boldsymbol{\xi})$ comprising deterministic and stochastic basis functions, each parameterized by a neural network.
  }
  \label{fig:nc-method}
\end{figure}

To construct orthogonal stochastic basis functions using neural networks (i.e., without leveraging a class of predetermined orthogonal functions such as orthogonal polynomials), we follow the residual-based approach initially proposed in \cite{bahmani2024resolution} in the context of neural operator learning, which we use here for stochastic processes.

\begin{lemma} The truncated residual $r^{(p)}(\boldsymbol{x}, \boldsymbol{\xi}) = u(\boldsymbol{x}, \boldsymbol{\xi}) - u^{(p)}(\boldsymbol{x}, \boldsymbol{\xi})$ of the spectral expansion series, as described by Equation \ref{eq:spec-expan}, is itself a random process that is orthogonal to the series basis functions $\Psi_i$ if these bases are orthonormal in $L^2(p(\boldsymbol{\xi}))$, i.e., mutually orthogonal and unit norm.
\end{lemma}

\begin{proof}
We simply need to show that the inner product of the residual with any stochastic basis $\Psi_i(\boldsymbol{\xi})$ is zero:
\begin{align}
    \langle r^{(p)}, \Psi_i \rangle
    &=
    \mathbb{E}_{\boldsymbol{\xi}}
    \left[
    \Psi_i(\boldsymbol{\xi})
    \left(
         u(\boldsymbol{x}, \boldsymbol{\xi}) - u^{(p)}(\boldsymbol{x}, \boldsymbol{\xi})
    \right)
    \right]\\
    &=
    \mathbb{E}_{\boldsymbol{\xi}}
    \left[
    \Psi_i(\boldsymbol{\xi})
    \left(
         u(\boldsymbol{x}, \boldsymbol{\xi})
         - 
         \sum_{j=0}^p \Phi_j(\boldsymbol{x}) \psi_j(\boldsymbol{\xi})
    \right)
    \right]\\
    &=
    \mathbb{E}_{\boldsymbol{\xi}}
    \left[
        \Psi_i(\boldsymbol{\xi})
        u(\boldsymbol{x}, \boldsymbol{\xi})
    \right]
    -
    \sum_{j=0}^p \Phi_j(\boldsymbol{x}) \mathbb{E}_{\boldsymbol{\xi}}
    \left[
    \Psi_i(\boldsymbol{\xi})
    \Psi_j(\boldsymbol{\xi})
    \right]\\
    &=
    \Phi_i(\boldsymbol{x})
    -
    \Phi_i(\boldsymbol{x}) 
    \|\Psi_i(\boldsymbol{\xi})\|^2_{L^2(p(\boldsymbol{\xi}))} = 0.
\end{align}
\end{proof}

\begin{lemma}
\label{lemma:res-decay}
For any multiplicative composition of a stochastic function $\bar{\Psi}(\boldsymbol{\xi}) \ne 0$ and a deterministic function $\bar{\Phi}(\boldsymbol{x}) \ne 0$ that reduces the truncated residual, the stochastic function cannot belong to the subspace spanned by the stochastic basis functions used in the truncation series, i.e.,

\begin{equation}
    \|
    r^{(p)}(\boldsymbol{x}, \boldsymbol{\xi})
    -
    \bar{\Phi}(\boldsymbol{x})
    \bar{\Psi}(\boldsymbol{\xi})
    \|^2
    <
    \|r^{(p)}(\boldsymbol{x}, \boldsymbol{\xi})
    \|^2
    \implies
    \bar{\Psi}(\boldsymbol{\xi})
    \notin
    \text{span}
    \{
        \Psi_j(\boldsymbol{\xi})
    \}_{j=1}^p.
\end{equation}
\end{lemma}

\begin{proof}
    By contradiction, we need to show that if $\bar{\Psi}(\boldsymbol{\xi})$ belongs to the subspace of basis functions, then there is no way to reduce the residual. If $\bar{\Psi}(\boldsymbol{\xi})$ belongs to the subspace, it can be spanned linearly by the basis functions, i.e., $\bar{\Psi}(\boldsymbol{\xi}) = \sum_{j=1}^p \alpha_j \Psi_j(\boldsymbol{\xi})$:
    \begin{align}
    \|
    r^{(p)}(\boldsymbol{x}, \boldsymbol{\xi})
    -
    \bar{\Phi}(\boldsymbol{x})
    \bar{\Psi}(\boldsymbol{\xi})
    \|^2
    &=
    \mathbb{E}_{\boldsymbol{\xi}}
    \left[
        \left(
        r^{(p)}(\boldsymbol{x}, \boldsymbol{\xi})
        -
        \bar{\Phi}(\boldsymbol{x})
        \bar{\Psi}(\boldsymbol{\xi})
        \right)^2
    \right]\\
    &=
    \| r^{(p)} \|^2
    +
    \underbrace{
    \bar{\Phi}^2(\boldsymbol{x})
    \| \bar{\Psi}(\boldsymbol{\xi}) \|^2
    }_{\ge 0}
    -
    2 \bar{\Phi}(\boldsymbol{x})
    \underbrace{
    \mathbb{E}_{\boldsymbol{\xi}}
    \left[
        r^{(p)}(\boldsymbol{x}, \boldsymbol{\xi} 
        ) \bar{\Psi}(\boldsymbol{\xi}))
    \right]
    }_{=0 (\text{ orthogonality})}
    \\
    &\ge
    \| r^{(p)} \|^2.
    \end{align}    
\end{proof}

\begin{algorithm}
\caption{\textbf{Continuous} Spectral Stochastic Dictionary Learning}
\label{algo:SSDL-cont}
\begin{algorithmic}[1] 

\State \textbf{Input:} 
\State $
\mathcal{D}_{u}=\left\{
\mathcal{D}_{u}^{(i)} \right\}_{i=1}^{N}, \mathcal{D}_{u}^{(i)} =
\left(
\boldsymbol{\xi}^{(i)},
\left\{
\boldsymbol{x}^{(i, j)}, u^{(i, j)}
\right\}_{j=1}^{M}
\right)$ \Comment{$N$ realizations of the stochastic process $u(\boldsymbol{x}, \boldsymbol{\xi})$}
\State $0 < \text{Tol.} \ll 1$ \Comment{A small tolerance of target accuracy}
\State \textbf{Output:}  $P, \Phi_0(\cdot; \boldsymbol{\theta}_0), \left\{ \Phi_p(\cdot; \boldsymbol{\theta}_p), \Psi_p(\cdot; \boldsymbol{\beta}_p) \right\}_{p=1}^{P}$ \Comment{A set of $2P+1$ parameterized basis functions}
%
\State $p \gets 0$
\State $\boldsymbol{\theta}_0 = {\text{argmin}}\
    \underset{(\boldsymbol{x}, \boldsymbol{\xi})\sim \mathcal{D}_{u}}{\mathbb{E}}\left\|
    u(\boldsymbol{x}, \boldsymbol{\xi}) -
    \Phi_0(\boldsymbol{x}; \boldsymbol{\theta}_0)
    \right\|_2^2
    $
\State $r(\boldsymbol{x}, \boldsymbol{\xi}) = u(\boldsymbol{x}, \boldsymbol{\xi}) - \Phi_0(\boldsymbol{x}; \boldsymbol{\theta}_0)$
\While{$\mathbb{E}_{\boldsymbol{x}, \boldsymbol{\xi}}[r^2(\boldsymbol{x}, \boldsymbol{\xi})] \ge \text{Tol.}$}
    \State $p \gets p + 1$
    \State Randomly initialize two new neural networks $\Phi_p(\boldsymbol{x};\boldsymbol{\theta}_p)$, $\Psi_p(\boldsymbol{\xi};\boldsymbol{\beta}_p)$
    \State $\boldsymbol{\theta}_p, \boldsymbol{\beta}_p = {\text{argmin}}\
    \underset{(\boldsymbol{x}, \boldsymbol{\xi})\sim \mathcal{D}_{u}}{\mathbb{E}}\left\|
    r(\boldsymbol{x}, \boldsymbol{\xi}) -
    \Phi_p(\boldsymbol{x}; \boldsymbol{\theta}_p)
    \Psi_p(\boldsymbol{\xi}; \boldsymbol{\beta}_p)
    \right\|_2^2
    $
    \State $r(\boldsymbol{x}, \boldsymbol{\xi}) 
    \gets
    r(\boldsymbol{x}, \boldsymbol{\xi})
    -
    \Phi_p(\boldsymbol{x}; \boldsymbol{\theta}_{p}) \Psi_p(\boldsymbol{\xi}; \boldsymbol{\beta}_p)$
\EndWhile
\end{algorithmic}
\end{algorithm}

Based on these two lemmas, we propose an algorithm that learns the $p$-th term in the spectral expansion by finding the unknown functions $\Phi_{p}(\boldsymbol{x}; \boldsymbol{\theta}_{p})$ and $\Psi_{p}(\boldsymbol{\xi}; \boldsymbol{\beta}_p)$ that best factorize the stochastic residual field $r^{(p-1)}(\boldsymbol{x}, \boldsymbol{\xi})$ in a multiplicative way. The neural network parameters $\boldsymbol{\theta}_p$ and $\boldsymbol{\beta}_p$ are optimized using gradient descent to minimize the expected error across the realizations of the random process. Based on Lemma~\ref{lemma:res-decay}, if the residual decays, we can conclude that the newly extended basis function introduces a new function that does not belong to the subspace spanned by the previous basis functions. We call this method ``\textbf{continuous}'' spectral stochastic dictionary learning, as outlined in Algorithm~\ref{algo:SSDL-cont}. Note that, similar to PCE-based methods, the zero-th term sets $\Psi_0(\boldsymbol{\xi}) = 1$, and accordingly, the deterministic basis function aims to capture the mean field across the realizations (see Line 6 in Algorithm~\ref{algo:SSDL-cont}).

Learning the basis functions in this way may introduce some challenges during optimization, particularly in Line 11, where an identifiability issue arises. Specifically, one can \textit{arbitrarily scale up} one of the basis functions and scale down the other by the same factor, while the loss function remains unchanged. Ideas involving alternating optimization and rescaling after each iteration, as proposed in \cite{bahmani2024resolution}, can be leveraged to mitigate these issues. However, along these lines, we propose a more robust algorithm as follows: By leveraging the bilinearity of each spectral term, if one of the functions is fixed, the other can be found analytically at the function level, regardless of the chosen parameterization. In other words, this is the optimal solution, independent of how the function has been parameterized via polynomials, neural networks, etc.
\begin{lemma}
For the fixed $\Psi_{p}(\boldsymbol{\xi})$, the optimal $\Phi_{p}(\boldsymbol{x})$ has a closed form solution as follows:
\begin{align}
    \underset{\Phi_p}{\text{argmin}} \ \mathbb{E}_{\boldsymbol{\xi}}
    [
        (\Phi_p(\boldsymbol{x}) \Psi_p(\boldsymbol{\xi})
        - r(\boldsymbol{x}, \boldsymbol{\xi}))^{2}
    ] &=
    \underset{\Phi_p}{\text{argmin}} \ 
    \Phi_p^2(\boldsymbol{x})\mathbb{E}_{\boldsymbol{\xi}}[\Psi_p^2(\boldsymbol{\xi})]
    -2\Phi_p(\boldsymbol{x})\mathbb{E}_{\boldsymbol{\xi}}[\Psi_p(\boldsymbol{\xi})r(\boldsymbol{x}, \boldsymbol{\xi})]
    +
    \mathbb{E}_{\boldsymbol{\xi}}
    [r^2(\boldsymbol{x}, \boldsymbol{\xi})]
    \\
    &=
    \underset{\Phi_p}{\text{argmin}} \ 
    \Phi_p^2(\boldsymbol{x})\mathbb{E}_{\boldsymbol{\xi}}[\Psi_p^2(\boldsymbol{\xi})]
    -2\Phi_p(\boldsymbol{x})\mathbb{E}_{\boldsymbol{\xi}}[\Psi_p(\boldsymbol{\xi})r(\boldsymbol{x}, \boldsymbol{\xi})]
    \\
    &=    
    \frac{\mathbb{E}_{\boldsymbol{\xi}}[\Psi_p(\boldsymbol{\xi})r(\boldsymbol{x}, \boldsymbol{\xi})]}{\mathbb{E}_{\boldsymbol{\xi}}[\Psi_p^2(\boldsymbol{\xi})]}.
\end{align}
\end{lemma}
\begin{lemma}
    For the fixed $\Phi_{p}(\boldsymbol{x})$, the optimal $\Psi_{p}(\boldsymbol{\xi})$ has a closed form solution as follows:
\begin{align}
    \underset{\Psi_p}{\text{argmin}} \ \mathbb{E}_{\boldsymbol{x}}
    [
        (\Phi_p(\boldsymbol{x}) \Psi_p(\boldsymbol{\xi})
        - r(\boldsymbol{x}, \boldsymbol{\xi}))^{2}
    ] &=
    \underset{\Psi_p}{\text{argmin}} \ 
    \Psi_p^2(\boldsymbol{\xi}) \mathbb{E}_{\boldsymbol{x}}[\Phi_p^2(\boldsymbol{x})]
    -2\Phi_p(\boldsymbol{x})\mathbb{E}_{\boldsymbol{x}}[\Psi_p(\boldsymbol{\xi})r(\boldsymbol{x}, \boldsymbol{\xi})]
    +
    \mathbb{E}_{\boldsymbol{x}}
    [r^2(\boldsymbol{x}, \boldsymbol{\xi})]
    \\
    &=
    \underset{\Psi_p}{\text{argmin}} \ 
    \Psi_p^2(\boldsymbol{\xi})\mathbb{E}_{\boldsymbol{x}}[\Phi_p^2(\boldsymbol{x})]
    -2\Psi_p(\boldsymbol{\xi})\mathbb{E}_{\boldsymbol{x}}[\Phi_p(\boldsymbol{x})r(\boldsymbol{x}, \boldsymbol{\xi})]
    \\
    &=    
    \frac{\mathbb{E}_{\boldsymbol{x}}[\Phi_p(\boldsymbol{x})r(\boldsymbol{x}, \boldsymbol{\xi})]}{\mathbb{E}_{\boldsymbol{x}}[\Phi_p^2(\boldsymbol{x})]}.
\end{align}
\end{lemma}
Based on these lemmas, we can first find the optimal basis vectors in a fully discrete manner, and then, in parallel, learn continuous versions of these bases that map the domain of the functions to the identified images via parameterization with neural networks. For example, \( \Psi_p(\boldsymbol{\xi}; \boldsymbol{\beta}_p): \boldsymbol{\xi} \to \Psi_p(\boldsymbol{\xi}) \). We call this method ``\textbf{discrete-continuous}'' spectral stochastic dictionary learning, as outlined in Algorithm~\ref{algo:SSDL-disc-cont}. The key difference between this algorithm and the previous one lies in Line 11, where we introduce the Closest Multiplicative Decomposition (CMD) iterations, as outlined in Algorithm~\ref{algo:MulDecomp}. This provides a practical use case for deriving the closed-form solution for stochastic and deterministic functions, with one of them assumed to be fixed. Note that at the end of each CMD step, we normalize the stochastic basis to have unit norm by appropriately scaling the stochastic and deterministic basis vectors. In this way, the magnitude of the deterministic basis reflects the importance of their corresponding stochastic basis in the overall spectral expansion for approximating the random process under study.

\begin{algorithm}
\caption{\textbf{Discrete-Continuous} Spectral Stochastic Dictionary Learning}
\label{algo:SSDL-disc-cont}
\begin{algorithmic}[1] 

\State \textbf{Input:} 
\State $
\mathcal{D}_{u}=\left\{
\mathcal{D}_{u}^{(i)} \right\}_{i=1}^{N}, \mathcal{D}_{u}^{(i)} = 
\left(
\boldsymbol{\xi}^{(i)}, \left\{\boldsymbol{x}^{(i, j)}, u^{(i, j)}
\right\}_{j=1}^{M}\right)$ \Comment{$N$ realizations of the stochastic process $u(\boldsymbol{x}, \boldsymbol{\xi})$}
\State $0 < \text{Tol.} \ll 1$ \Comment{A small tolerance of target accuracy}
\State \textbf{Output:}  $P, \Phi_0, \left\{ \Phi_p, \Psi_p \right\}_{p=1}^{P}$ \Comment{A set of $2P+1$ basis vectors}
%
\State $p \gets 0$
\State $ {\Phi_0}_j
    =
    \mathbb{E}_{i}
    [u^{(i, j})]
    $ \Comment{For all $1 \le j\le M$}
\State $r^{(i, j)} = u^{(i, j)} - {\Phi_0}_j$ \Comment{For all $1 \le j\le M$ and $1\le i \le N$}
\While{$\mathbb{E}_{i,j}[|| r^{(i, j)}||_2^2] \ge \text{Tol.}$}
    \State $
    \mathcal{D}_{r}=\left\{
    \mathcal{D}_{r}^{(i)} \right\}_{i=1}^{N}, \mathcal{D}_{r}^{(i)} = 
    \left(
    \boldsymbol{\xi}^{(i)}, \left\{\boldsymbol{x}^{(i, j)}, r^{(i, j)}
    \right\}_{j=1}^{M}\right)$
    \State $p \gets p + 1$
    \State $\Phi_p, \Psi_p \gets \text{CMD}(\mathcal{D}_r)$ \Comment{See Algorithm \ref{algo:MulDecomp}}
    \State $r^{(i,j)} 
    \gets
    r^{(i, j)}
    -
    {\Phi_p}_j {\Psi_p}_{i}$ \Comment{For all $1 \le j\le M$ and $1\le i \le N$}
\EndWhile
\end{algorithmic}
\end{algorithm}

%
%
%
%

\begin{algorithm}
\caption{Closest Multiplicative Decomposition}
\label{algo:MulDecomp}
\begin{algorithmic}[1] 

\State \textbf{Input:} 
\State $
\mathcal{D}_{r}=\left\{
\mathcal{D}_{r}^{(i)} \right\}_{i=1}^{N}, \mathcal{D}_{r}^{(i)} = 
\left(
\boldsymbol{\xi}^{(i)}, \left\{\boldsymbol{x}^{(i, j)}, r^{(i, j)}
\right\}_{j=1}^{M}\right)$ \Comment{$N$ realizations of the stochastic process $r(\boldsymbol{x}, \boldsymbol{\xi})$}
\State $0 < \text{Tol.} \ll 1$ \Comment{A small tolerance of target accuracy}
\State \textbf{Output:}  $\Phi(\boldsymbol{x}), \Psi(\boldsymbol{\xi})$ \Comment{Such that $r(\boldsymbol{x}, \boldsymbol{\xi}) \approx \Phi(\boldsymbol{x}) \Psi(\boldsymbol{\xi})$}
%
\State $ \Psi_{i} = 1$ \Comment{Initialize to ones for $1\le i \le N$}
\While{\text{Not Converged}} \Comment{Based on changes in $\Psi(\boldsymbol{\xi})$}
    \State $\Phi_j = \mathbb{E}_{i}[\Psi_{i} r^{(i, j)}] / \| \Psi \|_{L^2(p(\boldsymbol{\xi}))}$ \Comment{For all $1 \le j \le M$}
    \State $\Psi_{i} = \mathbb{E}_{j}[\Phi_j r^{(i, j)}] / \|\Phi \|_2^2$
\EndWhile
\end{algorithmic}
\end{algorithm}

A useful property of spectral expansion with orthogonal stochastic basis functions is the ability to estimate the variance analytically, without the need for Monte Carlo estimation. One can easily show that if $\Psi_{i}(\boldsymbol{\xi})$ are orthogonal, the variance of the truncated approximation $u^{(p)}$ can be calculated as follows:
\begin{equation}
    \mathbb{V}_{\boldsymbol{\xi}}[u(\boldsymbol{x}, \boldsymbol{\xi})]
    \approx 
    \mathbb{V}_{\boldsymbol{\xi}}[u^{(p)}(\boldsymbol{x}, \boldsymbol{\xi})] = \sum_{i=1}^P \Phi_i^2(\boldsymbol{x})\mathbb{E}_{\boldsymbol{\xi}}[\Psi_i^2(\boldsymbol{\xi})].
    \label{eqn:var_SSNO}
\end{equation}
If the stochastic basis functions are normalized, this equation simplifies to a summation over only the deterministic basis functions. We use this method of variance estimation through the identified basis functions as a sanity check to indirectly assess how well they satisfy orthogonality. If their orthogonality is poor, the variance estimated this way will deviate significantly from the true variance. Recall that the mean of the random process is equal to the zeroth-order term $\Phi_0(\boldsymbol{x})$, i.e., 
\begin{equation}
    \mathbb{E}_{\boldsymbol{\xi}}[u(\boldsymbol{x}, \boldsymbol{\xi})]
    \approx 
    \mathbb{E}_{\boldsymbol{\xi}}[u^{(p)}(\boldsymbol{x}, \boldsymbol{\xi})] = \Phi_0(\boldsymbol{x}).
    \label{eqn:mean_SSNO}
\end{equation}

\remark{This decomposition of mean and variance is commonly used in the global sensitivity analysis with Sobol' method \cite{sudret2008global}.}

\section{Numerical Examples}
\label{sec:exam}
In this section, five numerical examples are presented to discuss and illustrate the practical aspects of the proposed algorithms. In the first problem, the effectiveness of Algorithm~\ref{algo:SSDL-cont} is assessed using a stochastic ODE from the literature. In the second problem, a stochastic heat conduction problem is introduced to compare the performance of Algorithm~\ref{algo:SSDL-cont} to that of Algorithm~\ref{algo:SSDL-disc-cont}. In the third example, a stochastic beam deflection problem from the literature is selected to compare the proposed scheme with classical PCE. In the fourth example, a stochastic heat conduction PDE in two spatial dimensions is solved to benchmark neural chaos performance when vector random variables are involved. In the final example, we demonstrate the capability of neural chaos in addressing stochastic problems involving both linearly and nonlinearly dependent random variables.

\subsection{Example 1: First order SDE with a single random variable}
\label{sec:ex-xui-1d}
We begin by studying the robustness of the proposed method for
handling problems with scalar random variables sampled from different distributions. 
Here, data is generated by solving the simple stochastic differential equation introduced in \cite{xiu2002wiener} given by:
\begin{equation}
    \frac{d u}{dx} = \xi x; \xi\sim p(\xi),
    \label{eqn:sde1}
\end{equation}
where the random variable $\xi$ is sampled from a known probability distribution function $p(\xi)$. We consider four scenarios where the random variable is sampled from uniform, normal, gamma, and Poisson distributions ($\text{Uniform}(0, 1), \mathcal{N}(\mu=0, \sigma=1), \Gamma(k=1, \theta=1), \text{Poisson}(\lambda=1)$, respectively).
One thousand samples from each of these distributions are 
generated and the solutions of \cref{eqn:sde1}, using the values of these generated random variables, are stored at 20 equidistant points $x \in [0, 1]$. Seven hundred simulations are used for training, while the remaining 300 simulations are used for testing the stochastic neural operator.


The convergence rate of the model identified using Algorithm~\ref{algo:SSDL-disc-cont} as the number of terms in the spectral expansion increases is shown in Fig.~\ref{fig:ode1d-convg-disc-algo}. All models converge exponentially with the number of terms, but the rate depends on the type of distribution of the input random variable. After learning the corresponding neural network basis functions, the error distribution between the model predictions and the ground truth data is plotted in Fig.~\ref{fig:ode1d-errDist-disc-algo} for training and test realizations. On average, both training and test errors demonstrate satisfactory prediction capability, and the error distributions are almost identical between the training and test data, which demonstrates good generalization as well. The training error is slightly higher with the neural network basis functions compared to the basis vectors used in Fig.~\ref{fig:ode1d-convg-disc-algo}, which is due to the small neural network training error.

\begin{figure}[h]
  \centering
  \begin{subfigure}[b]{0.23\textwidth}
    \includegraphics[width=\textwidth]{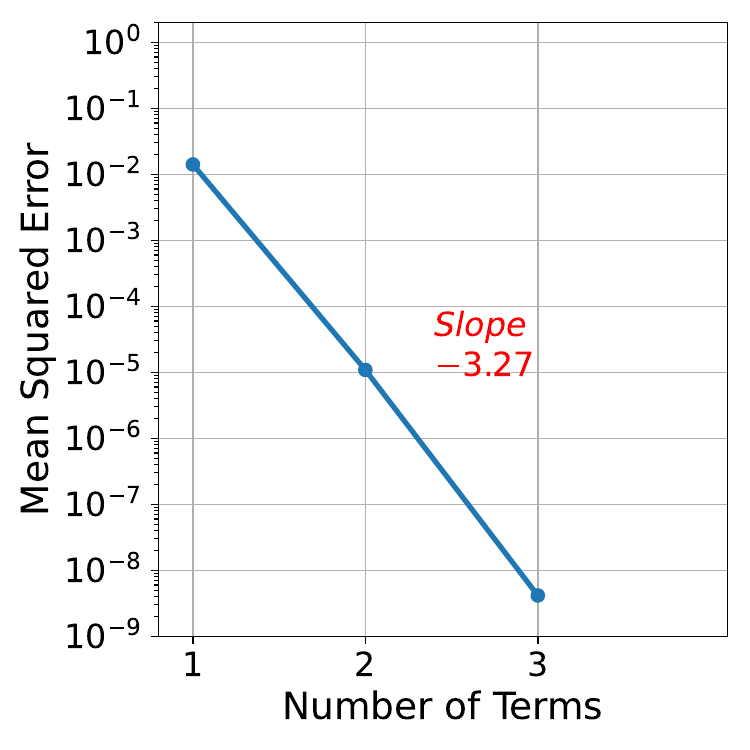}
    \caption{}
  \end{subfigure}
    \begin{subfigure}[b]{0.23\textwidth}
    \includegraphics[width=\textwidth]{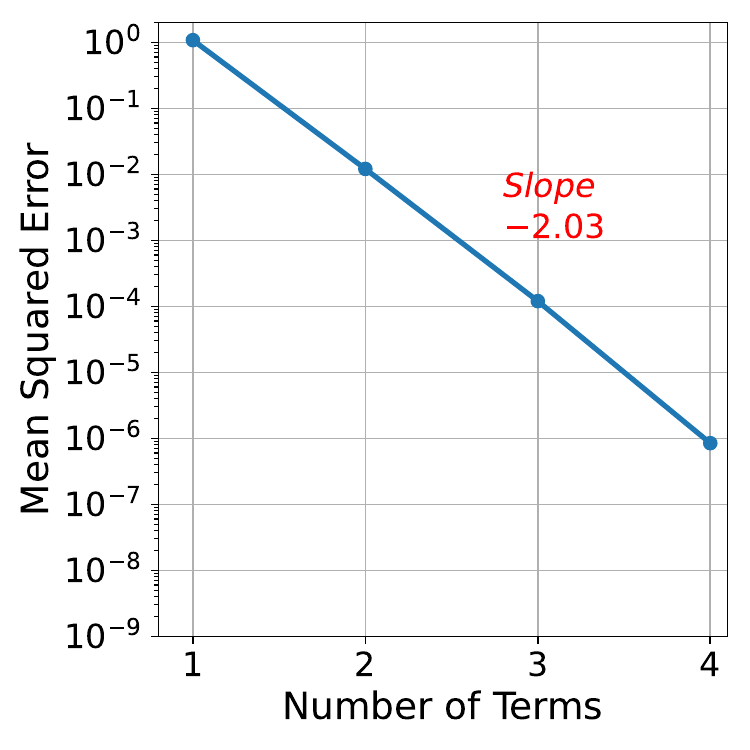}
    \caption{}
  \end{subfigure}
     \begin{subfigure}[b]{0.23\textwidth}
    \includegraphics[width=\textwidth]{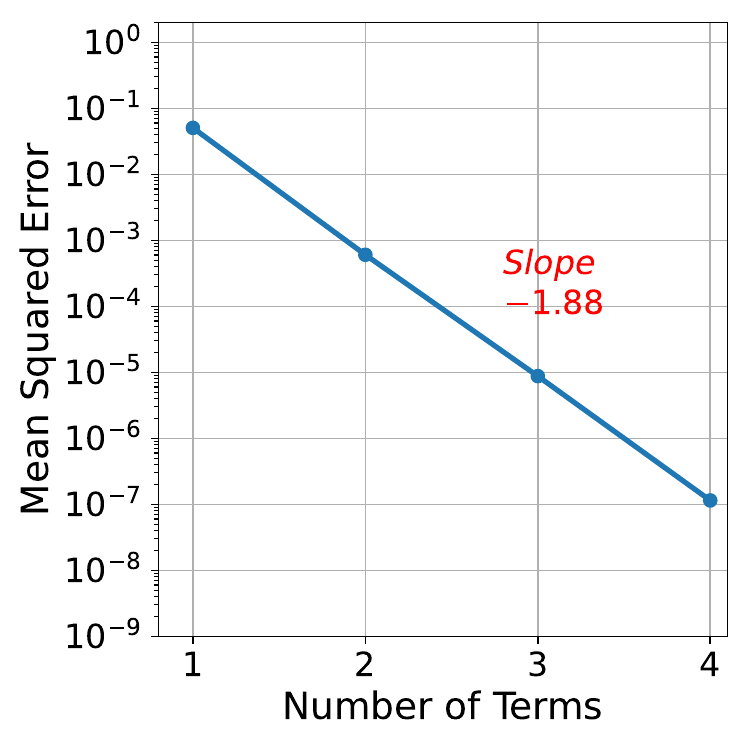}
    \caption{}
  \end{subfigure}
     \begin{subfigure}[b]{0.23\textwidth}
    \includegraphics[width=\textwidth]{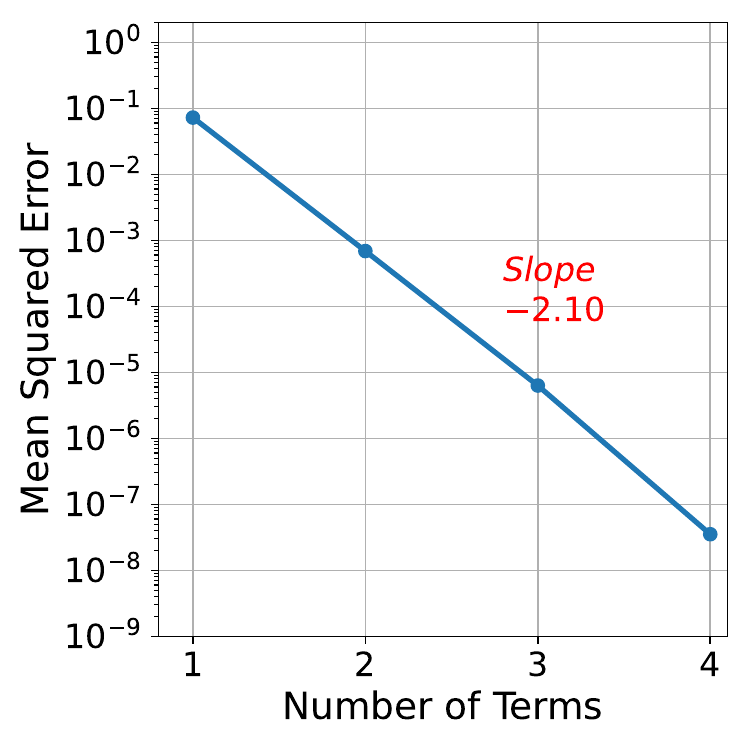}
    \caption{}
  \end{subfigure}
  \caption{
  Example 1 -- Mean squared error between the data and the learned spectral expansion via Algorithm~\ref{algo:SSDL-disc-cont} for increasing number of terms used in the expansion and different input distributions: (a) Uniform, (b) Normal, (c) Gamma, (d) Poisson. The slope indicates the rate of exponential error decay.
  }
  \label{fig:ode1d-convg-disc-algo}
\end{figure}

\begin{figure}[h]
  \centering
  \begin{subfigure}[b]{0.23\textwidth}
    \includegraphics[width=\textwidth]{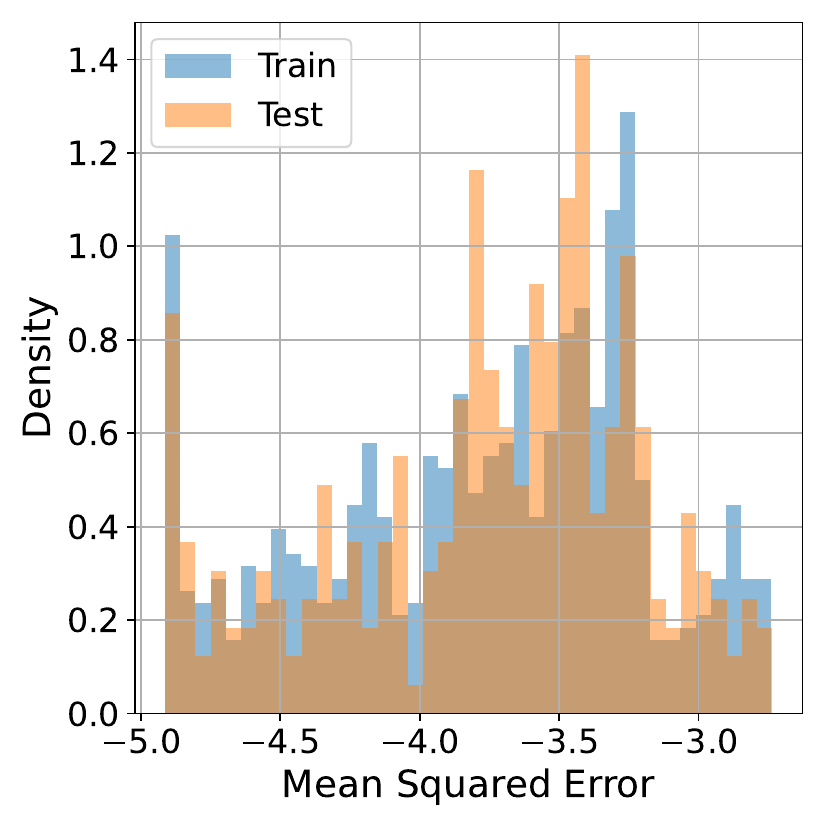}
    \caption{}
  \end{subfigure}
    \begin{subfigure}[b]{0.23\textwidth}
    \includegraphics[width=\textwidth]{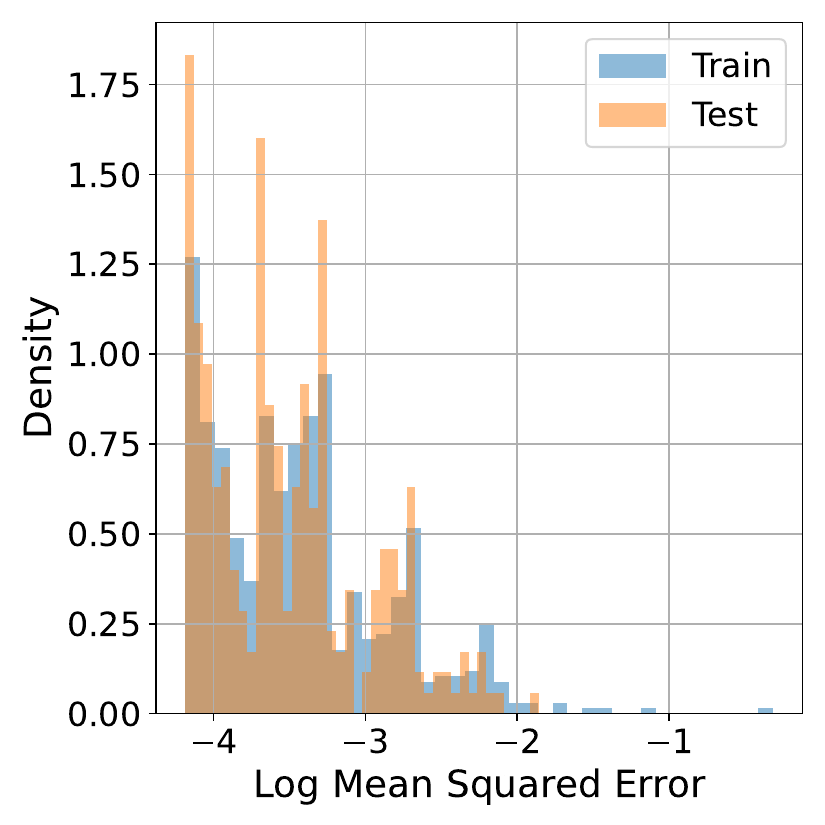}
    \caption{}
  \end{subfigure}
     \begin{subfigure}[b]{0.23\textwidth}
    \includegraphics[width=\textwidth]{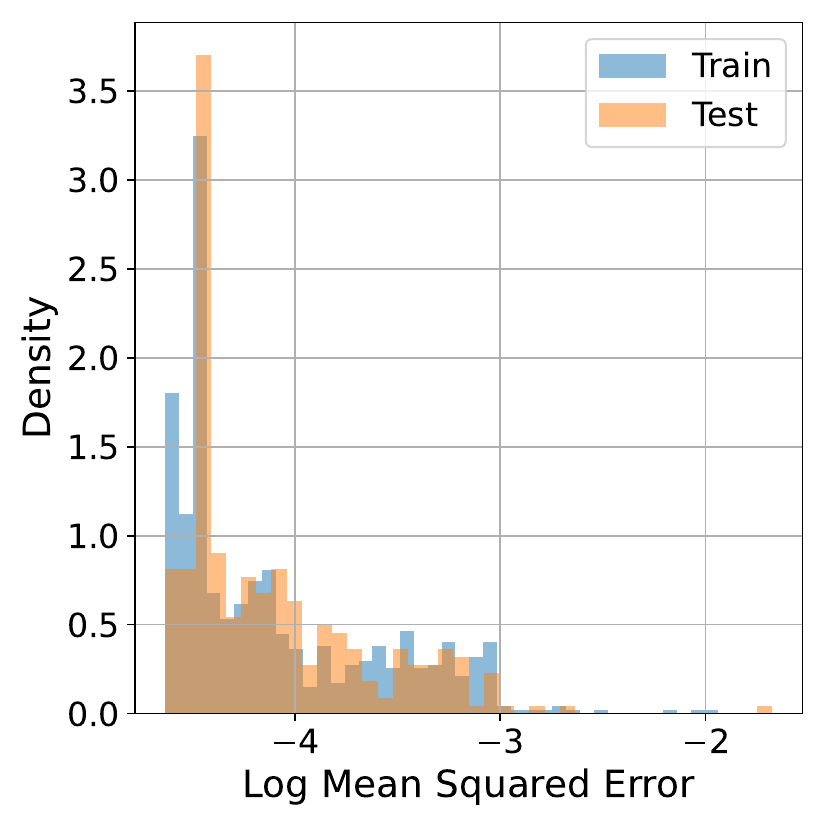}
    \caption{}
  \end{subfigure}
     \begin{subfigure}[b]{0.23\textwidth}
    \includegraphics[width=\textwidth]{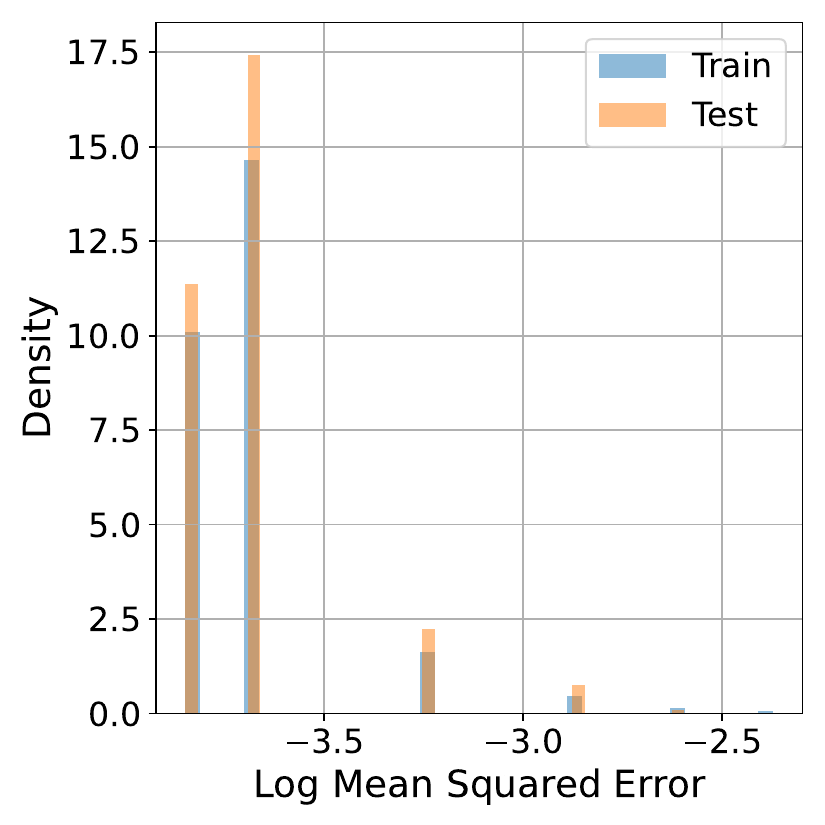}
    \caption{}
  \end{subfigure}
  \caption{
  Example 1 -- Distribution of error between the data and the model built via Algorithm~\ref{algo:SSDL-disc-cont} using neural network basis functions for different input distributions: (a) Uniform, (b) Normal, (c) Gamma, (d) Poisson. The training data includes 700 realizations, while the test data includes 300 realizations.
  }
  \label{fig:ode1d-errDist-disc-algo}
\end{figure}

\Cref{fig:ode1d-stoch-basis,fig:ode1d-det-basis} plot the learned stochastic and deterministic basis functions, respectively. The dots represent the learned basis vectors directly from Algorithm~\ref{algo:SSDL-disc-cont}, and the solid lines represent their fit as neural network basis functions, which is needed for inference on new data. As expected, neural networks adequately approximate the basis vectors; however, there is a discrepancy between them due to training error, for which we set a threshold of $5 \times 10^{-4}$ to stop the training iterations. This threshold hyperparameter controls the balance among approximation error, computation time, and overfitting. Importantly, both the stochastic and deterministic basis functions are data-dependent and distribution-dependent. Data-dependence will be explored later. Regarding distribution-dependence, while it's expected that the stochastic basis functions are dependent on the distribution, it may not be obvious that the deterministic basis functions should also be distribution dependent. But,  intuitively, if we solve \cref{eqn:sde1} deterministically for a set value of $\xi$, the optimal set of basis functions will depend on the selected value of $\xi$. Therefore, it follows that, if we change the distribution $p(\xi)$, the corresponding best deterministic basis will change as well.  This is clearly observed in ~\cref{fig:ode1d-det-basis}.

\begin{figure}[!ht]
  \centering
  \begin{subfigure}[b]{0.4\textwidth}
    \includegraphics[width=\textwidth]{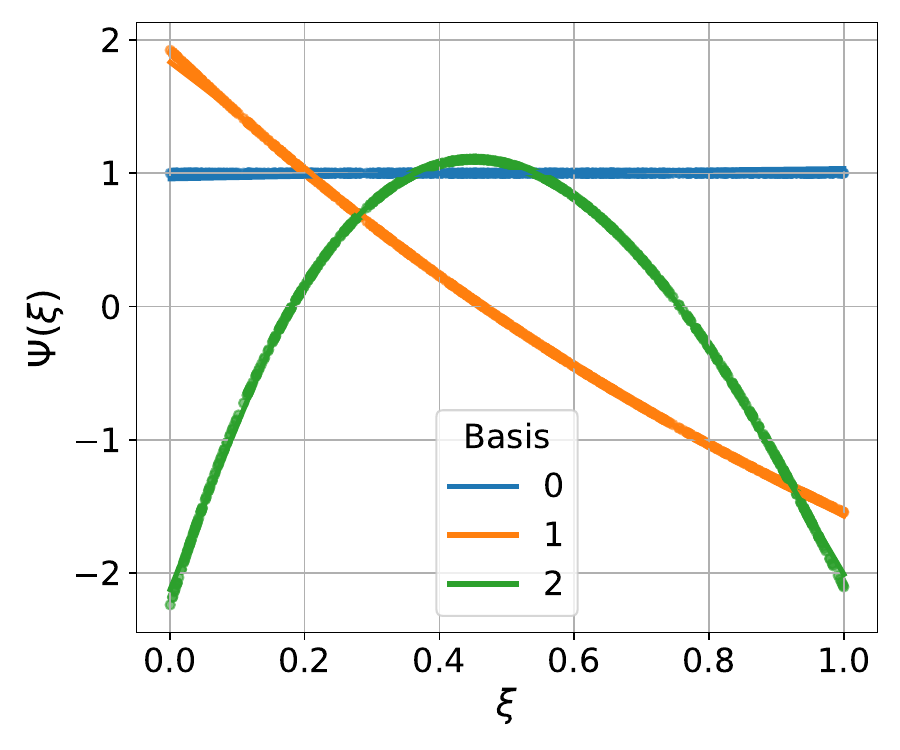}
    \caption{}
  \end{subfigure}
  \begin{subfigure}[b]{0.4\textwidth}
    \includegraphics[width=\textwidth]{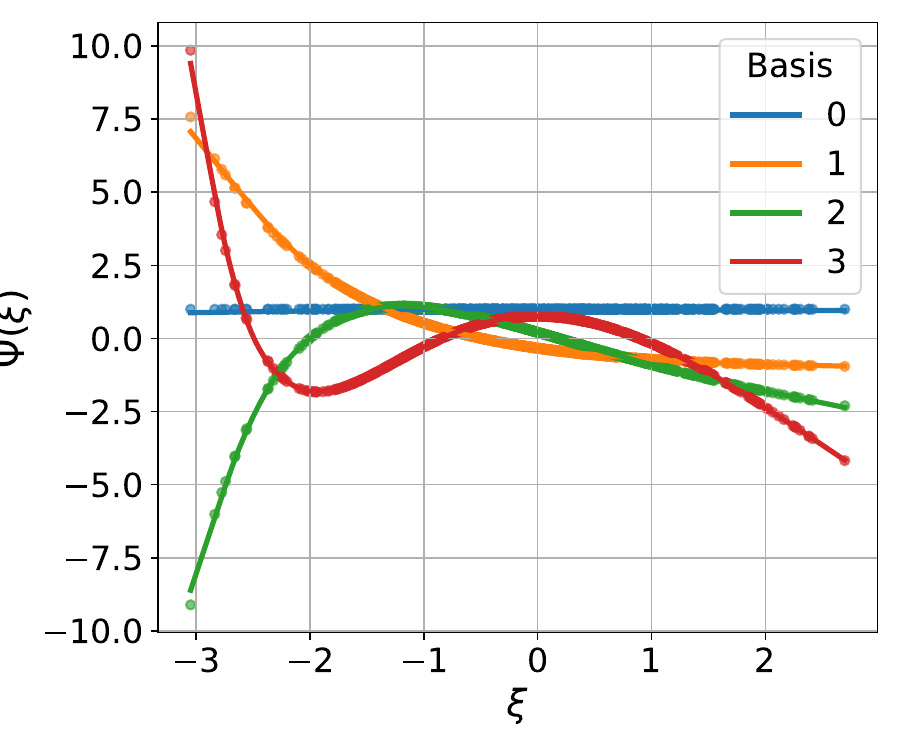}
    \caption{}
  \end{subfigure}
    \begin{subfigure}[b]{0.4\textwidth}
    \includegraphics[width=\textwidth]{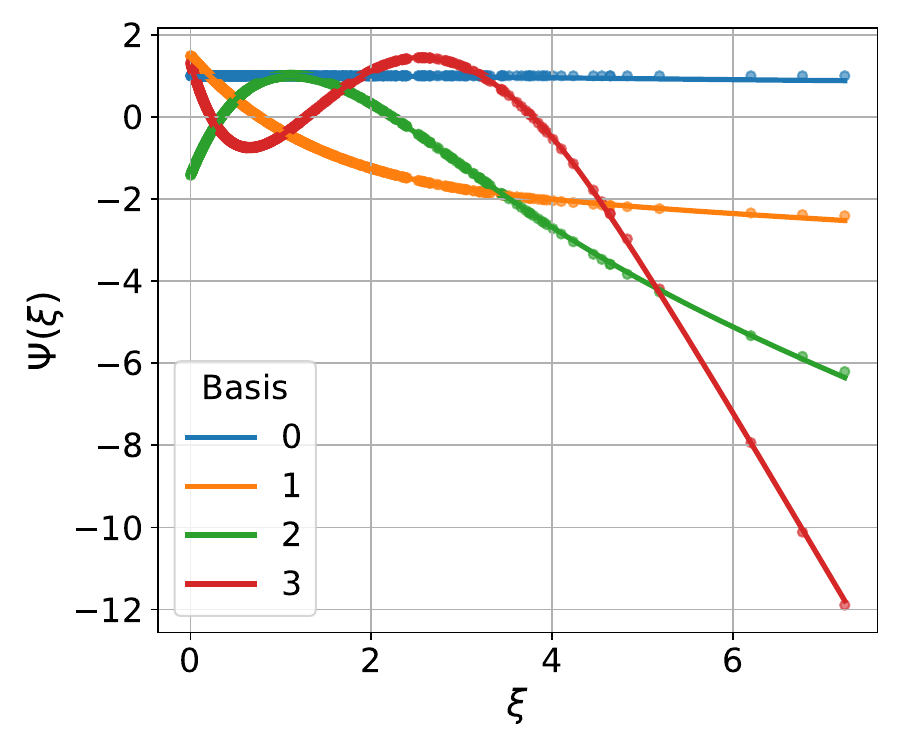}
    \caption{}
  \end{subfigure}
  \begin{subfigure}[b]{0.4\textwidth}
    \includegraphics[width=\textwidth]{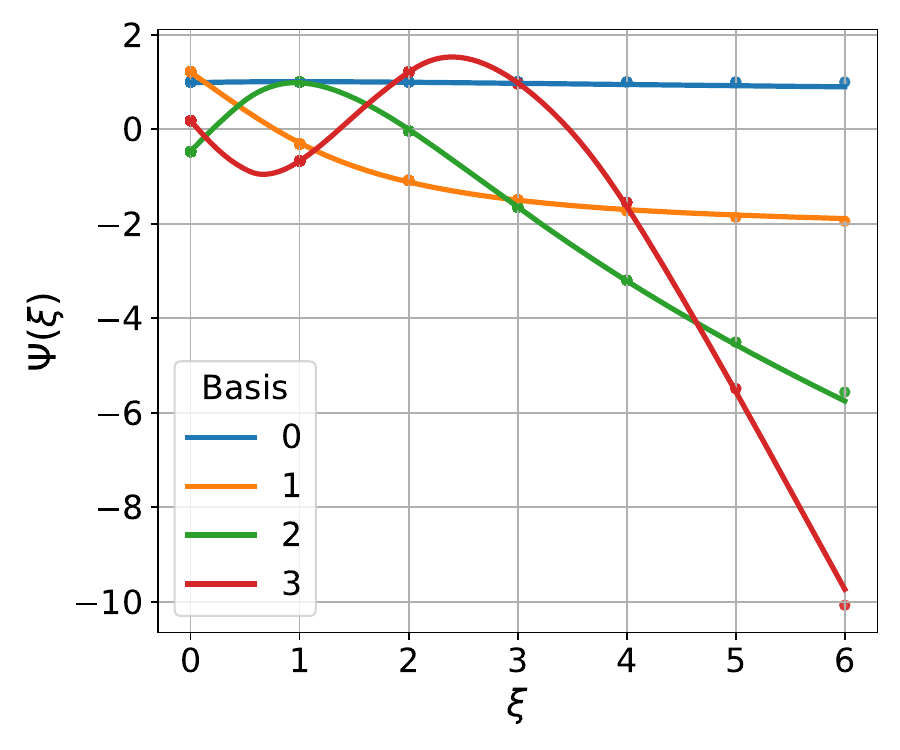}
    \caption{}
  \end{subfigure}
  \caption{
  Example 1 -- Learned stochastic basis functions for data sampled from different input distributions: (a) Uniform, (b) Normal, (c) Gamma, (d) Poisson. The dots represent the basis vectors learned from Algorithm~\ref{algo:SSDL-disc-cont}, while the solid lines correspond to their respective neural network basis functions.
  }
  \label{fig:ode1d-stoch-basis}
\end{figure}

\begin{figure}[!ht]
  \centering
  \begin{subfigure}[b]{0.4\textwidth}
    \includegraphics[width=\textwidth]{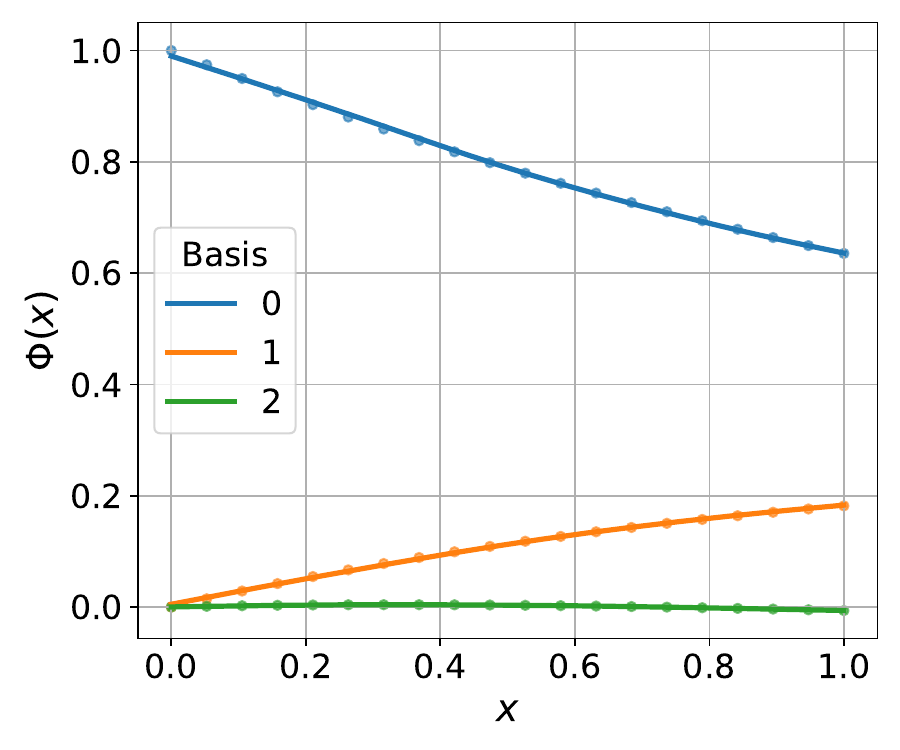}
    \caption{}
  \end{subfigure}
  \begin{subfigure}[b]{0.4\textwidth}
    \includegraphics[width=\textwidth]{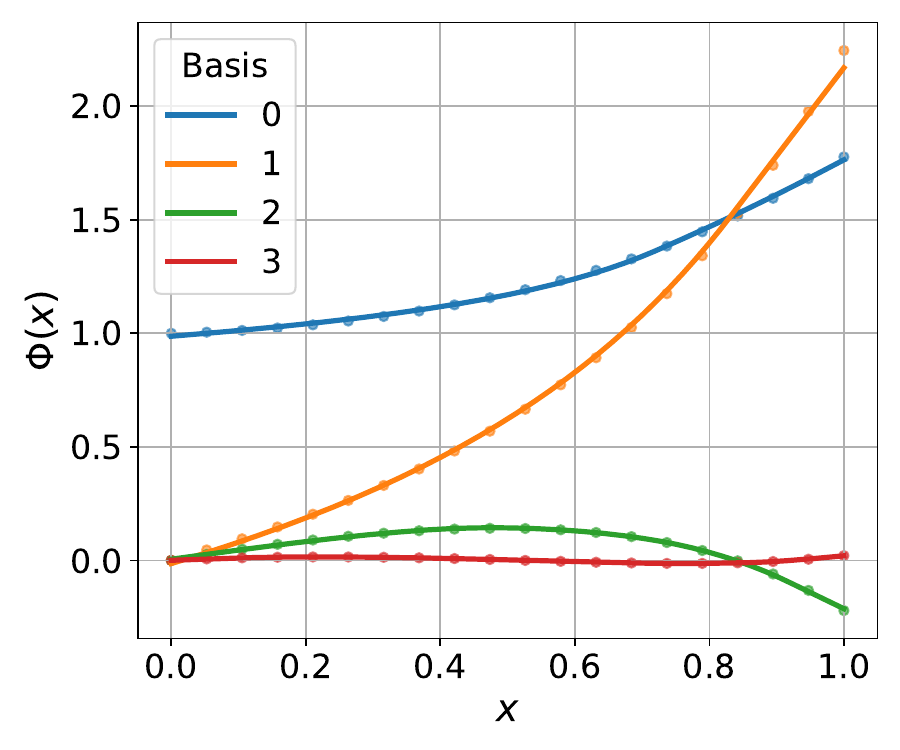}
    \caption{}
  \end{subfigure}
    \begin{subfigure}[b]{0.4\textwidth}
    \includegraphics[width=\textwidth]{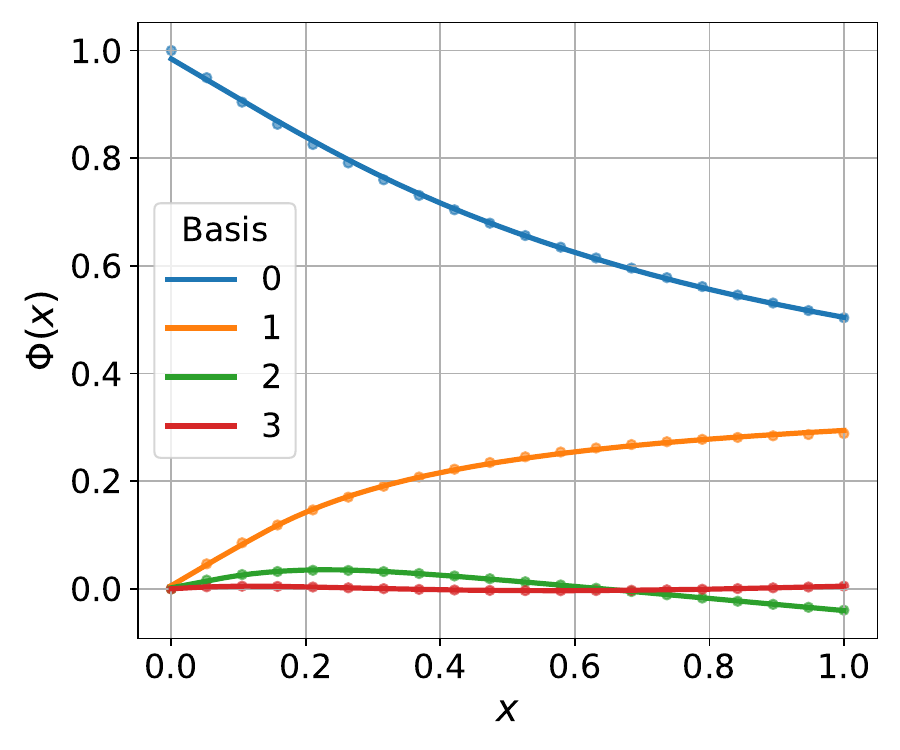}
    \caption{}
  \end{subfigure}
  \begin{subfigure}[b]{0.4\textwidth}
    \includegraphics[width=\textwidth]{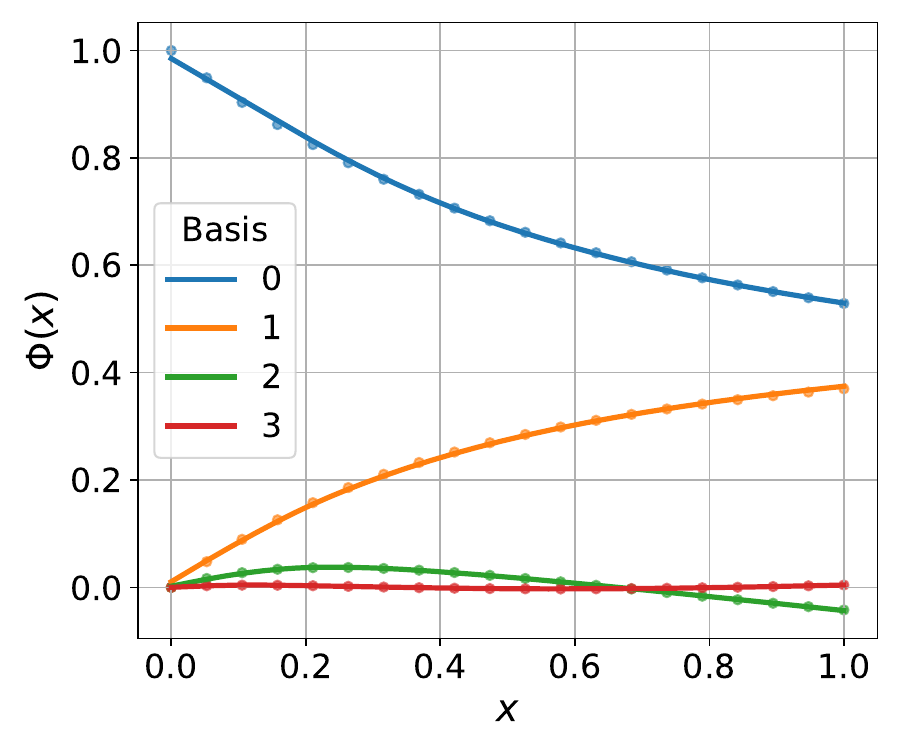}
    \caption{}
  \end{subfigure}
  \caption{
  Example 1 -- Learned deterministic basis functions for data sampled from different input distributions: (a) Uniform, (b) Normal, (c) Gamma, (d) Poisson. The dots represent the basis vectors learned from Algorithm~\ref{algo:SSDL-disc-cont}, while the solid lines correspond to their respective neural network basis functions.
  }
  \label{fig:ode1d-det-basis}
\end{figure}

As mentioned, one of the advantages of a spectral expansion is direct mean and variance estimation (from \cref{eqn:mean_SSNO,eqn:var_SSNO}), without the need for Monte Carlo simulation during inference. We demonstrate this feature of the proposed method for the different input distributions in \cref{fig:ode1d-mean,fig:ode1d-var}. These figures show the mean $\pm$ one standard deviation confidence bounds from 100 repeated trials of the spectral stochastic neural operator along with the `ground truth' estimated by Monte Carlo simulation with $10^5$ simulations. These plots further serve to indirectly check the orthogonality of the learned stochastic basis functions. If they are not orthogonal, the mean and standard deviation estimates from the expansion coefficients would not match the true mean and variance. 
As the results suggest, there is good agreement between the model's estimated mean and standard deviation and the ground truth. However, the error represented by the confidence bounds could be reduced if more data were provided during training.

\begin{figure}[h]
  \centering
  \begin{subfigure}[b]{0.4\textwidth}
    \includegraphics[width=\textwidth]{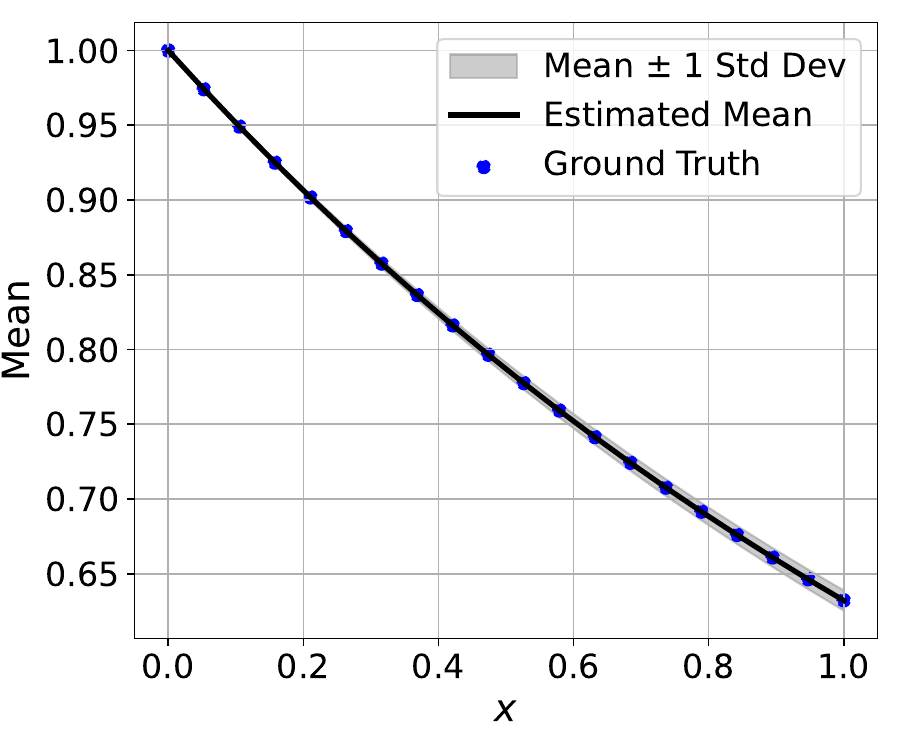}
    \caption{}
  \end{subfigure}
  \begin{subfigure}[b]{0.4\textwidth}
    \includegraphics[width=\textwidth]{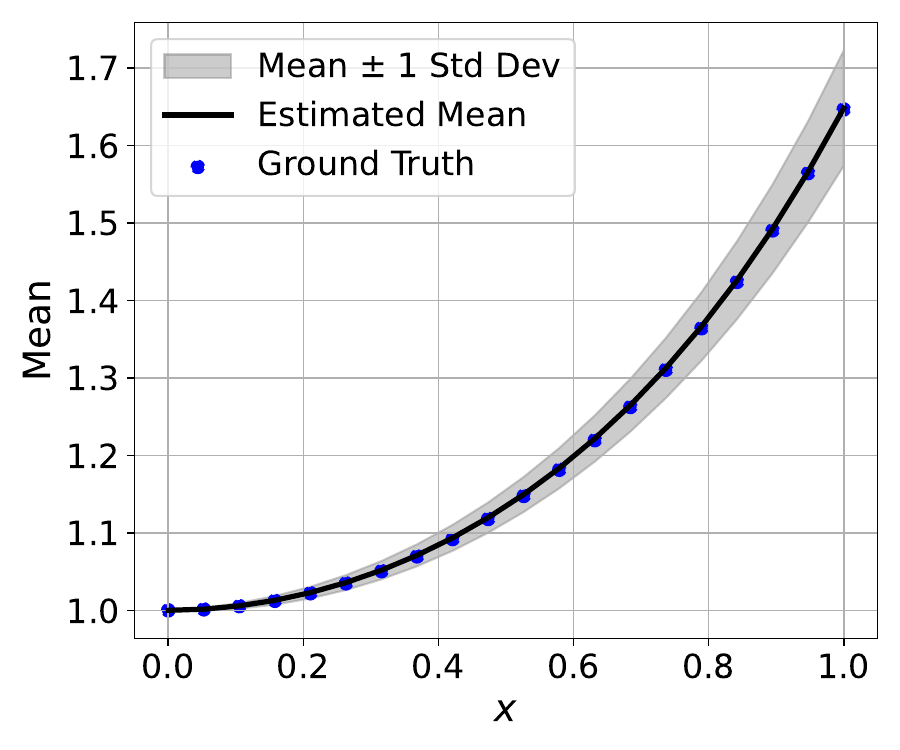}
    \caption{}
  \end{subfigure}
    \begin{subfigure}[b]{0.4\textwidth}
    \includegraphics[width=\textwidth]{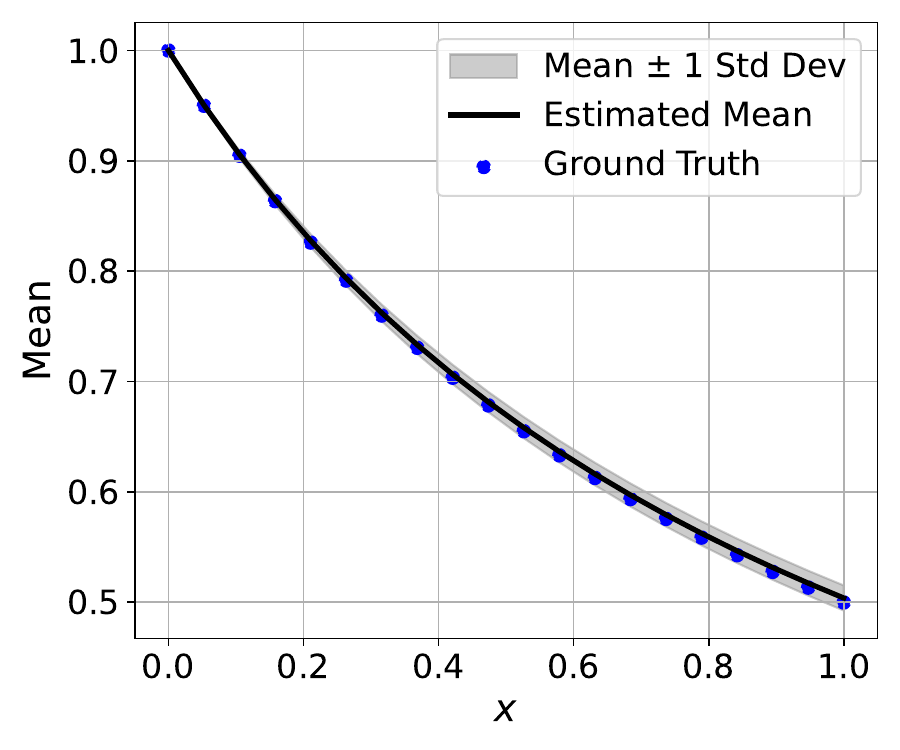}
    \caption{}
  \end{subfigure}
  \begin{subfigure}[b]{0.4\textwidth}
    \includegraphics[width=\textwidth]{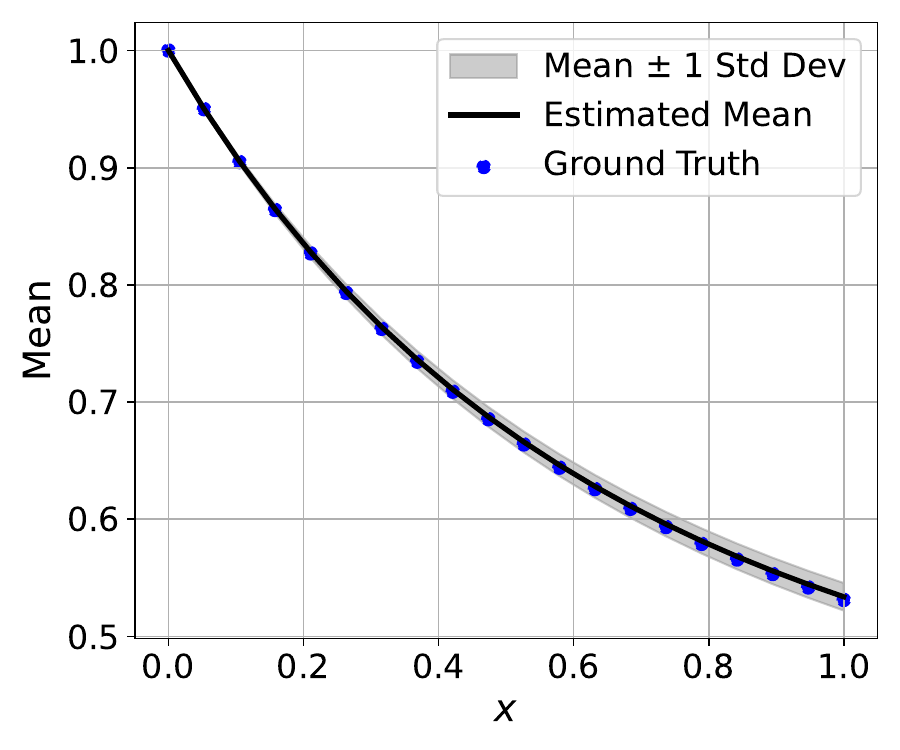}
    \caption{}
  \end{subfigure}
  \caption{
  Example 1 -- Mean field estimate obtained from the first coefficient of the spectral stochastic neural operator for different input distributions: (a) Uniform, (b) Normal, (c) Gamma, (d) Poisson. The plot shows the mean $\pm$ one standard deviation for 100 independent training trials. The ground truth is estimated from Monte Carlo simulation using $10^5$ realizations.
  }
  \label{fig:ode1d-mean}
\end{figure}

\begin{figure}[h]
  \centering
  \begin{subfigure}[b]{0.4\textwidth}
    \includegraphics[width=\textwidth]{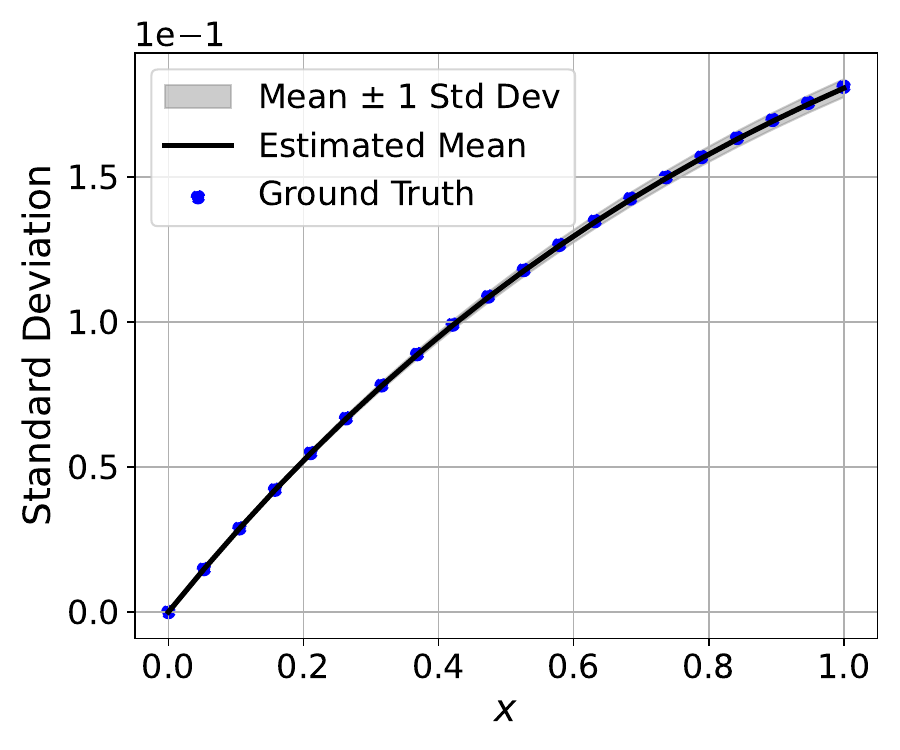}
    \caption{}
  \end{subfigure}
  \begin{subfigure}[b]{0.4\textwidth}
    \includegraphics[width=\textwidth]{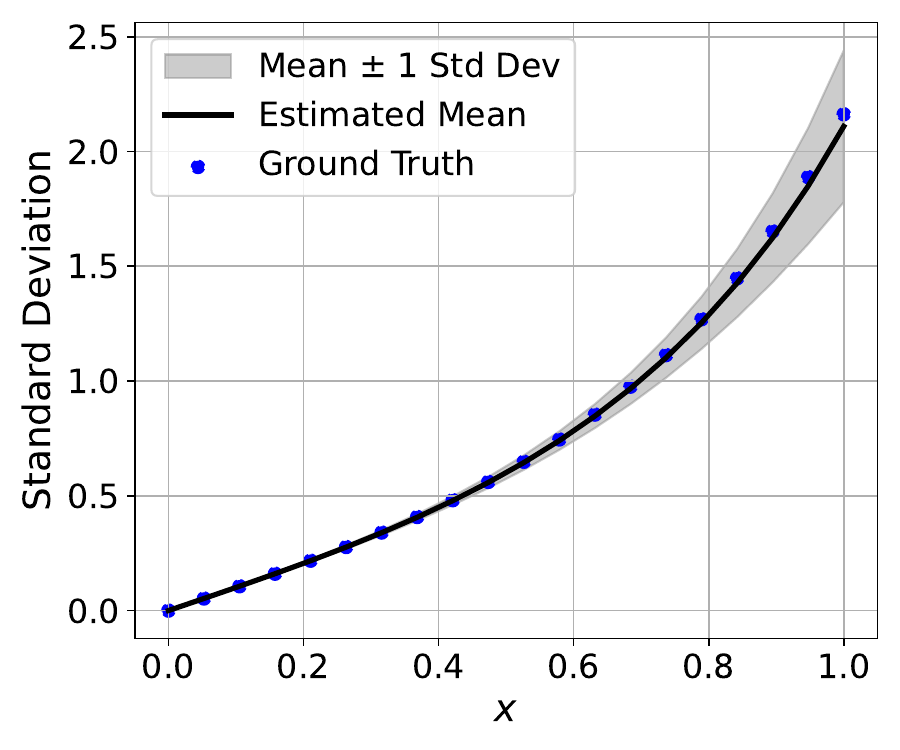}
    \caption{}
  \end{subfigure}
    \begin{subfigure}[b]{0.4\textwidth}
    \includegraphics[width=\textwidth]{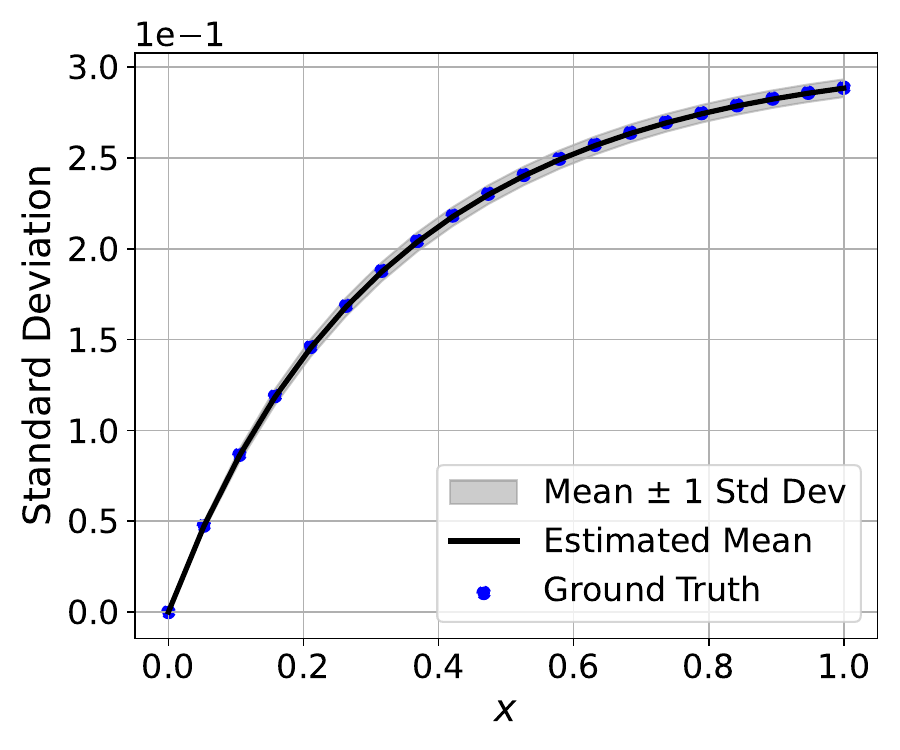}
    \caption{}
  \end{subfigure}
  \begin{subfigure}[b]{0.4\textwidth}
    \includegraphics[width=\textwidth]{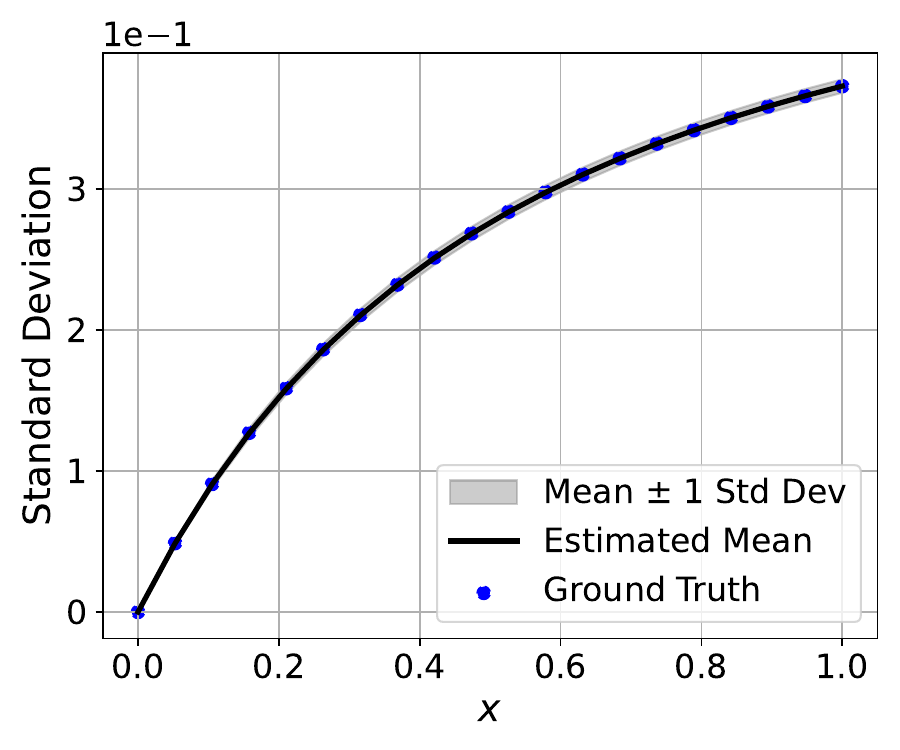}
    \caption{}
  \end{subfigure}
  \caption{
  Example 1 -- Standard Deviation field estimate obtained from the first coefficient of the spectral stochastic neural operator for different input distributions: (a) Uniform, (b) Normal, (c) Gamma, (d) Poisson. The plot shows the mean $\pm$ one standard deviation for 100 independent training trials. The ground truth is estimated from Monte Carlo simulation using $10^5$ realizations.
  }
  \label{fig:ode1d-var}
\end{figure}

\subsection{Example 2: Second order SDE with 1D random variable}
\label{sec:ex-heat1D}

In this problem, we compare the continuous and discrete-continuous  algorithms (Algorithms~\ref{algo:SSDL-cont} and \ref{algo:SSDL-disc-cont}, respectively) and demonstrate the superior performance of the discrete-continuous algorithm. Additionally, we compare the proposed method with the classical PCE-based surrogate model to demonstrate how the proposed method converges with fewer terms than PCE.

We consider the stochastic heat conduction boundary value problem under steady state as follows:
\begin{equation}
    \frac{d}{dx}
    \left(
        \kappa(x, \xi) \frac{d u}{d x}
    \right)
    = \sin(2\pi x), x\in[0, 1], \xi \sim \text{Uniform}[1, 3],
\end{equation}
where $u(x, \xi)$ is the unknown random temperature field under deterministic homogeneous boundary conditions (i.e., $u(0, \xi) = u(1, \xi) = 0 $). The stochasticity arises from the assumed random heat conduction $\kappa(x, \xi) = 1.1 + \cos(x\xi)$. The only input random variable in this process is $\xi \sim \text{Uniform}[1, 3]$. For training and testing purposes, 1000 realizations are generated, with the spatial domain discretized on a one-dimensional grid with 30 equidistant points. Seventy percent of the data are used for training, while the remaining 30\% are used as test data.
%

Under this setup, classical PCE uses Legendre basis functions. 
However, many other basis functions, such as Fourier or rescaled Chebyshev basis functions, can also be constructed to be orthogonal under a uniform distribution (albeit not optimal in the sense of the Weiner-Askey scheme\cite{xiu2002wiener}). Given that multiple sets of orthogonal basis functions can be established \textit{a priori}, it seems more reasonable to suggest that orthogonality should not be the sole criterion; selection of the appropriate basis should also depend on the actual stochastic operator. In this regard, the proposed approach offers greater flexibility (and, as we'll show, a  more compact representation) than methods that prescribe the basis \textit{a priori}.
The learned deterministic and stochastic basis functions from Algorithm~\ref{algo:SSDL-disc-cont} are plotted in \cref{fig:heat1d-basis-disc-algo}. The neural network basis functions directly learned from the fully continuous approach in Algorithm~\ref{algo:SSDL-cont} are plotted in \cref{fig:heat1d-basis-cont-algo}. One can observe an almost one-to-one match between the dominant basis functions (zeroth and first terms) between these two algorithms, up to a sign difference. However, the second stochastic basis from Algorithm~\ref{algo:SSDL-cont} does not match the optimal solution from Algorithm~\ref{algo:SSDL-disc-cont}. There are two reasons for this: 1) the second deterministic basis is nearly zero everywhere in the domain, meaning the second term, in comparison to the other terms, does not contribute much; and 2) the optimization in the fully continuous version becomes more difficult for higher-order terms due to the complexity of training two neural networks simultaneously in a multiplicative manner. To further assess the second point, we compare the convergence of these two algorithms with increasing number of terms in \cref{fig:heat1d-err-comp}. Clearly, the rate of convergence of the fully continuous method decays for higher-order terms.  Although theoretically the fully continuous algorithm should converge to the optimal solution, in practice, this highly nonlinear and non-convex optimization is challenging from an optimization perspective. Hence, the discrete-continuous algorithm offers much more flexibility, as it learns the neural network basis in an offline manner and as a post-processing step.

\begin{figure}[!ht]
  \centering
  \begin{subfigure}[b]{0.4\textwidth}
    \includegraphics[width=\textwidth]{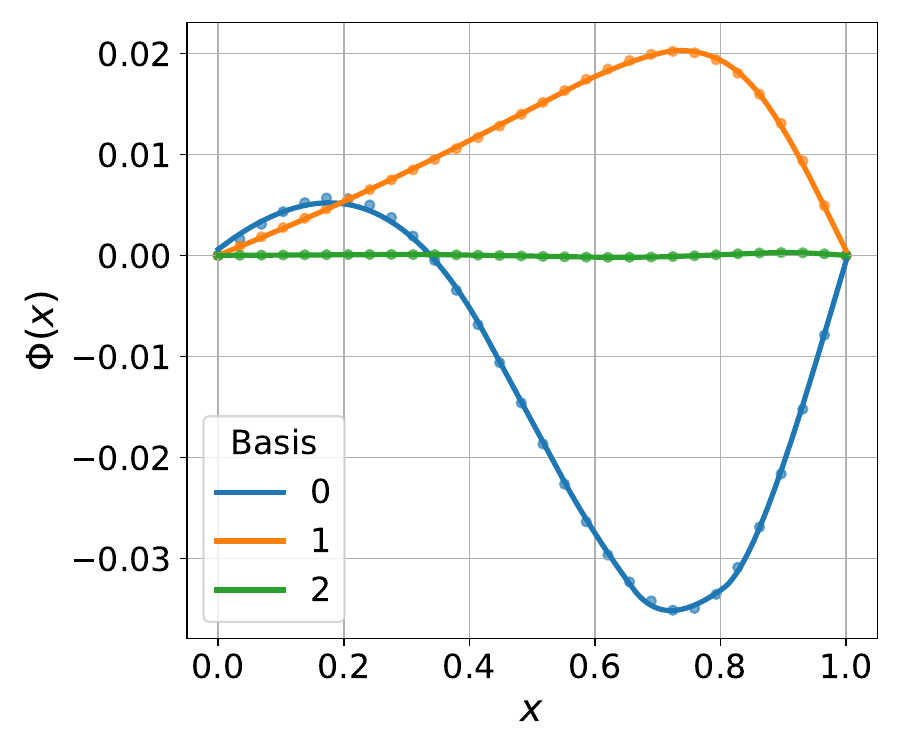}
    \caption{}
  \end{subfigure}
    \begin{subfigure}[b]{0.4\textwidth}
    \includegraphics[width=\textwidth]{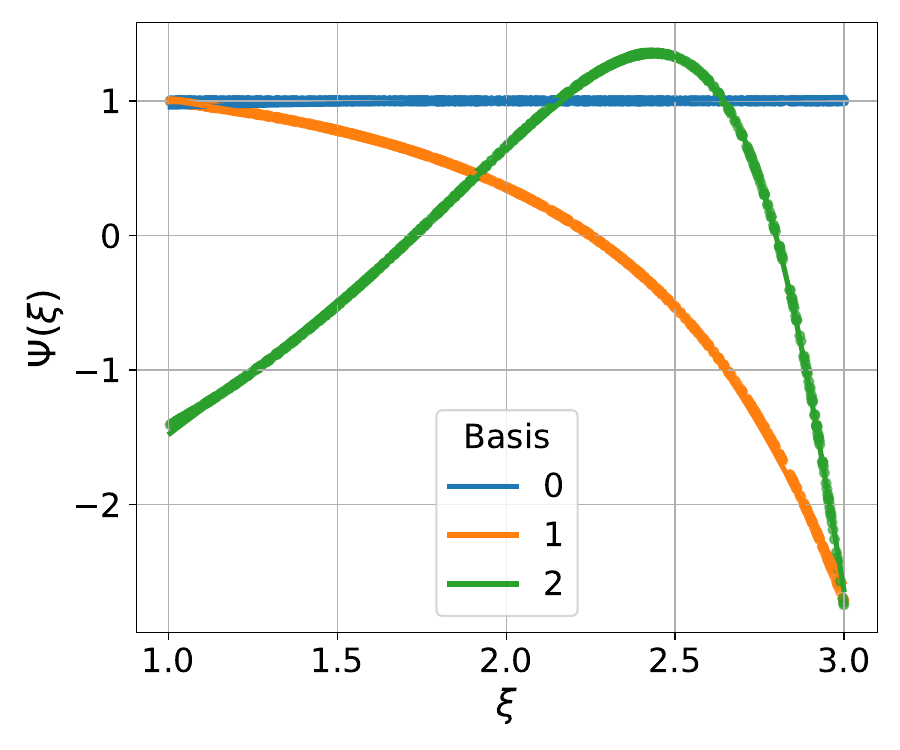}
    \caption{}
  \end{subfigure}
  \caption{
  Example 2 -- Learned (a) deterministic and (b) stochastic basis functions. The dots represent the basis functions learned from Algorithm~\ref{algo:SSDL-disc-cont}, while the solid lines correspond to their corresponding neural network basis functions.
  }
  \label{fig:heat1d-basis-disc-algo}
\end{figure}

\begin{figure}[!ht]
  \centering
  \begin{subfigure}[b]{0.4\textwidth}
    \includegraphics[width=\textwidth]{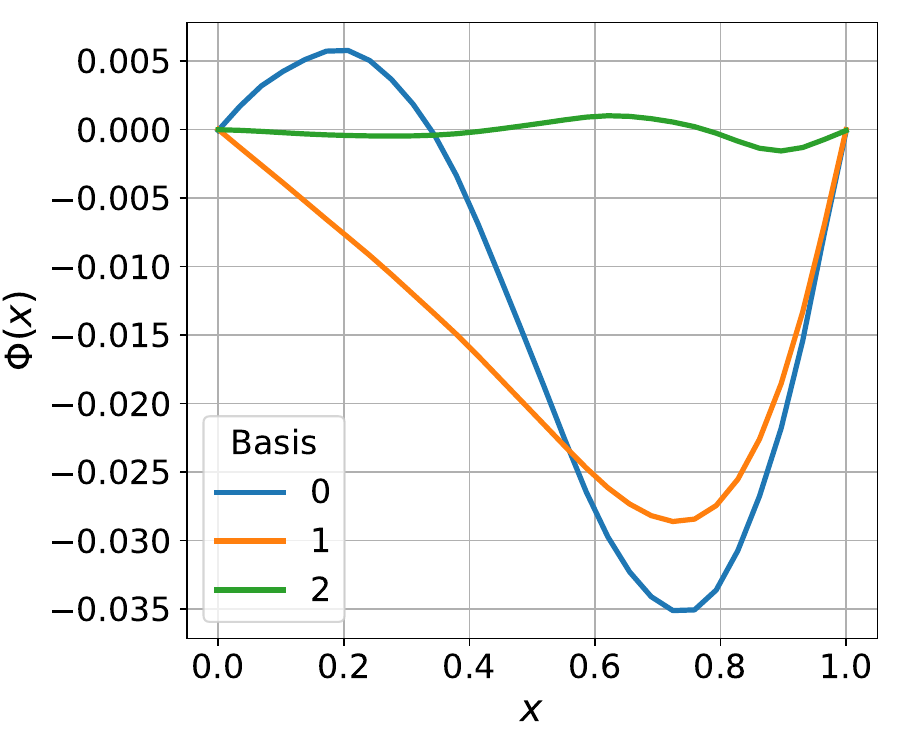}
    \caption{}
  \end{subfigure}
    \begin{subfigure}[b]{0.4\textwidth}
    \includegraphics[width=\textwidth]{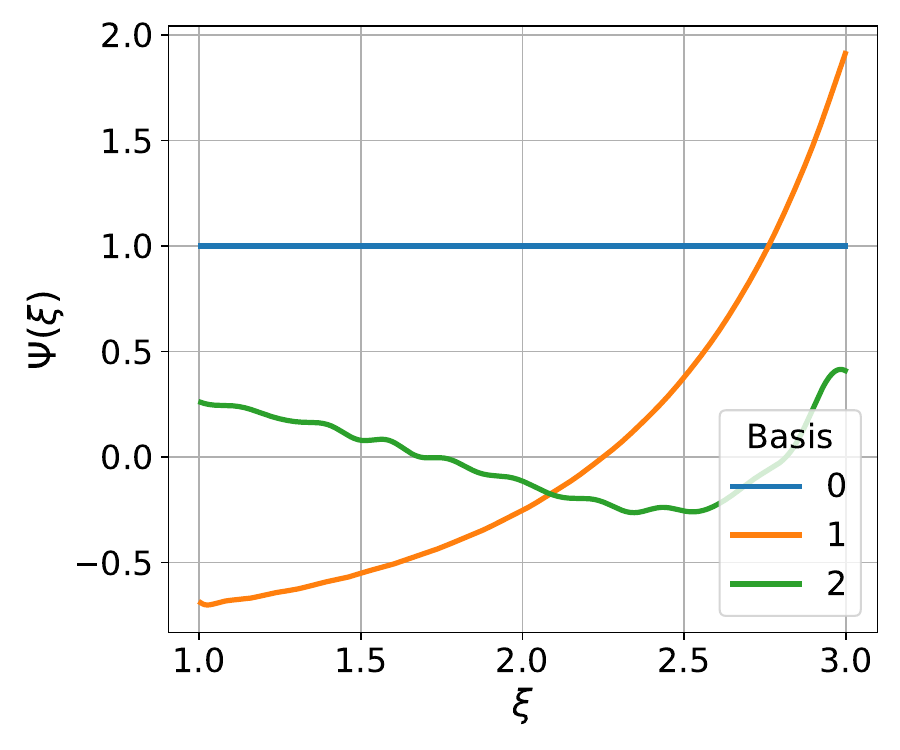}
    \caption{}
  \end{subfigure}
  \caption{
  Example 2 --- Learned (a) deterministic and (b) stochastic basis functions using Algorithm~\ref{algo:SSDL-cont}.
  }
  \label{fig:heat1d-basis-cont-algo}
\end{figure}

\begin{figure}[h]
  \centering
    \includegraphics[width=0.35\textwidth]{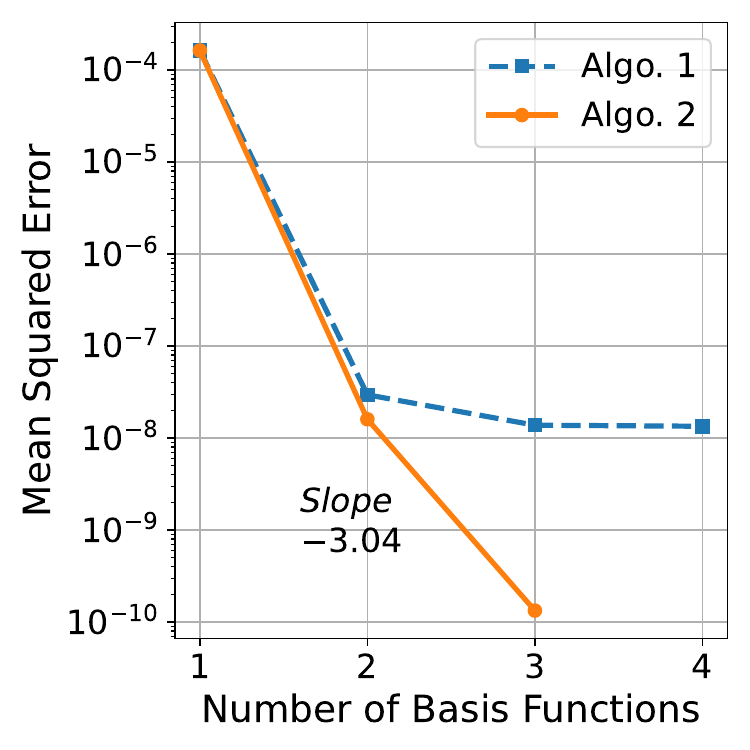}
  \caption{
  Example 2 -- Mean squared error between the data and the learned spectral expansion using 
  Algorithm~\ref{algo:SSDL-cont} (dashed line) and Algorithm~\ref{algo:SSDL-disc-cont} (solid line).
  }
  \label{fig:heat1d-err-comp} 
\end{figure}

We now explore how many PCE terms are needed to replicate the same basis functions, assuming the identified basis functions are correct and optimal. \Cref{fig:heat1d-pce-basis-disc-algo} shows the fits of Legendre polynomials of different orders (i.e., 2, 4, and 8) to the second and third identified stochastic basis functions. 
To achieve an accuracy threshold of \(\sim1 \times 10^{-6}\) for this simple one-dimensional problem, more than 8 Legendre polynomials are required, whereas the proposed method only requires two neural basis functions. This is expected, as neural network representations are much more flexible. This becomes especially important in higher dimensions where PCE requires tensor product bases (unlike neural chaos) and compact basis sets are essential.

\begin{figure}[!ht]
  \centering
  \begin{subfigure}[b]{0.7\textwidth}
    \includegraphics[width=\textwidth]{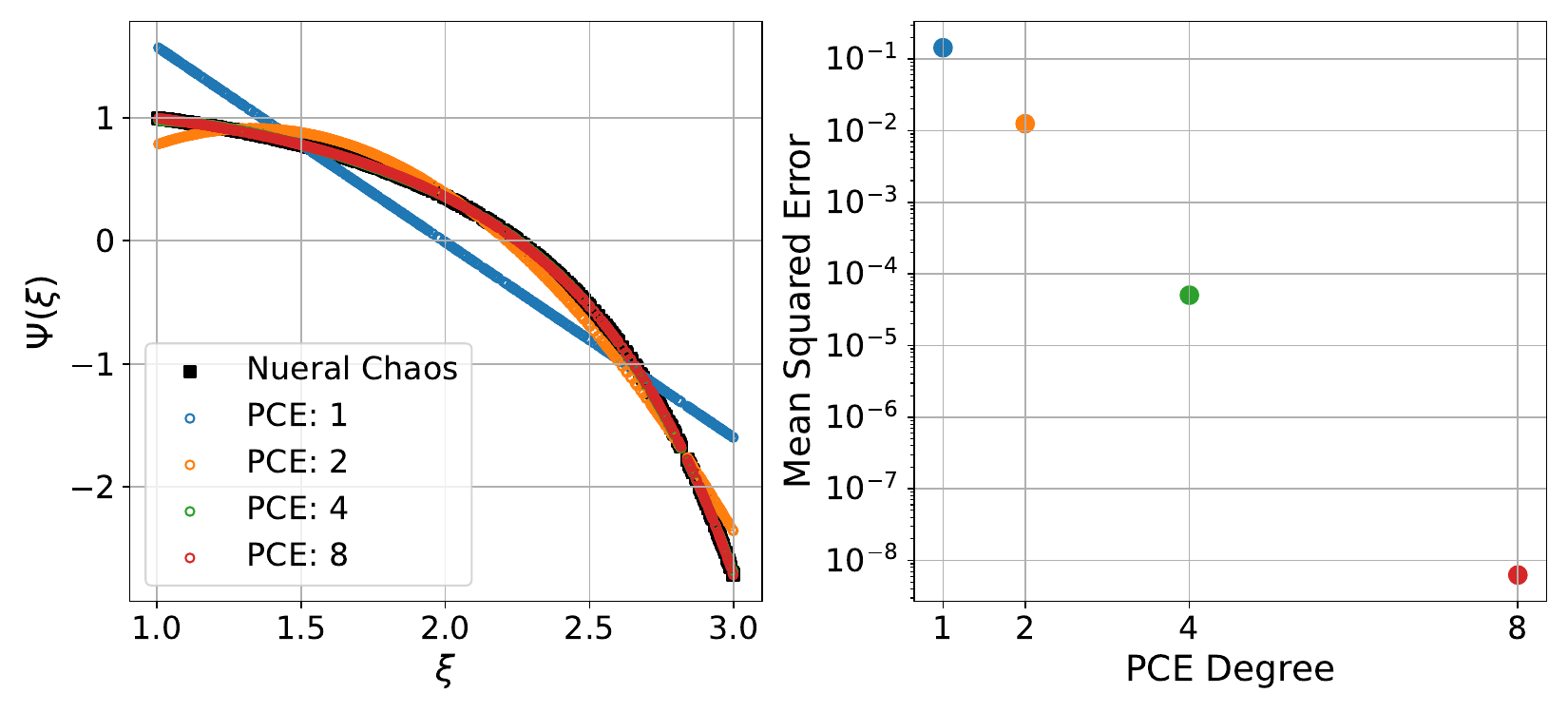}
    \caption{}
  \end{subfigure}
    \begin{subfigure}[b]{0.7\textwidth}
    \includegraphics[width=\textwidth]{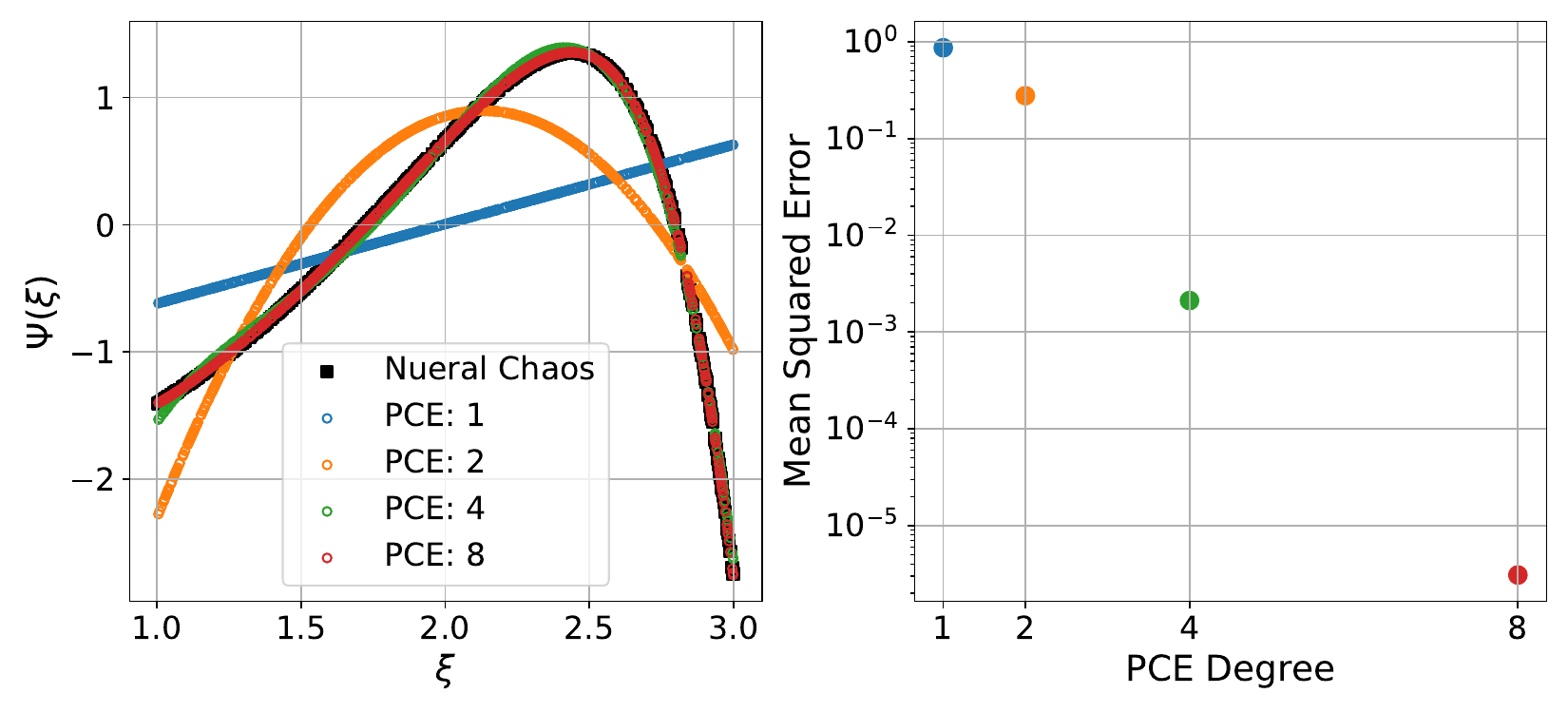}
    \caption{}
  \end{subfigure}
  \caption{
  Example 2 -- Approximation error of the learned neural chaos basis functions by Legendre polynomials: (a) second stochastic basis function, (b) third stochastic basis function. The left plot shows the functions, while the right plot indicates the MSE based on the number of terms used. In the left plot, points with solid square markers represent the learned basis from the proposed method, while circular dots in different colors correspond to different PCEs with varying numbers of terms.
  }
  \label{fig:heat1d-pce-basis-disc-algo}
\end{figure}

\subsection{Example 3: Fourth order SDE with 7D random variables}
\label{sec:ex-beam}
In this example, we consider a random vector input of higher dimensions. This problem setup is adopted from \cite{sharma2024physics}.

The governing equation for deflection of a beam of length $L=10$m under the Euler-Bernoulli simplification is as follows:
\begin{equation}
    \frac{d^2}{dx^2}
    \left(
        K(x,\boldsymbol{\xi}) \frac{d^2 u}{d x^2}
    \right)
    = -0.005, x\in[0, 10],
\end{equation}
where $u(x, \boldsymbol{\xi})$ is the unknown beam deflection under the deterministic simply supported boundary conditions ($u(0, \boldsymbol{\xi}) = M(0, \boldsymbol{\xi}) = u(10, \boldsymbol{\xi}) = M(10, \boldsymbol{\xi})$) and the bending moment $M(x, \boldsymbol{\xi})$ is defined as follows:
\begin{equation}
    M(x, \boldsymbol{\xi}) = - K(x, \boldsymbol{\xi}) \frac{d^2 u}{d x^2}.
\end{equation}
Here, we assume the only source of stochasticity in the deflection field is attributed to the stiffness random field $K(x,\boldsymbol{\xi})$, which is modeled as a Gaussian random process $K\sim \mathcal{GP}(8, \text{Cov}(x_1, x_2;l_c))$, with mean value $\mu_K=8$GPa and radial basis function covariance defined as follows:
\begin{equation}
    \text{Cov}(x_1, x_2;l_c)) = \exp
    \left(
        -\frac{(x_1 - x_2)^2}{2 l_c^2}
    \right),
\end{equation}
where its characteristic length scale is $l_c=2$m. One hundred representative samples of this random field, with the mean removed, are plotted in \cref{fig:beam1d-input-data}(a), and the decay of the eigenvalues of the covariance matrix computed from 1000 samples is depicted in \cref{fig:beam1d-input-data}(b), illustrating that the effective dimensionality of the random field can be set equal to 7. 
The Gaussian random field is discretized and simulated using the Karhunen-Loève (KL) expansion~\cite{phoon2002simulation} with seven independent standard normal random variables.

\begin{figure}[!ht]
  \centering
  \begin{subfigure}[b]{0.4\textwidth}
    \includegraphics[width=\textwidth]{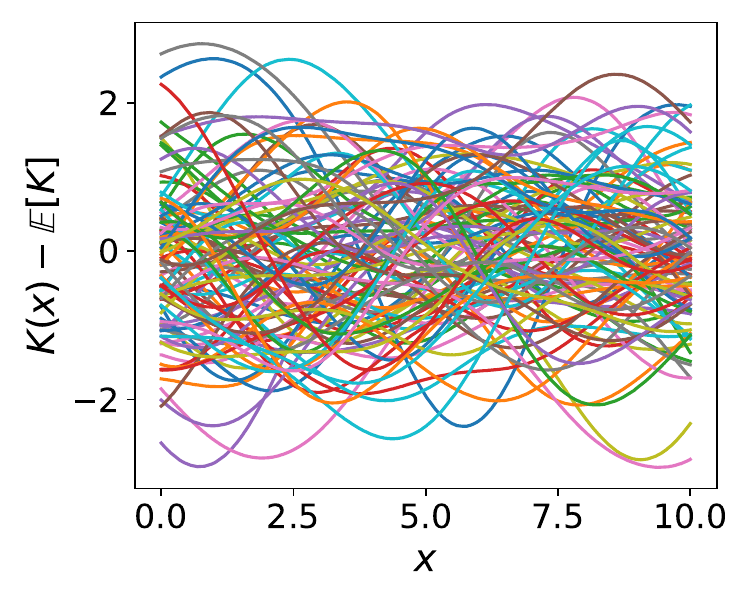}
    \caption{}
  \end{subfigure}
    \begin{subfigure}[b]{0.4\textwidth}
    \includegraphics[width=\textwidth]{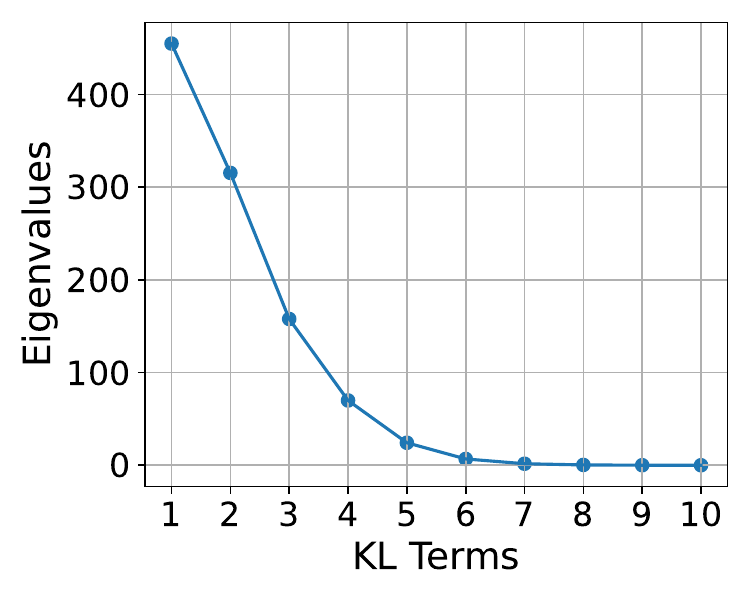}
    \caption{}
  \end{subfigure}
  \caption{
  Example 3 -- (a) One hundred sample realizations of the beam stiffness random field, and (b) Eigenvalue decay of the covariance matrix from 1000 samples of the stiffness random field.
  }
  \label{fig:beam1d-input-data}
\end{figure}

One thousand realizations of the stiffness random fields are simulated and 700 are used for training, while the remaining 300 are used as test data, unseen during training. The one-dimensional spatial domain is discretized into 51 equidistant points.
The MSE convergence rate of the proposed algorithm with respect to the number of identified bases is shown in \cref{fig:beam1d-err-disc-algo}(a-b). In \cref{fig:beam1d-err-disc-algo}(a), the convergence of the discrete algorithm on the training data is illustrated, while \cref{fig:beam1d-err-disc-algo}(b) shows the convergence from both the training and test data after learning the continuous neural network basis functions showing modest training loss.  
This discrepancy can be reduced by training the neural network for longer or by making the network more expressive. However, this may affect their generalization capability for unseen data and increase the test error, reflecting the bias-variance tradeoff. 
In \cref{fig:beam1d-err-disc-algo}(b), we observe strong generalization capability up to the third term, where the training and test errors are nearly identical. However, adding a fourth term increases the generalization gap. This is expected because increasing model complexity also tends to increase the generalization gap. The distribution of errors between the proposed model predictions (with four terms), both in training and testing, are plotted in \cref{fig:beam1d-err-disc-algo}(c). 
Mild overfitting is evident but the errors remain sufficiently small that prediction accuracy is satisfactory.

\begin{figure}[!ht]
  \centering
  \begin{subfigure}[b]{0.33\textwidth}
    \includegraphics[width=\textwidth]{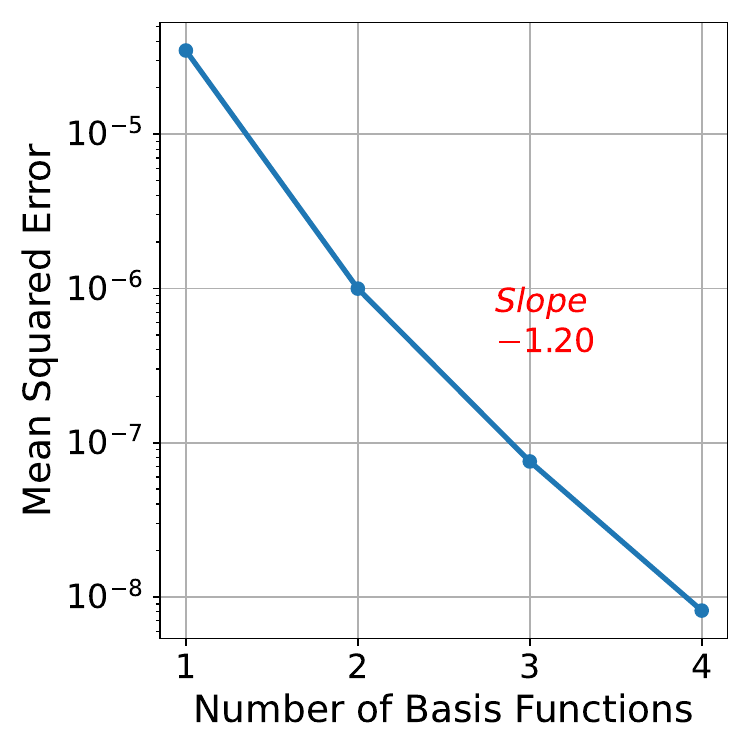}
    \caption{}
  \end{subfigure}
  \begin{subfigure}[b]{0.33\textwidth}
    \includegraphics[width=\textwidth]{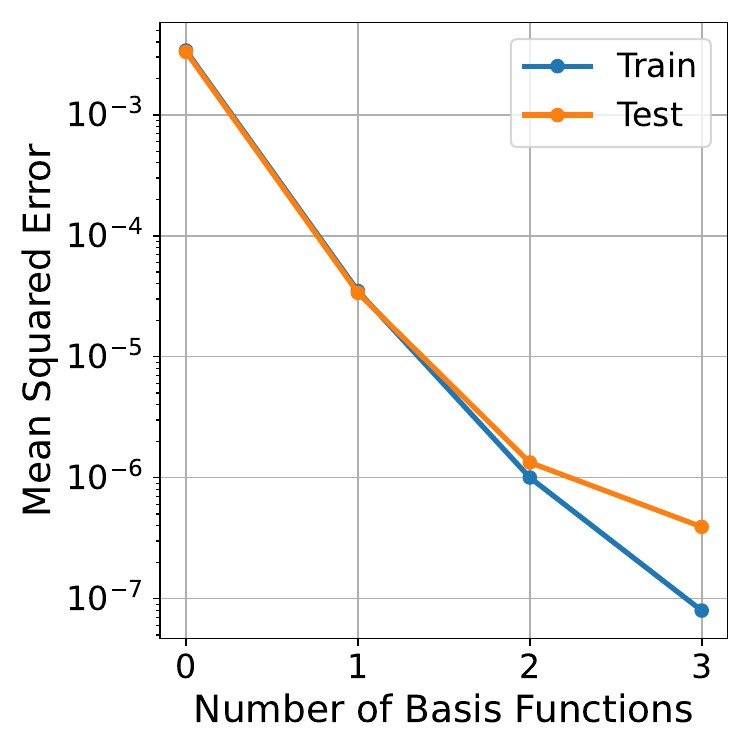}
    \caption{}
  \end{subfigure}
    \begin{subfigure}[b]{0.33\textwidth}
    \includegraphics[width=\textwidth]{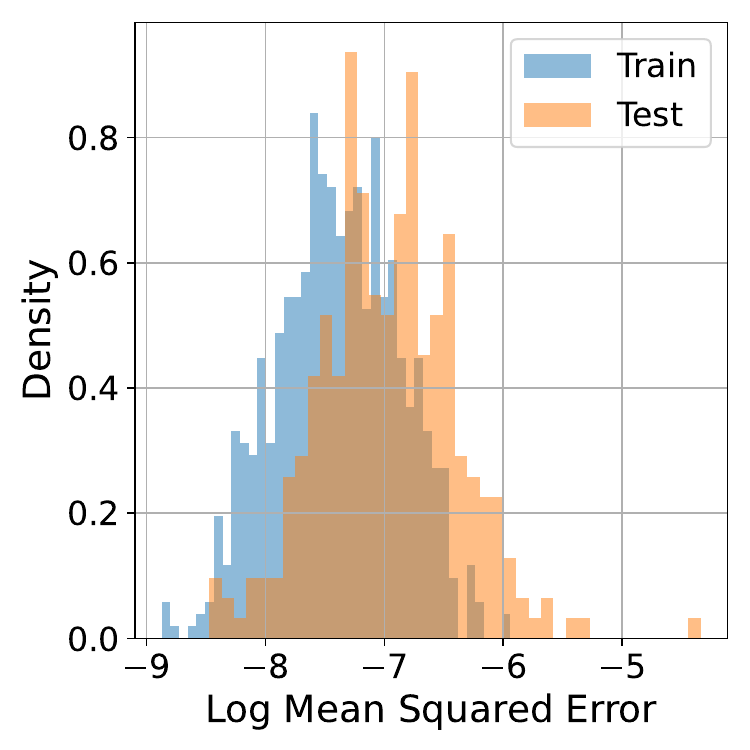}
    \caption{}
  \end{subfigure}
  \caption{
    Example 3 -- (a) Mean squared error between the data and the learned discrete spectral expansion via Algorithm~\ref{algo:SSDL-disc-cont} for increasing number of basis functions. The slope indicates the rate of error decay on a logarithmic scale. (b) The same error after training the continuous neural network basis functions.  (c) The distribution of error between the data and the learned model for training and test data.
  }
  \label{fig:beam1d-err-disc-algo}
\end{figure}

Since the stochastic basis functions are seven-dimensional (they do not employ tensor products of one-dimensional basis functions), it is not possible to plot these learned stochastic basis functions usefully in 3D. The learned deterministic basis functions are shown in \cref{fig:beam1d-det-basis-disc-algo}. Note that most of the contribution comes from the zeroth and first basis functions, as the scale of the other bases is very small compared to these basis functions. As expected, these learned basis functions resemble the expected scaled bending modes of the beam structure.

\begin{figure}[!ht]
  \centering
    \includegraphics[width=.4\textwidth]{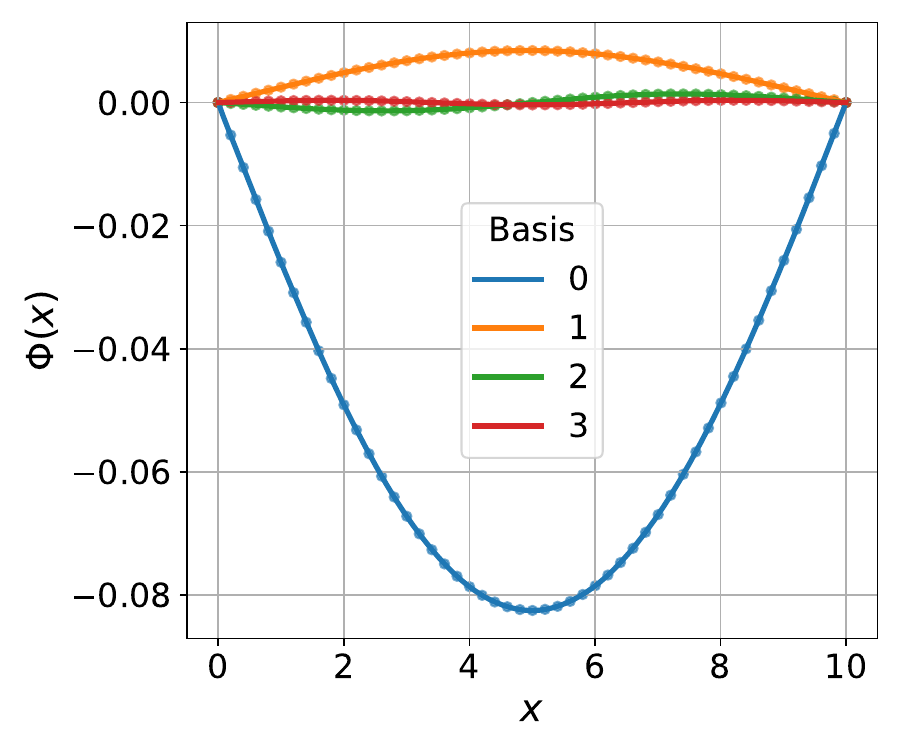}
  \caption{
  Example 3 -- Learned deterministic basis functions for the Euler-Bernoulli beam. The dots represent the basis functions learned from Algorithm~\ref{algo:SSDL-disc-cont}, while the solid lines show the corresponding neural network basis functions.
  }
  \label{fig:beam1d-det-basis-disc-algo}
\end{figure}

The direct mean and variance estimates from 
\cref{eqn:var_SSNO,eqn:mean_SSNO} are presented in \cref{fig:beamt1d-basis-mean-var-disc-algo}. These plots show good agreement between the model estimation and the ground truth obtained from Monte Carlo simulation with $10^5$ samples.

\begin{figure}[!ht]
  \centering
  \begin{subfigure}[b]{0.4\textwidth}
    \includegraphics[width=\textwidth]{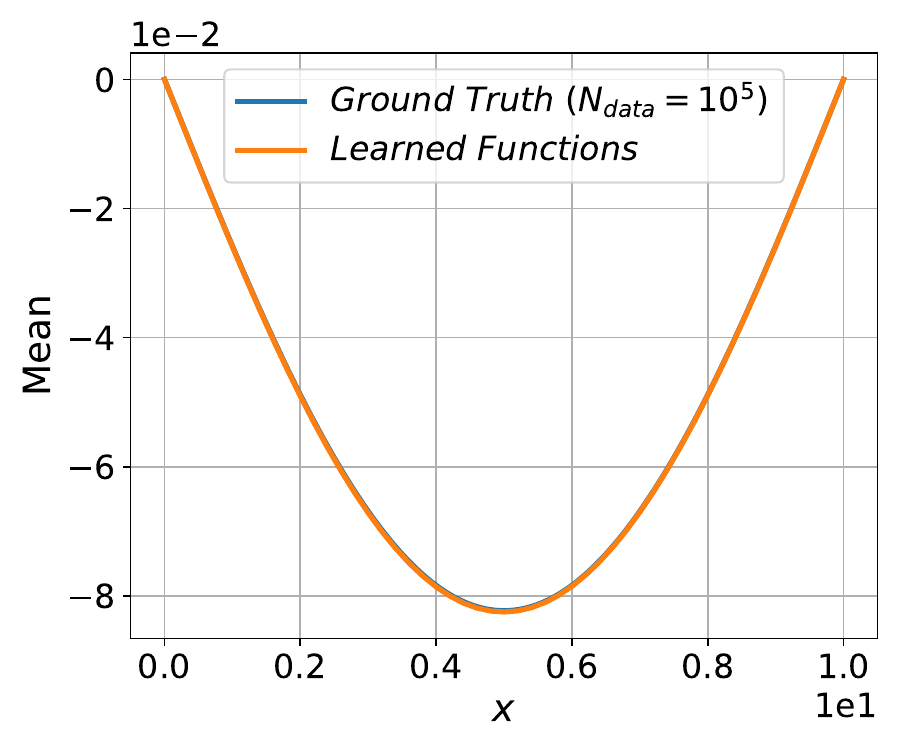}
    \caption{}
  \end{subfigure}
    \begin{subfigure}[b]{0.4\textwidth}
    \includegraphics[width=\textwidth]{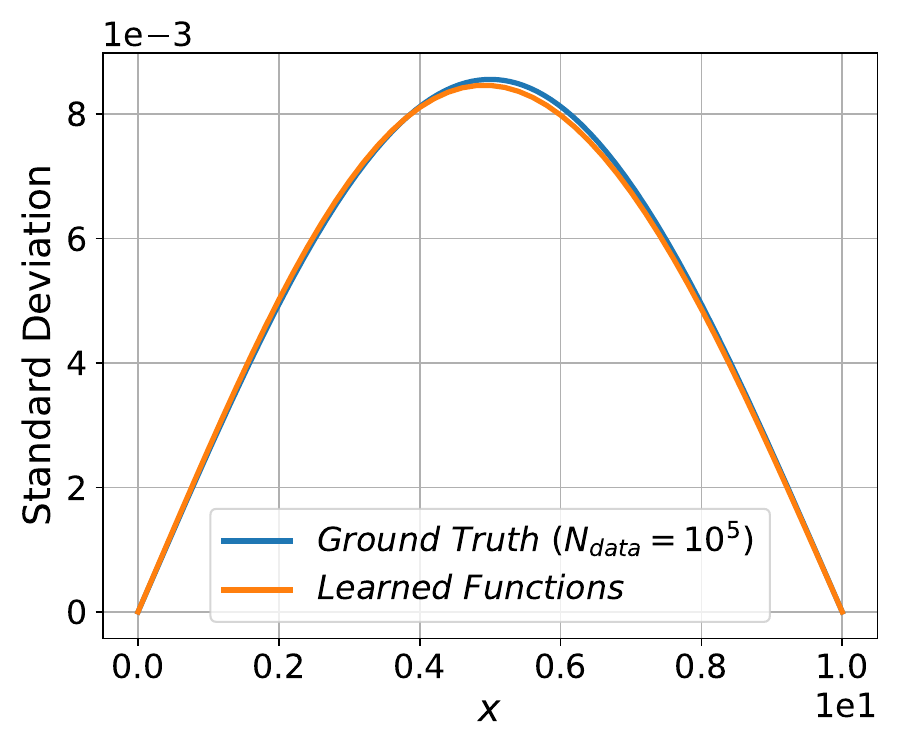}
    \caption{}
  \end{subfigure}
  \caption{
  Example 3 -- (a) Mean and (b) standard deviation field estimates obtained directly from the coefficients of the learned model. The ground truth corresponds to Monte Carlo estimation using $10^5$ realizations.
  }
  \label{fig:beamt1d-basis-mean-var-disc-algo}
\end{figure}

\subsubsection{Comparison with Polynomial Chaos Expansion}
In a non-intrusive manner, one may use classical PCE as a surrogate model in a similar way to the proposed method. For comparison, we utilize the same data used for the proposed method to train different PCEs with different number of basis functions. Following the strategy in \cite{novak2024physics, sharma2024physics}, we build the PCE models by considering the product of deterministic and random variables as a new random vector, where the components are independent and their marginals are known. Specifically, the deterministic variable is treated as a random variable with a uniform distribution, while the stochastic variables, as mentioned, follow a standard normal distribution. For this purpose, we use the PCE implemented in the UQpy package \cite{olivier2020uqpy,tsapetis2023uqpy}, where the multivariate orthogonal basis functions are constructed from the tensor product of orthogonal univariate basis functions. The coefficients are then determined via least squares regression on the training data.

The convergence of PCE with increasing number of basis functions (or equivalently, the maximum polynomial order of the univariate basis functions) is shown in \cref{fig:beamt1d-PCE-convg}. To achieve a similar training error as neural chaos, the PCE requires 5,005 terms in the spectral expansion (equal to the number of basis functions), while neural chaos uses only 4 terms in the spectral expansion (equivalent to 8 basis functions -- 4 deterministic and 4 stochastic). Note that this comparison is based on the number of spectral terms (or basis functions) needed to achieve a similar error. One might argue that it is more appropriate to count the number of unknown parameters than the number of unknown functions. In that case, the total number of parameters for the neural network basis functions may exceed the total number of PCE coefficients. 
It would be an interesting topic for future research to compare, in terms of flexibility, accuracy, and training cost, the use of a small set of highly parameterized basis functions (neural chaos) versus the use of a large set of basis functions that each have very few parameters (PCE).
\begin{figure}[!ht]
  \centering
    \includegraphics[width=.4\textwidth]{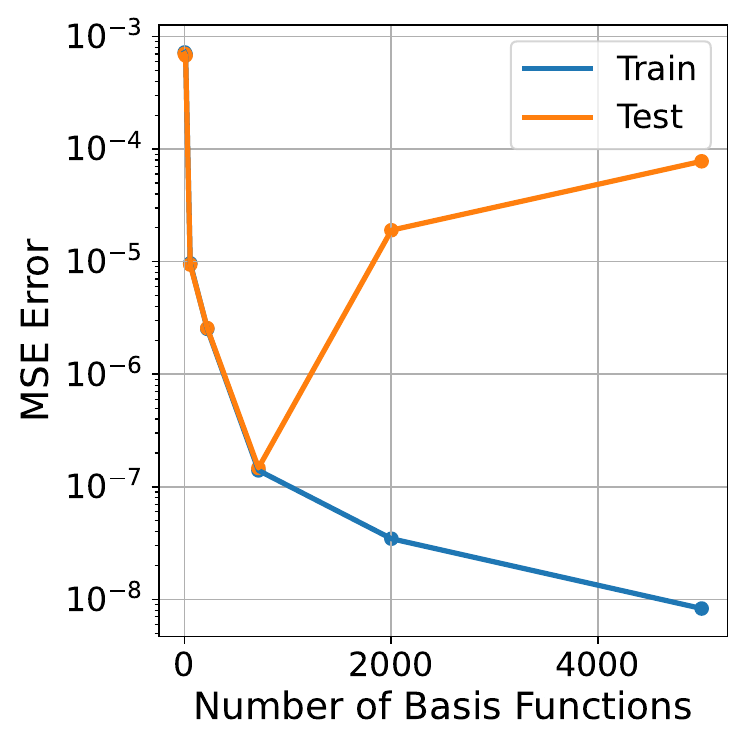}
  \caption{
    Example 3 -- Convergence of classical PCE for training and test data versus the total number of basis functions used for training the surrogate mode. The total number of multivariate PCE basis functions are 1, 10, 55, 220, 715, 2002, and 5005 when the maximum degree of the univariate basis functions are 0, 1, 2, 3, 4, 5, and 6 respectively.
  }
  \label{fig:beamt1d-PCE-convg}
\end{figure}

Apart from the number of functions or parameters, we observe an interesting difference between neural chaos and PCE in this example. PCE demonstrates perfect generalization up to 220 basis functions (polynomial order 4), but its generalization gap suddenly becomes large, and the test error increases as more basis functions are added. This is due to the well-known instability (e.g., overfiting, bad conditioned design matrix, etc.) associated with fitting high-order polynomials for regression -- as observed, for example, in Runge's phenomenon \cite{runge1901empirische}. Meanwhile, as shown in \cref{fig:beam1d-err-disc-algo}(b), the test error in neural chaos continues to decrease, and the generalization gap increases smoothly without a reduction in test accuracy.

\subsection{Example 4: Nonlinear SDE with a 13D random variable in 2D spatial domain}
\label{sec:ex-heat-2d}
The next problem extends Example \ref{sec:ex-heat1D} to a nonlinear elliptic equation in a 2D spatial domain with a higher-dimensional stochastic variable. We consider the heat conduction problem under homogeneous boundary conditions as follows:
\begin{equation}
    \nabla \cdot
    \left(
        -(1 + \frac{u^2}{2}) \nabla u
    \right)
    =
    100 \sin(5x_1) \sin(4x_2) + f(\boldsymbol{x}, \boldsymbol{\xi}), \boldsymbol{x} \in [0,1]^2,
\end{equation}
where the heat source $f(\boldsymbol{x}, \boldsymbol{\xi})$ is a Gaussian random field, $f\sim \mathcal{GP}(0, \text{Cov}(\boldsymbol{x}_i, \boldsymbol{x}_j; l_c=0.2))$, with the following covariance function:
\begin{equation}
\text{Cov}(\boldsymbol{x}_i, \boldsymbol{x}_j; l_c)
=
\exp
\left(
    -\frac{\|\boldsymbol{x}_i - \boldsymbol{x}_j \|_2^2}{2l_c^2}
\right).
\end{equation}
For illustration, two realizations of this random field are shown in \cref{fig:heat2D-soure-data}, which are discretized on a two-dimensional grid of size $20 \times 20$. From the decay of the eigenvalues of the covariance matrix from 5000 realizations of the field shown in \cref{fig:heat2D-kl-coeffs}, we choose the first 13 coefficients of the KL decomposition as the effective dimension of the input random field.

\begin{figure}[!ht]
  \centering
  \begin{subfigure}[b]{0.4\textwidth}
    \includegraphics[width=\textwidth]{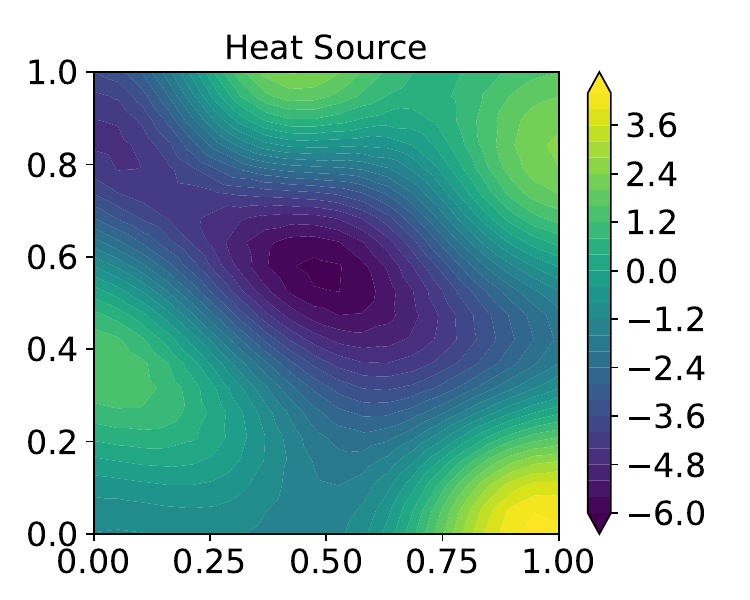}
  \end{subfigure}
    \begin{subfigure}[b]{0.4\textwidth}
    \includegraphics[width=\textwidth]{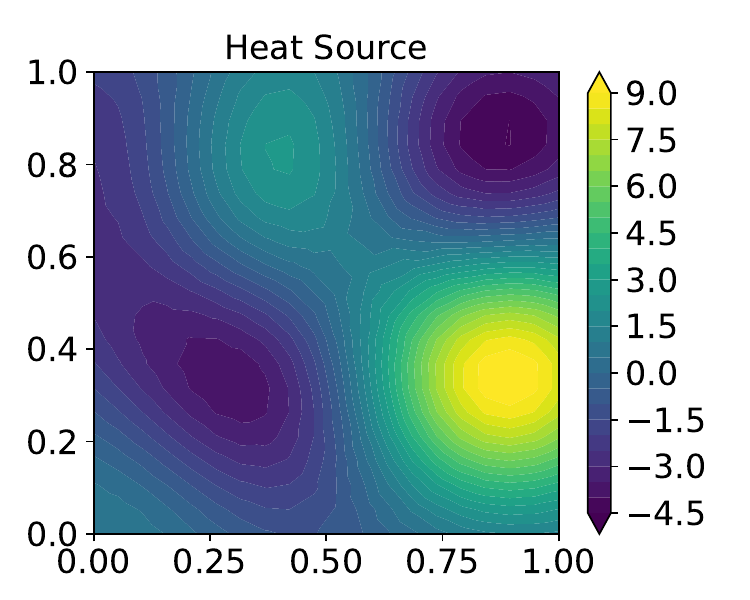}
  \end{subfigure}
  \caption{
  Example 4 -- Two random realizations of the heat source random field.
  }
  \label{fig:heat2D-soure-data}
\end{figure}

\begin{figure}[!ht]
  \centering
    \includegraphics[width=.4\textwidth]{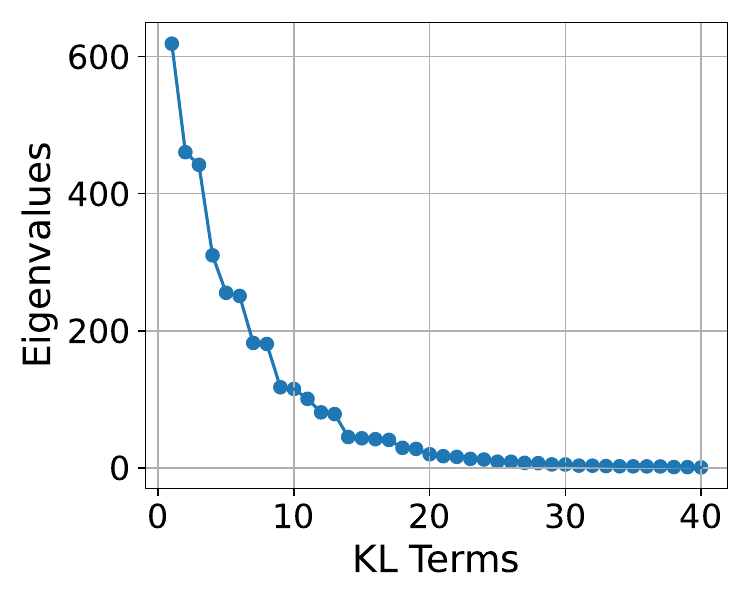}
  \caption{
  Example 4 -- Eigenvalues of the heat source random field obtained from 5000 realizations. From this decay, we choose an effective dimension of 13 random variables.
  }
  \label{fig:heat2D-kl-coeffs}
\end{figure}

The error convergence for increasing number of basis functions is shown in \cref{fig:heat2d-err-disc-algo}(a).  We observe that the rate of convergence in this high-dimensional problem is reduced compared to previous examples, requiring a larger number of basis functions, although it is still satisfactory. This could be attributed to the high dimensionality of the problem, the nonlinearity, or both. The distribution of errors between the trained model's predictions and the data for 4000 training realizations and 1000 test realizations is shown in \cref{fig:heat2d-err-disc-algo}(b). Minor overfitting is evident, as the error distribution is shifted slightly to the right compared to that of the training set, but overall the model generalizes well.
The leading learned deterministic basis functions are depicted in \cref{fig:heat2D-det-basis-disc-algo}. As we can see, in this 2D spatial domain, a significantly higher number of basis functions is required to achieve reasonable prediction accuracy compared to the previous 1D spatial domain, where fewer than four bases were sufficient.

\begin{figure}[h]
  \centering
  \begin{subfigure}[b]{0.355\textwidth}
    \includegraphics[width=\textwidth]{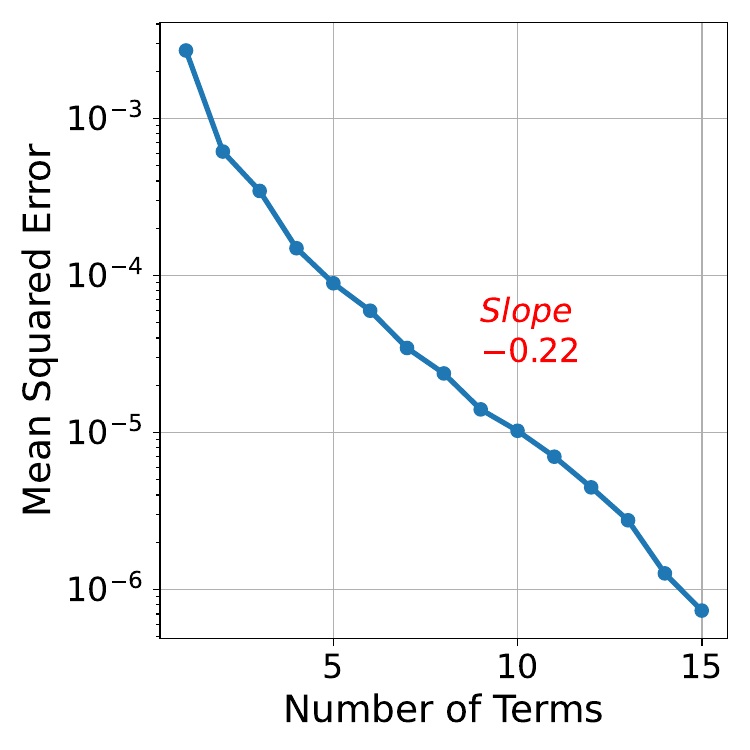}
    \caption{}
  \end{subfigure}
    \begin{subfigure}[b]{0.35\textwidth}
    \includegraphics[width=\textwidth]{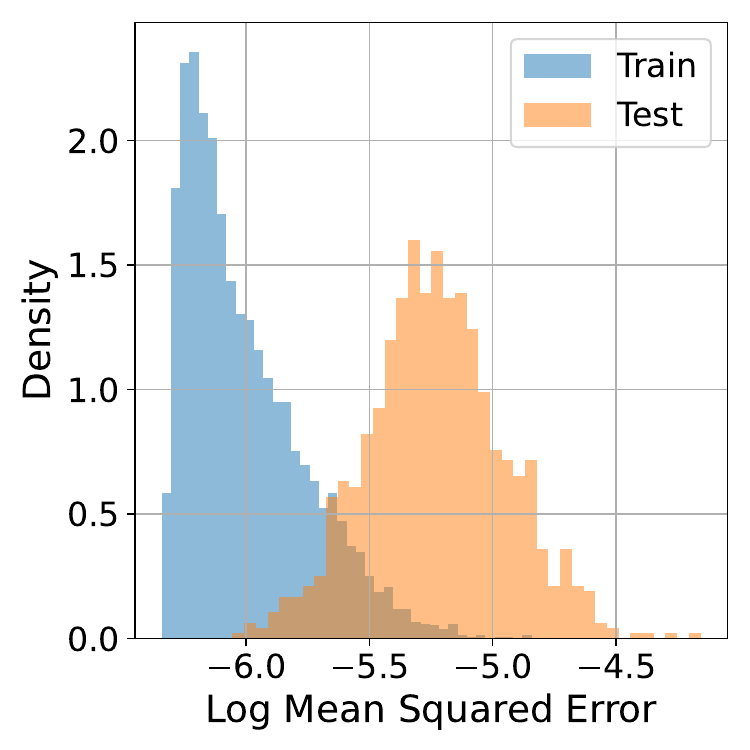}
    \caption{}
  \end{subfigure}
  \caption{
  Example 4 -- (a) Mean squared error between the data and the learned spectral expansion via Algorithm~\ref{algo:SSDL-disc-cont} for increasing number of basis functions.  The slope indicates the rate of error decay on a logarithmic scale. (b) The distribution of error between the data and the learned model. The training data includes 4000 realizations, while the test data includes 1000 realizations.
  }
  \label{fig:heat2d-err-disc-algo}
\end{figure}

\begin{figure}[!ht]
  \centering
    \includegraphics[width=.7\textwidth]{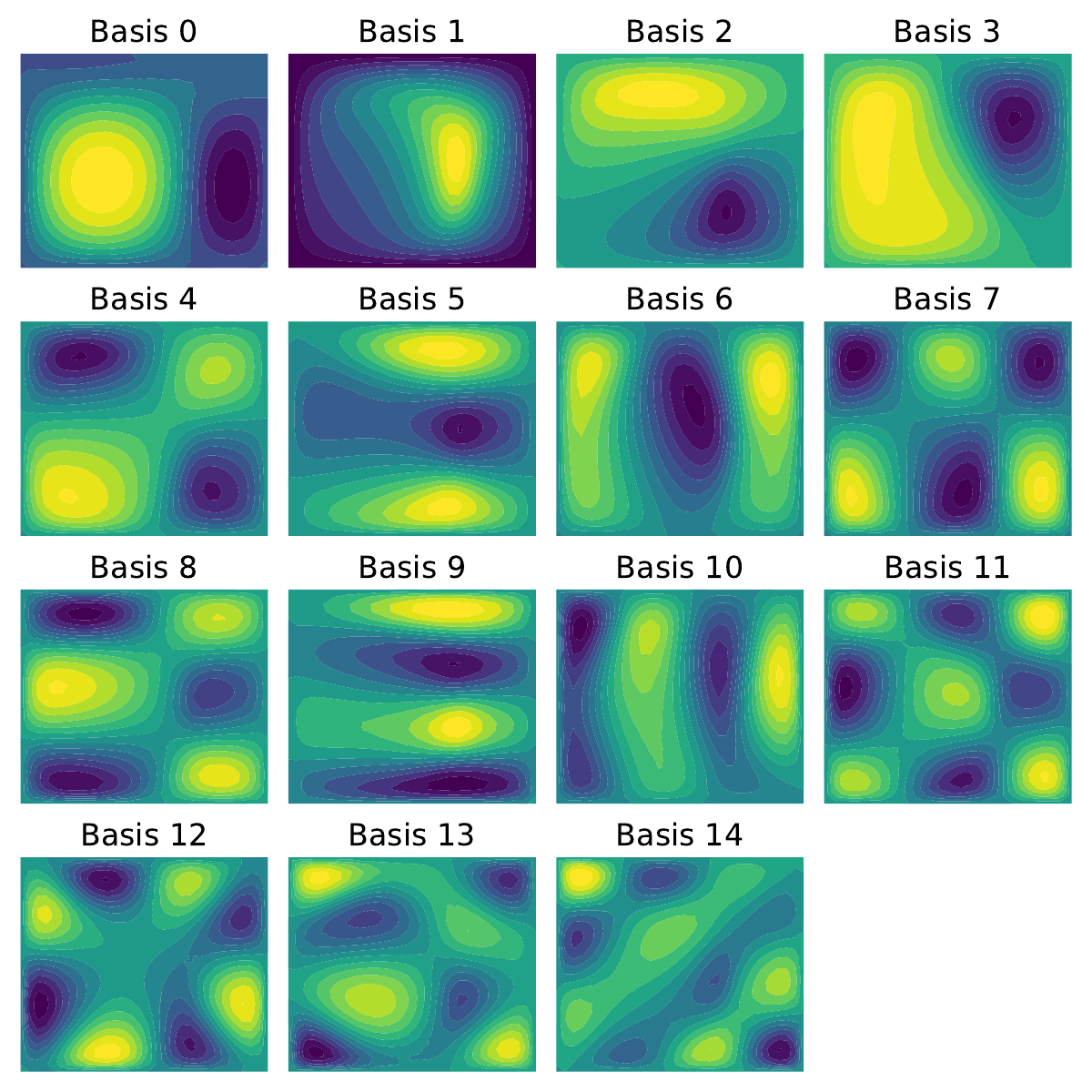}
  \caption{
  Example 4 -- Learned deterministic basis functions via Algorithm~\ref{algo:SSDL-cont}.
  }
  \label{fig:heat2D-det-basis-disc-algo}
\end{figure}

For comparison, \cref{fig:heat2D-mean-std-est} presents the mean and standard deviation fields estimated using Monte Carlo simulations with $10^5$ samples with those derived from the learned basis functions, demonstrating very good agreement.

\begin{figure}[!ht]
  \centering
  \begin{subfigure}[b]{0.6\textwidth}
    \includegraphics[width=\textwidth]{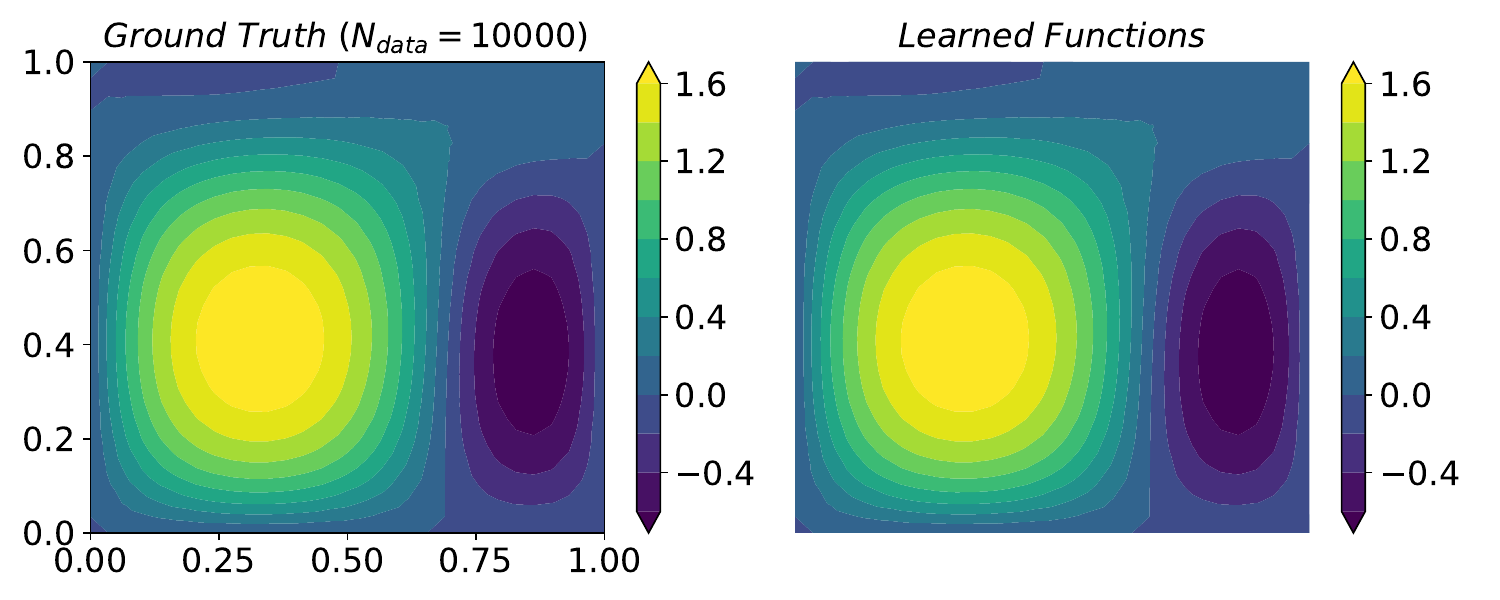}
    \caption{}
  \end{subfigure}
    \begin{subfigure}[b]{0.6\textwidth}
    \includegraphics[width=\textwidth]{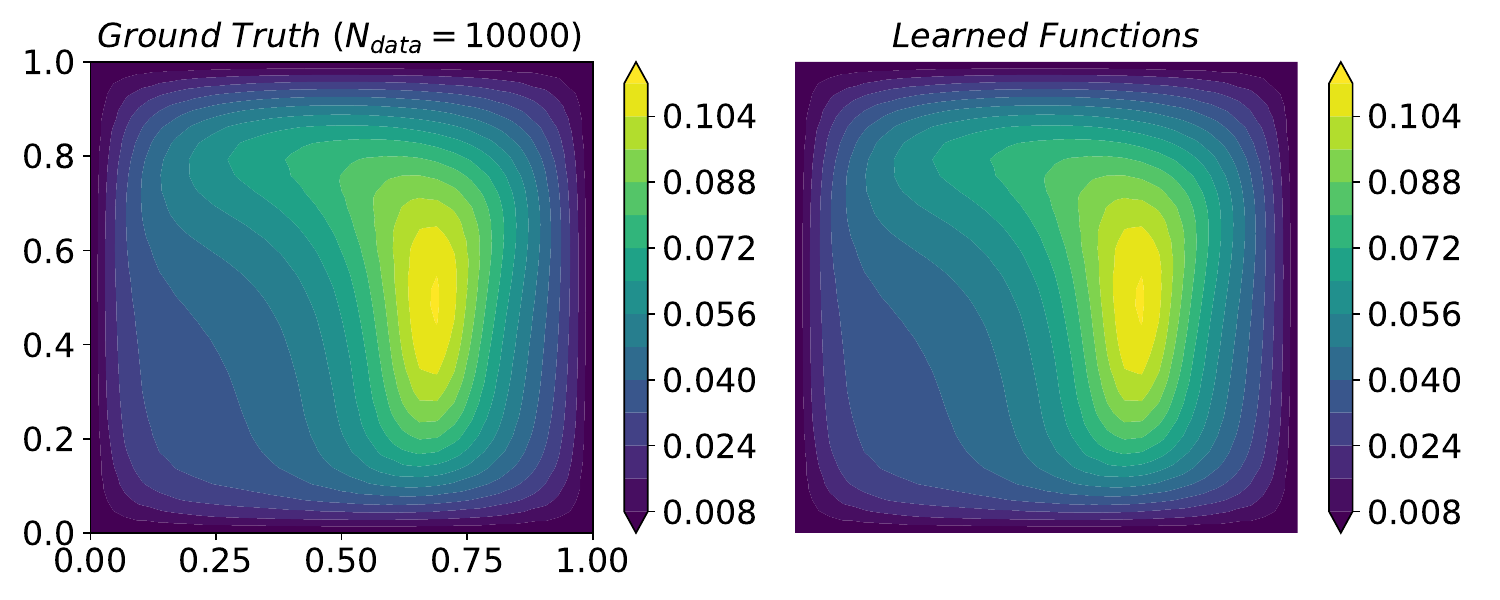}
    \caption{}
  \end{subfigure}
  \caption{
  Example 4 -- Estimated (a) Mean and (b) standard deviation fields. The ground truth plots on the left correspond to Monte Carlo estimation using $10^4$ realizations. The plots on the right compute the mean and standard deviation directly from the neural chaos expansion.}
  \label{fig:heat2D-mean-std-est}
\end{figure}

\subsection{Example 5: SDE with  dependent random variables}
\label{sec:ex-dependence}
In this problem, we demonstrate the application of the proposed method in handling problems with \textit{dependent} random variables. The problem setup is inspired by \cite{rahman2018polynomial}. We consider the following boundary value problem:
\begin{equation}
    - \frac{d}{d x}\left(\exp(\xi_1) \frac{d u}{d x}\right) = \exp(\xi_2); x \in [0, 1], (\xi_1, \xi_2) \sim p(\xi_1, \xi_2),
\end{equation}
where the Dirichlet boundary condition $u(0) = 0$ and the Neumann boundary condition $\exp(\xi_1) \frac{d u}{d x}(1) = 1$ are imposed, and the coefficients are random variables having joint probability distribution $p(\boldsymbol{\xi}) = p(\xi_1, \xi_2)$. We consider two cases: linear dependence and nonlinear dependence, which are imposed using a Gaussian copula and a Gumbel copula \citep{sklar1959fonctions,joe2014dependence,emile1960distributions}, respectively.

\subsubsection{Random Variables with Linear Dependence}
To train the model, 1000 realizations are sampled from a Gaussian distribution with mean $\boldsymbol{\mu}$ and covariance $\boldsymbol{\Sigma}$ as follows:
\begin{equation}
    (\xi_1, \xi_2) \sim \mathcal{N}
    \left(
        \boldsymbol{\mu} = 
        \begin{bmatrix}
            0 \\
            0
        \end{bmatrix},
        \boldsymbol{\Sigma} = 
        \begin{bmatrix}
            0.25^2 & -0.5 \times 0.25^2\\
            -0.5\times 0.25^2 & 0.25^2
        \end{bmatrix}
    \right).
\end{equation}
The training data consists of 70\% of the data randomly selected, while the remaining 30\% is reserved for unseen test cases. These samples are shown in \cref{fig:sharifRahman-guassian-stoch-basis-disc-algo} along with contours of the second and third stochastic basis functions learned using the proposed methods. The learned deterministic basis functions are plotted in \cref{fig:sharifRahman-guassian-det-basis-disc-algo}.
\begin{figure}[!ht]
  \centering
  \begin{subfigure}[b]{0.4\textwidth}
    \includegraphics[width=\textwidth]{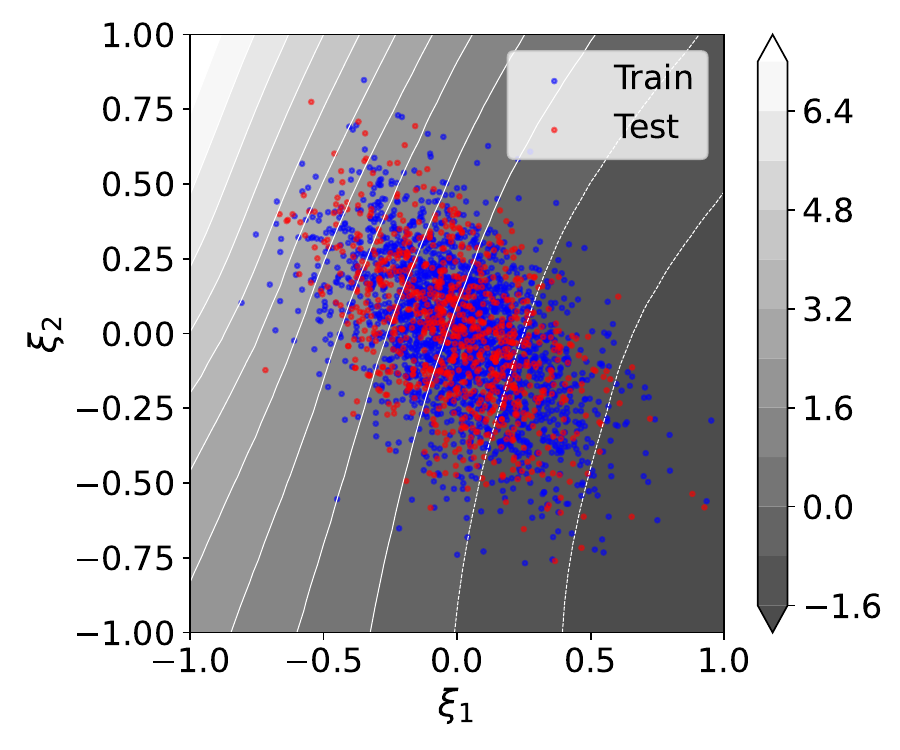}
    \caption{}
  \end{subfigure}
    \begin{subfigure}[b]{0.4\textwidth}
    \includegraphics[width=\textwidth]{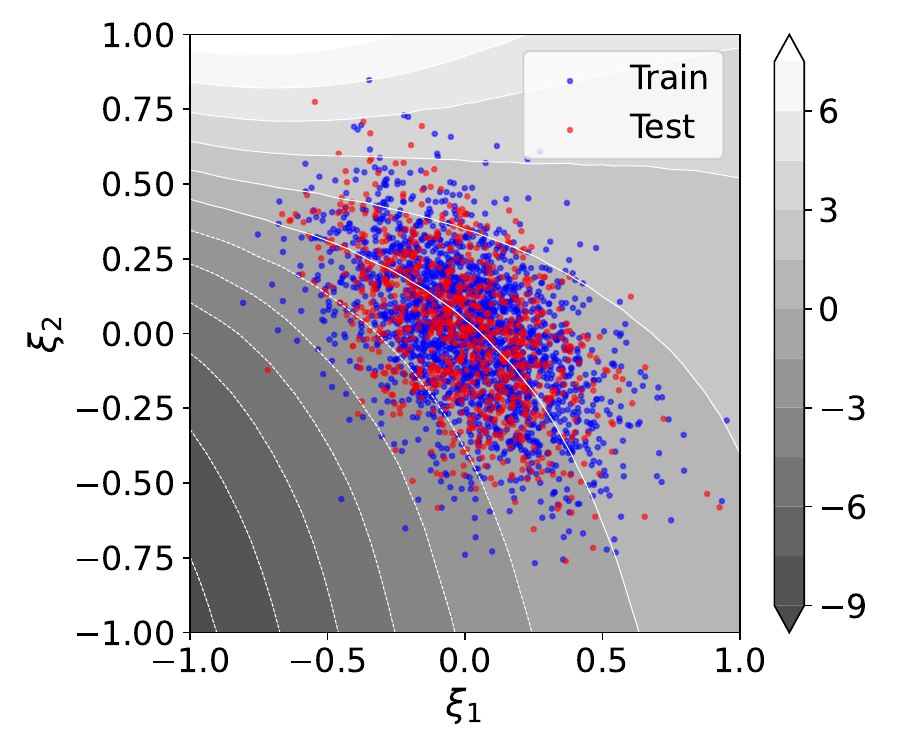}
    \caption{}
  \end{subfigure}
  \caption{
    Example 5, Linear Dependence -- Contours of the learned stochastic basis functions: (a) Second basis function, (b) Third basis function. The gray shading represents the value of the basis at each point $(\xi_1, \xi_2)$, while the blue and red dots indicate the training and test samples, respectively.
  }
  \label{fig:sharifRahman-guassian-stoch-basis-disc-algo}
\end{figure}

\begin{figure}[!ht]
  \centering
    \includegraphics[width=.4\textwidth]{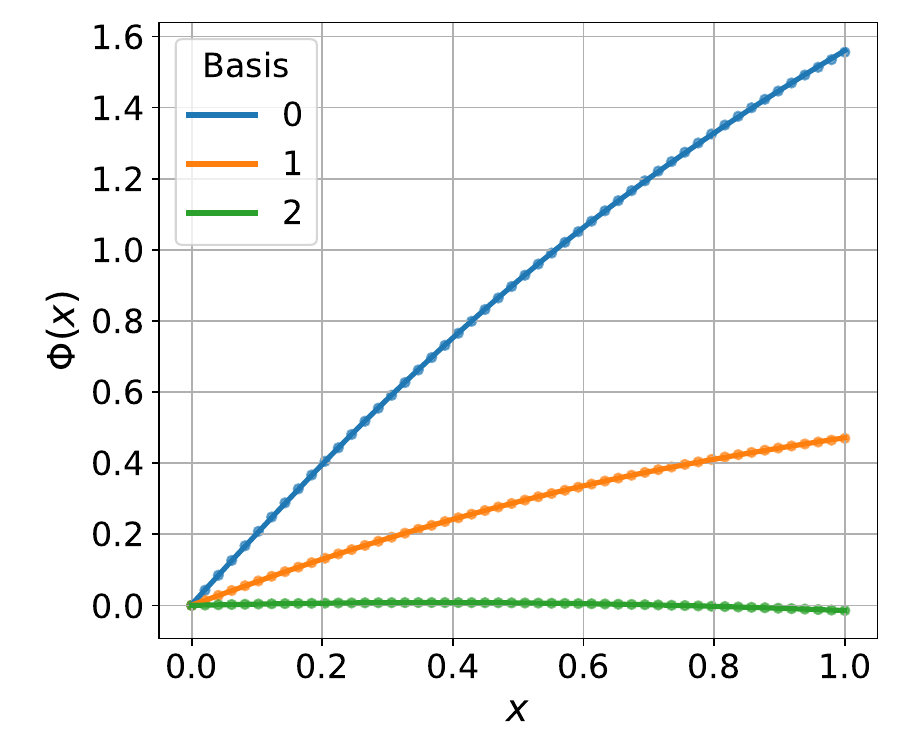}
  \caption{
    Example 5, Linear Dependence -- Learned deterministic basis functions. The dots represent the basis vectors learned from Algorithm~\ref{algo:SSDL-disc-cont}, while the solid lines correspond to their respective learned neural network basis functions.
  }
  \label{fig:sharifRahman-guassian-det-basis-disc-algo}
\end{figure}

\Cref{fig:sharifRahman-guassian-err-disc-algo} shows convergence of the error with increasing number of basis functions and the distribution of errors for the training and test sets after fitting the neural network basis functions. The discrete algorithm converges to within machine precision with only three terms, as shown in \cref{fig:sharifRahman-guassian-err-disc-algo}(a). Its continuous neural network version still performs well, with strong generalization between the training and test sets, as shown in \cref{fig:sharifRahman-guassian-err-disc-algo}(b). 
The mean and variance fields estimated directly from the spectral stochastic neural operator shown in \cref{fig:sharifRahman-guassian-mean-var-disc-algo} are in good agreement with the Monte Carlo estimates. This can be attributed to the mutual orthogonality of the basis functions with respect to the \textit{joint distribution} of the random variables.

\begin{figure}[!ht]
  \centering
  \begin{subfigure}[b]{0.4\textwidth}
    \includegraphics[width=\textwidth]{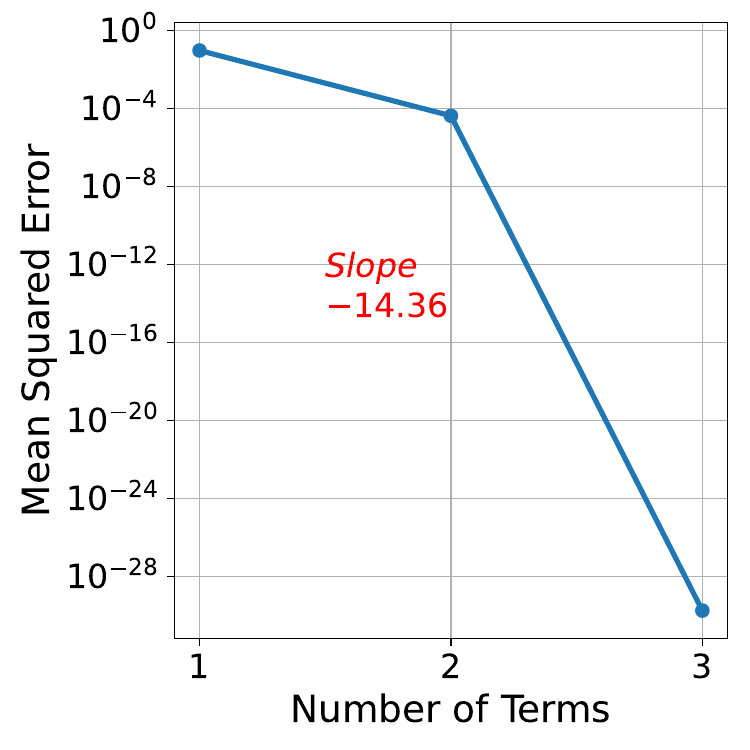}
    \caption{}
  \end{subfigure}
    \begin{subfigure}[b]{0.4\textwidth}
    \includegraphics[width=\textwidth]{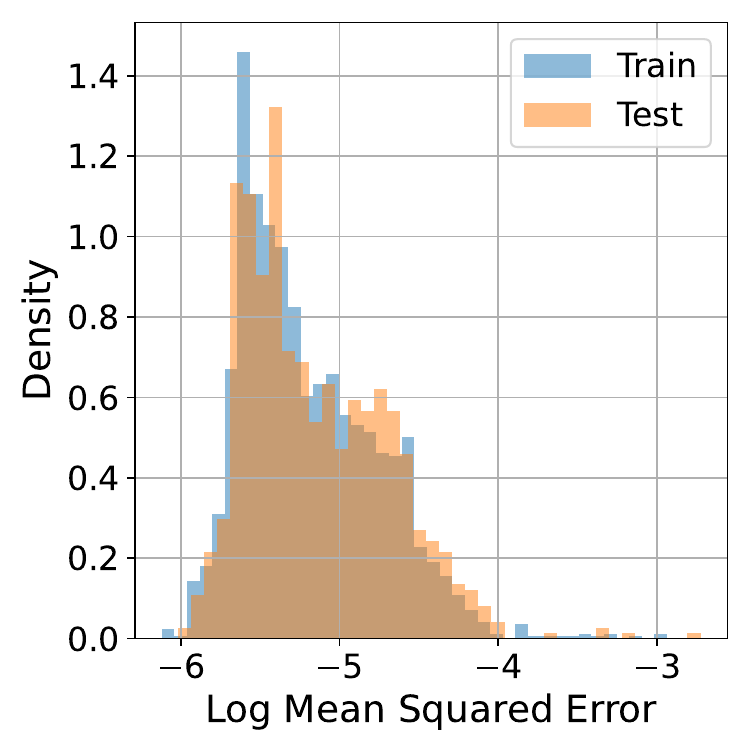}
    \caption{}
  \end{subfigure}
  \caption{
  Example 5, Linear Dependence -- (a) Mean squared error between the data and the learned spectral expansion via Algorithm~\ref{algo:SSDL-disc-cont} for increasing number of basis functions. The slope indicates the rate of error decay on a logarithmic scale. (b) The distribution of error between the data and model after training the neural network basis functions.
  }
  \label{fig:sharifRahman-guassian-err-disc-algo}
\end{figure}

\begin{figure}[!ht]
  \centering
  \begin{subfigure}[b]{0.4\textwidth}
    \includegraphics[width=\textwidth]{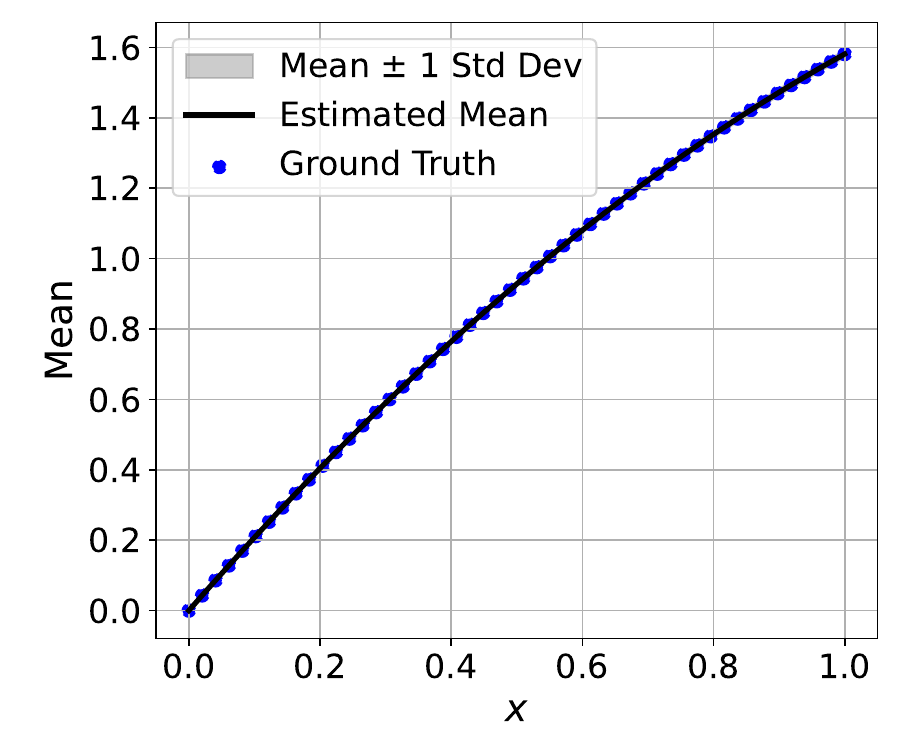}
    \caption{}
  \end{subfigure}
    \begin{subfigure}[b]{0.4\textwidth}
    \includegraphics[width=\textwidth]{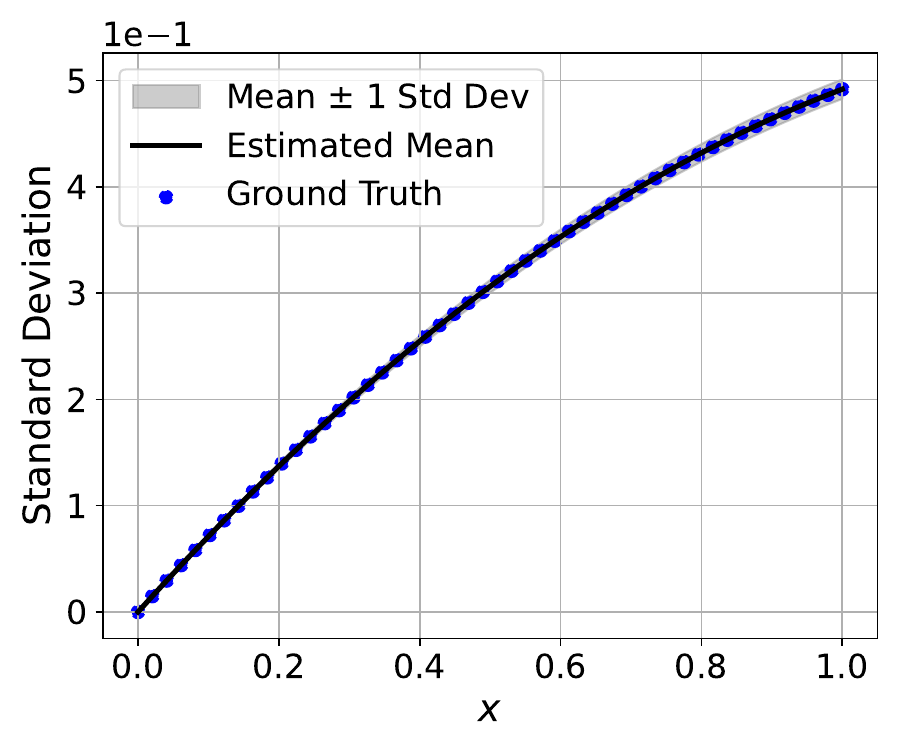}
    \caption{}
  \end{subfigure}
  \caption{
  Example 5, Linear Dependence -- (a) Mean and (b) standard deviation field estimates derived directly from the proposed spectral stochastic expansion. The plots show the mean $\pm$ one standard deviation for 100 independent training trials. The ground truth is estimated from Monte Carlo simulation using $10^5$ realizations.
  }
  \label{fig:sharifRahman-guassian-mean-var-disc-algo}
\end{figure}

\subsubsection{Random Variables with Nonlinear Dependence}
In this example, we assume that the marginal distribution of the $\xi_1$ follows a Gamma distribution (with shape and scale parameters 2 and 1, respectively) and $\xi_2$ follows a standard normal distribution, i.e., $p(\xi_1)\sim \text{Gamma}(2, 1), p(\xi_2)\sim \mathcal{N}(0, 1)$. 
Following Sklar's theorem \citep{sklar1959fonctions,joe2014dependence}, the joint cumulative distribution function (CDF) of $\xi_1$ and $\xi_2$ is expressed in terms of the Gumbel copula $C_{\tau} (u, v)$
and their marginal distributions as follows:
\begin{equation}
    F_{\xi_1, \xi_2}(y_1, y_2)
    =
    C_r(F_{\xi_1}(y_1), F_{\xi_2}(y_2))
    F_{\xi_1}(y_1)
    F_{\xi_2}(y_2);
    \quad y_1, y_2 \in \mathbb{R},
\end{equation}
where $F$ denotes the CDF, and $C_r$ represents the Gumbel copula function \cite{emile1960distributions}, defined as:
\begin{equation}
C_{\tau} (u, v)
=
\exp
\left[
-(
(-\log(u))^\tau
)
+
(-\log(v))^\tau
)
^{\frac{1}{\tau}}
\right],
\end{equation}
with parameter $\tau=2$. Note that traditional PCE with tensor product basis polynomials \textit{cannot be used} to model these nonlinearly dependent random variables. Again, 1,000 samples of this random variable are drawn and 70\% are used for training and 30\% for test data. These samples are shown along with contours of the second and third learned stochastic basis functions in \cref{fig:sharifRahman-gumble-stoch-basis-disc-algo}, with blue and red points corresponding to the training and test data, respectively. Notice the distinct difference in the form of the stochastic basis functions compared to the linear case in \cref{fig:sharifRahman-guassian-stoch-basis-disc-algo}. The learned deterministic basis functions are plotted in \cref{fig:sharifRahman-gumble-det-basis-disc-algo}. Again, we notice differences from those with linear dependence in \cref{fig:sharifRahman-guassian-det-basis-disc-algo}.
\begin{figure}[!ht]
  \centering
  \begin{subfigure}[b]{0.4\textwidth}
    \includegraphics[width=\textwidth]{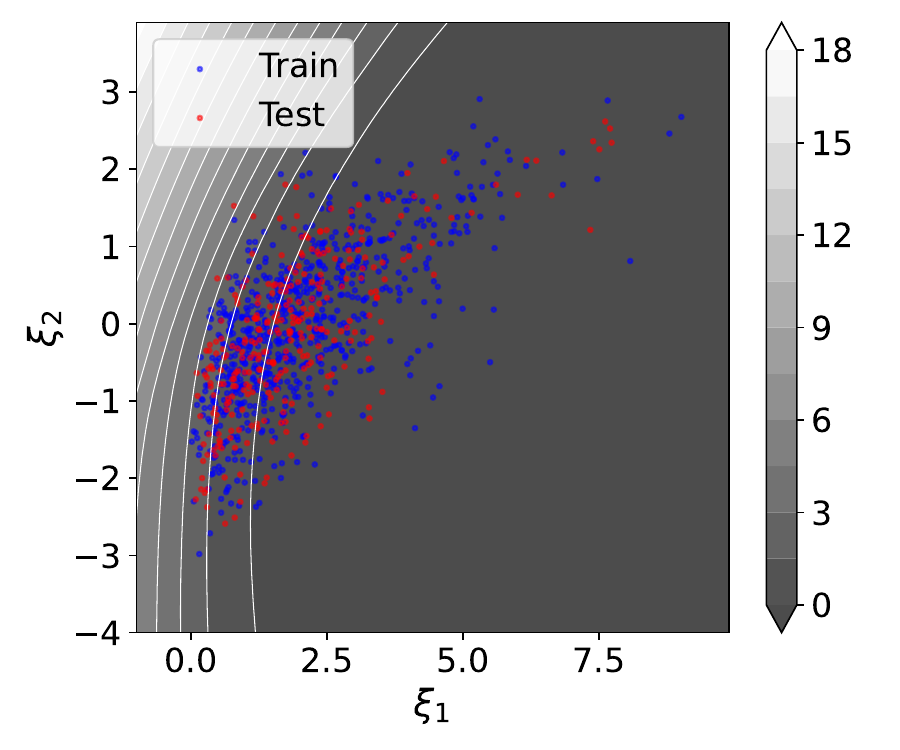}
    \caption{}
  \end{subfigure}
    \begin{subfigure}[b]{0.4\textwidth}
    \includegraphics[width=\textwidth]{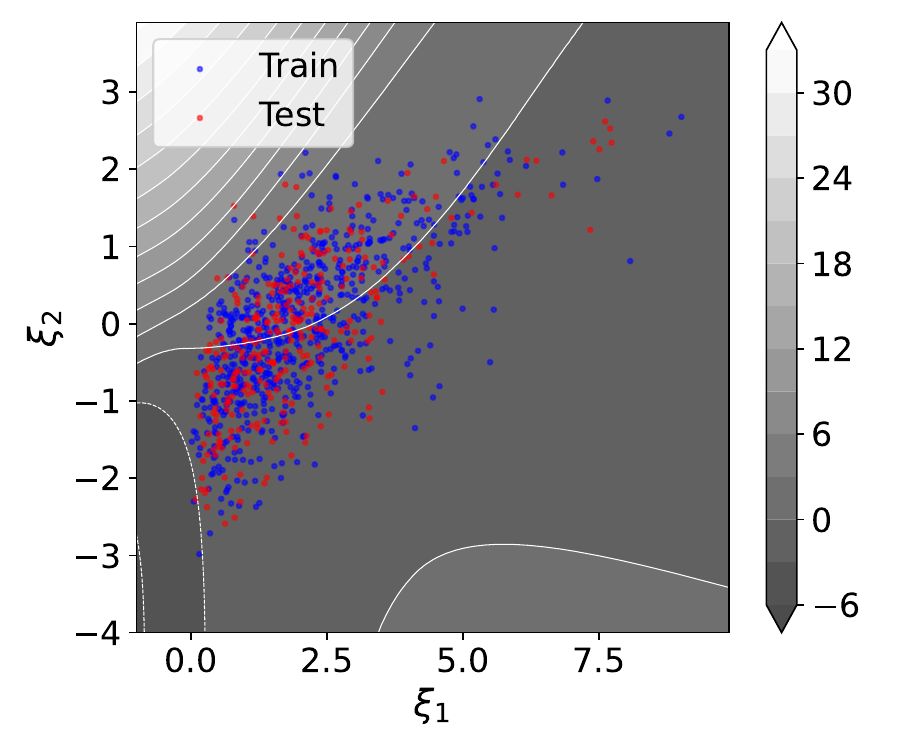}
    \caption{}
  \end{subfigure}
  \caption{
    Example 5, Nonlinear Dependence -- Contours of the learned stochastic basis functions: (a) Second basis function, (b) Third basis function. The gray shading represents the value of the basis at each point $(\xi_1, \xi_2)$, while the blue and red dots indicate the training and test samples, respectively.
  }
  \label{fig:sharifRahman-gumble-stoch-basis-disc-algo}
\end{figure}

\begin{figure}[!ht]
  \centering
    \includegraphics[width=.4\textwidth]{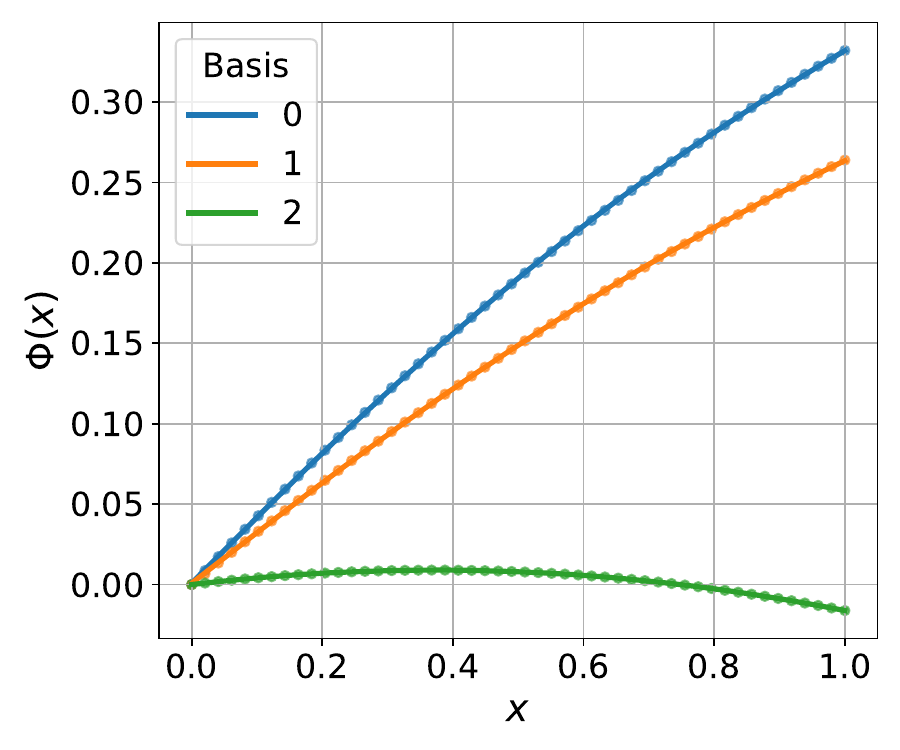}
  \caption{
    Example 5, Nonlinear Dependence -- Learned deterministic basis functions. The dots represent the basis vectors learned from Algorithm~\ref{algo:SSDL-disc-cont}, while the solid lines correspond to their respective learned neural network basis functions.
  }
  \label{fig:sharifRahman-gumble-det-basis-disc-algo}
\end{figure}

\cref{fig:sharifRahman-gumble-err-disc-algo} shows convergence of the error with increasing number of fitted basis functions and the distribution of errors for the training and test sets after fitting the neural network basis functions. Again, the discrete algorithm converges to within machine precision with only three terms, as shown in \cref{fig:sharifRahman-gumble-err-disc-algo}(a). Its continuous neural network version still performs well, with strong generalization between the training and test sets, as shown in \cref{fig:sharifRahman-gumble-err-disc-algo}(b). 
The mean and variance fields estimated directly from the spectral stochastic neural operator shown in \cref{fig:sharifRahman-gumble-mean-var-disc-algo} are in good agreement with the Monte Carlo estimates. This can be attributed to the orthogonality of the basis functions with respect to the \textit{joint distribution} of the random variables.

\begin{figure}[!ht]
  \centering
  \begin{subfigure}[b]{0.4\textwidth}
    \includegraphics[width=\textwidth]{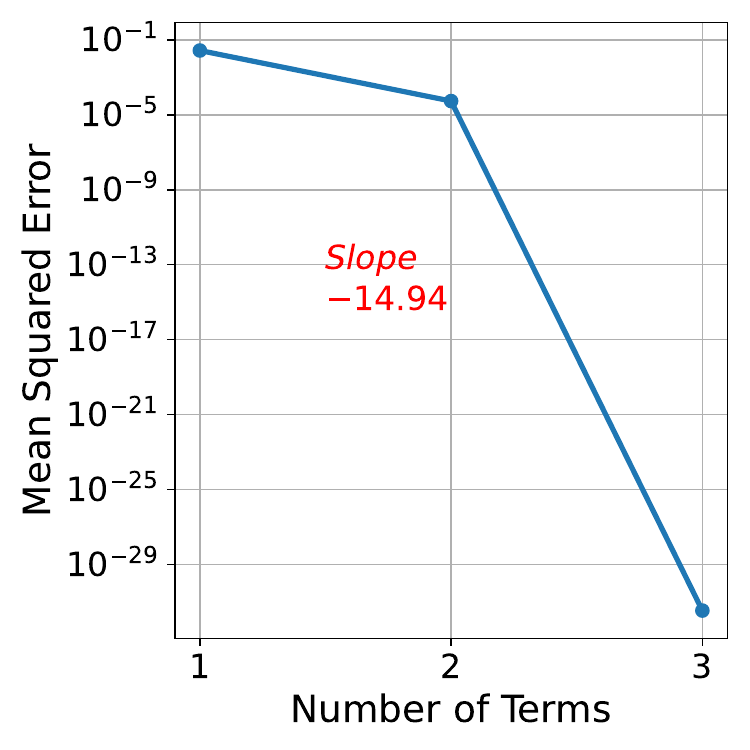}
    \caption{}
  \end{subfigure}
    \begin{subfigure}[b]{0.4\textwidth}
    \includegraphics[width=\textwidth]{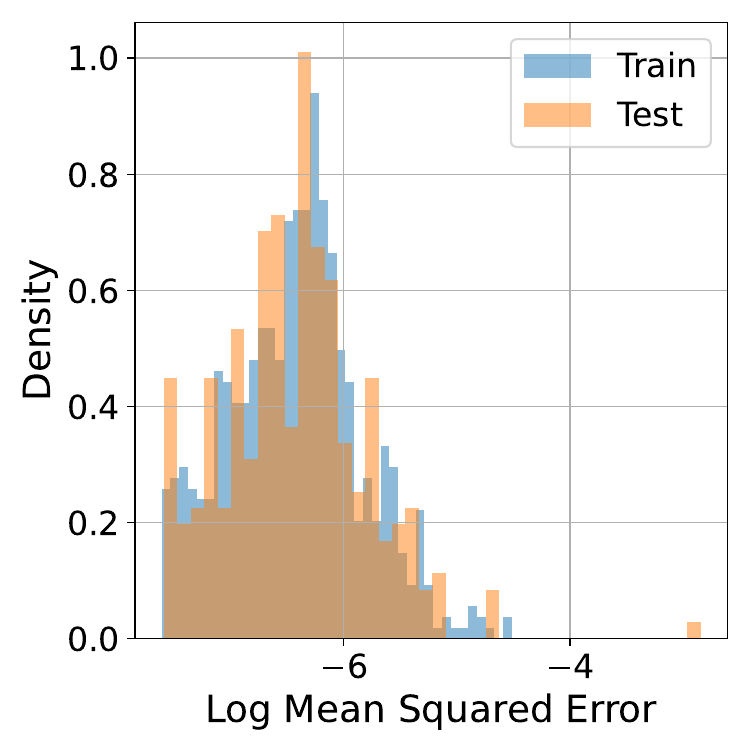}
    \caption{}
  \end{subfigure}
  \caption{
  Example 5, Nonlinear Dependence -- (a) Mean squared error between the data and the learned spectral expansion via Algorithm~\ref{algo:SSDL-disc-cont} for increasing number of basis functions. The slope indicates the rate of error decay on a logarithmic scale. (b) The distribution of error between the data and model after training the neural network basis functions.
  }
  \label{fig:sharifRahman-gumble-err-disc-algo}
\end{figure}


\begin{figure}[!ht]
  \centering
  \begin{subfigure}[b]{0.4\textwidth}
    \includegraphics[width=\textwidth]{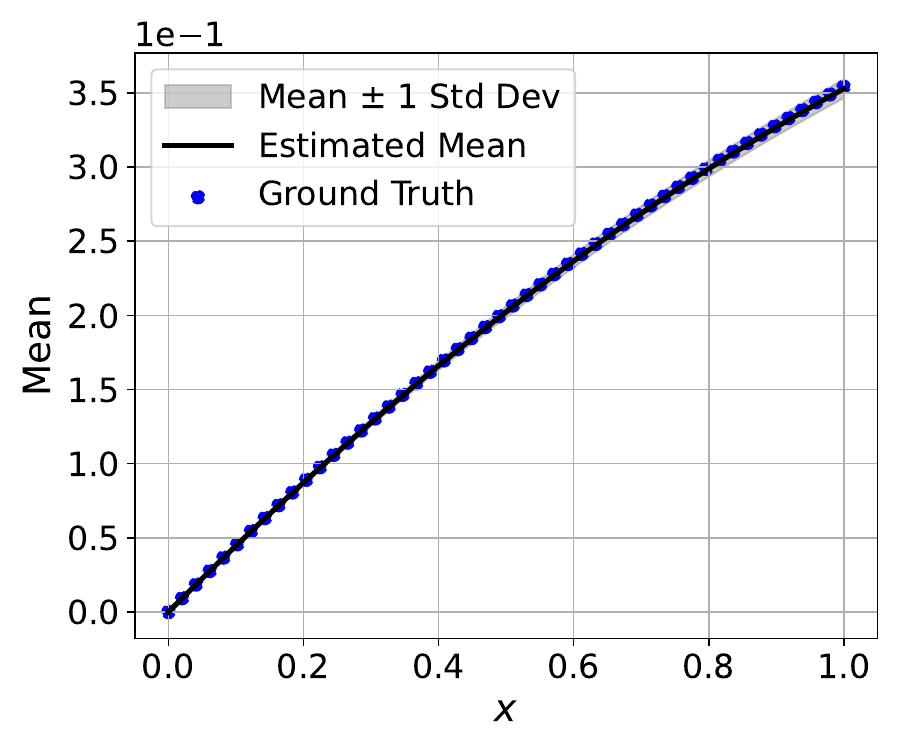}
    \caption{}
  \end{subfigure}
    \begin{subfigure}[b]{0.4\textwidth}
    \includegraphics[width=\textwidth]{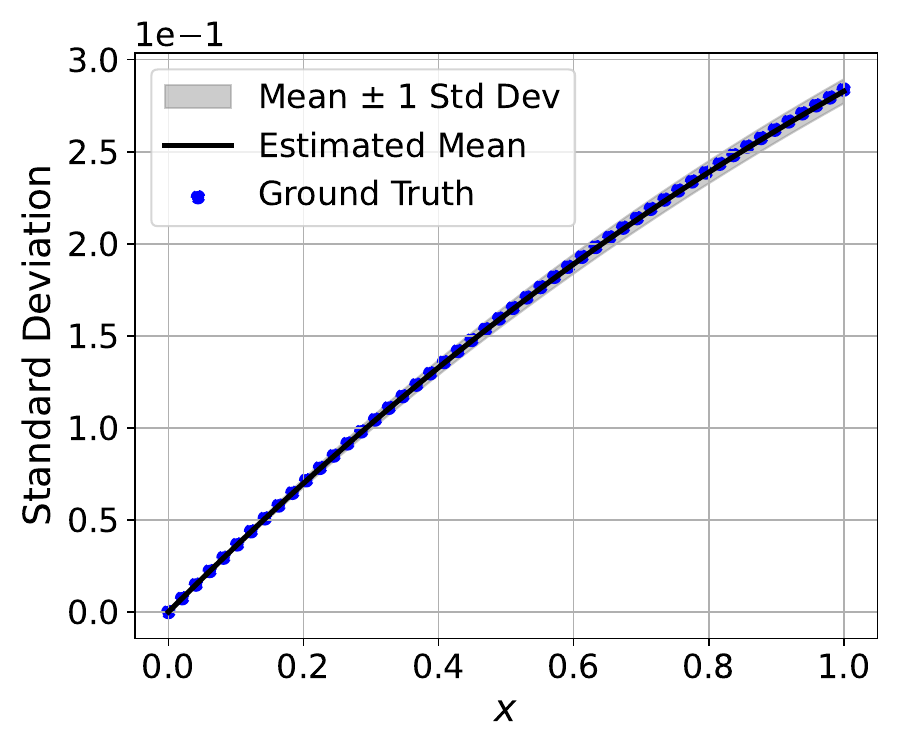}
    \caption{}
  \end{subfigure}
  \caption{
  Example 5, Nonlinear Dependence -- (a) Mean and (b) variance field estimates derived directly from the proposed spectral stochastic expansion. The plots show the mean $\pm$ one standard deviation for 100 independent training trials. The ground truth is estimated from Monte Carlo simulation using $10^5$ realizations.
  }
  \label{fig:sharifRahman-gumble-mean-var-disc-algo}
\end{figure}

\section{Conclusions}
\label{sec:conclusion}
In this work, we propose Neural Chaos, inspired by PCE, as a non-intrusive surrogate modeling approach for stochastic operator learning. To achieve the highest flexibility, instead of using pre-defined basis functions, Neural Chaos parametrizes the basis functions via neural networks and finds them. We demonstrate that this flexibility enables a significantly more compact representation than PCE for stochastic problems. In the present applications, more than 5000 PCE basis functions fail to match the representation power of Neural Chaos with fewer than 20 neural basis functions.

We propose two algorithms to adaptively learn the orthogonal basis functions with respect to the joint probability distribution of the random variables governing the stochastic processes. The optimization challenges of the fully continuous algorithm are discussed and numerically examined, while the discrete-continuous algorithm has empirically demonstrated greater flexibility, robustness, and a significantly better convergence rate.

Neural Chaos does not require any \textit{a priori} assumptions about the joint or marginal distributions of the random variables, making it directly applicable to problems with dependent random variables without modification. Furthermore, it constructs stochastic basis functions \textit{directly based on the joint distribution of the random variables}, eliminating the need for a tensor product structure or the assumption of marginal independence. In building neural chaos, we have established a connection between operator learning paradigms and spectral expansion theory, tying the two together in a theoretically grounded formulation. Neural Chaos introduces new avenues for studying and simulating stochastic operations with significantly greater flexibility compared to classical PCE, while maintaining a remarkably simple implementation.

As discussed, Nueral Chaos can directly operate on dependent random variables; however, it does not reveal which subset of these variables are the most significant in governing the system's stochasticity or, more intriguingly, whether a lower-dimensional transformation of these variables might account for most of the system's stochasticity. Addressing these questions could lead to more efficient approaches for uncertainty quantification, reducing both data requirements and computational resources. One potential method to explore these questions is the use of conformal autoencoders \cite{kevrekidis2024conformal,bahmani2023distance,peterfreund2020local}, which disentangle data representations in a lower-dimensional space. We leave this direction for future research.

\section*{Acknowledgments}

This material is based upon work supported by the U.S. Department of Energy, Office of Science, Office of Advanced Scientific Computing Research, under Award Number DE-SC0024162.

\bibliographystyle{plain}
\bibliography{bibliography}  

\section*{Appendix}
Supplementary results are presented in Appendix \ref{appx:ablation1}, and the hyperparameters of the neural networks used are detailed in Appendix \ref{appx:hyper-params}.

\appendix
\section{Ablation Study}
\label{appx:ablation1}
In this section, we aim to study the effect of predefined basis functions compared to the learnable ones proposed here, as well as the data dependency of these learnable basis functions, within the problem setup discussed in \ref{sec:ex-xui-1d}.
\subsection{Pre-defined Deterministic Basis Functions}
Although the convergence of spectral expansion (Eq.~\ref{eq:spec-expan}) is guaranteed in the limit of an infinite series, regardless of the chosen hypothesis class for stochastic and deterministic functions, the rate of convergence—specifically, the number of terms required to reach a target approximation error—is influenced by the selected hypothesis classes. To study the impact of fixing the hypothesis class a priori, we set the deterministic functions as orthogonal cosine basis functions at different frequencies, optimizing only the stochastic basis functions in Algorithm~\ref{algo:SSDL-disc-cont}. The frequencies of the basis functions are integer values, starting at 1 and incrementing by 1 with each additional term. Since the deterministic basis is predefined, a single iteration in Algorithm~\ref{algo:MulDecomp} suffices to obtain the corresponding stochastic basis function.

We use the data generated from the Gamma distribution in Section \ref{sec:ex-xui-1d}. The error distribution of the proposed algorithm, based on the number of terms, is shown in \cref{fig:1d-fix-det-basis}(a). As observed, the rate of error reduction decreases significantly compared to \cref{fig:ode1d-convg-disc-algo}(c). Moreover, many more terms may be required to achieve the same accuracy obtained with the learnable deterministic basis functions. Note that we could not exceed 6 basis functions, as cosine bases with higher frequencies require finer spatial domain discretization to avoid aliasing, in accordance with the Nyquist–Shannon sampling theorem. This indicates that, regardless of the flexibility of the stochastic basis functions, a biased deterministic basis may lead to poor approximation and require an unnecessarily large number of basis functions. For completeness and comparison, the variance estimated using Sobol's method is provided in \cref{fig:1d-fix-det-basis}(b), which shows a significant discrepancy compared to the ground truth.

\begin{figure}[!ht]
  \centering
  \begin{subfigure}[b]{0.4\textwidth}
    \includegraphics[width=\textwidth]{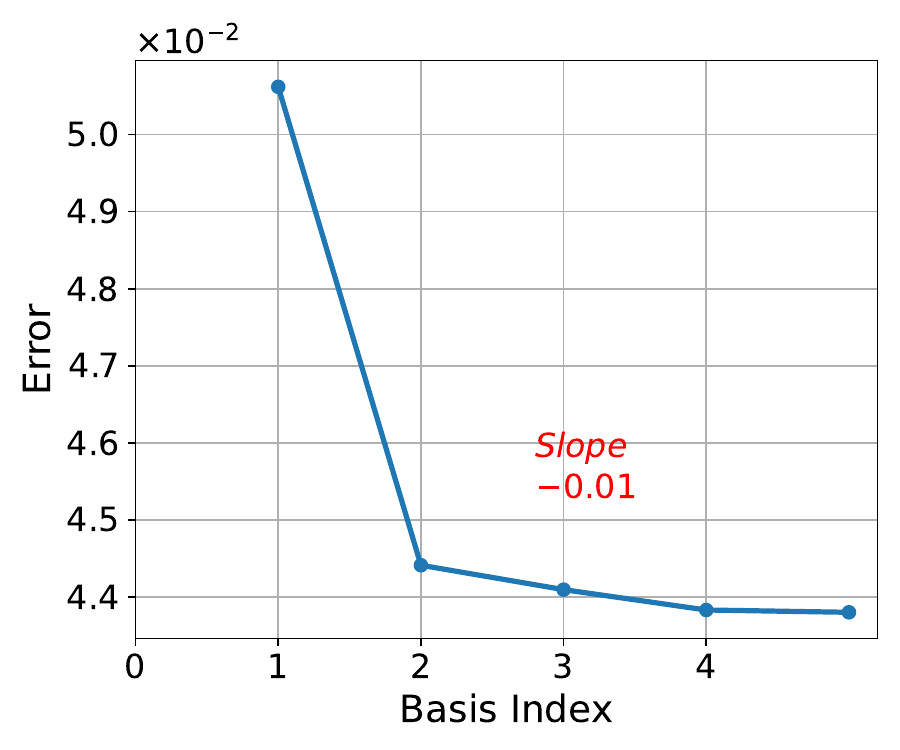}
    \caption{}
  \end{subfigure}
    \begin{subfigure}[b]{0.4\textwidth}
    \includegraphics[width=\textwidth]{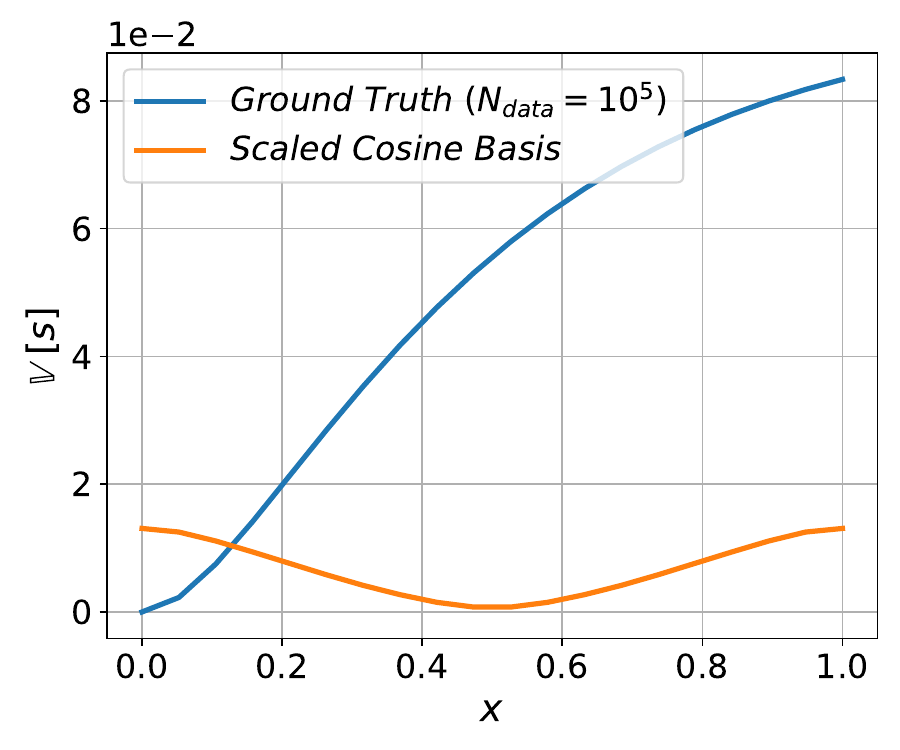}
    \caption{}
  \end{subfigure}
  \caption{
  (a) Mean squared error between the data and the learned spectral expansion via Algorithm~\ref{algo:SSDL-disc-cont}
  }
  \label{fig:1d-fix-det-basis}
\end{figure}

\subsection{Dependency on Data Size}
Due to the non-intrusive nature of the proposed approach, the quality of the approximation depends on the data's quality and quantity. To assess this, we compare the learned stochastic basis functions and the variance estimations obtained via these basis functions when the number of realizations used for training is 50 and 700, as shown in \cref{fig:1d-datasize-comp-stoch-basis,fig:1d-datasize-comp-var}. Interestingly, the overall shapes of the basis functions are similar, though they appear as stretched and scaled versions of each other. This suggests that, while the basis functions are data-dependent, their dependence may be mild. As expected, the amount of data affects the accuracy of the surrogate model and, consequently, the learned basis functions, which directly impact the variance estimation. Nonetheless, a reasonable approximation of the variance is achieved even with a model trained on only 50 realizations.

\begin{figure}[!ht]
  \centering
  \begin{subfigure}[b]{0.4\textwidth}
    \includegraphics[width=\textwidth]{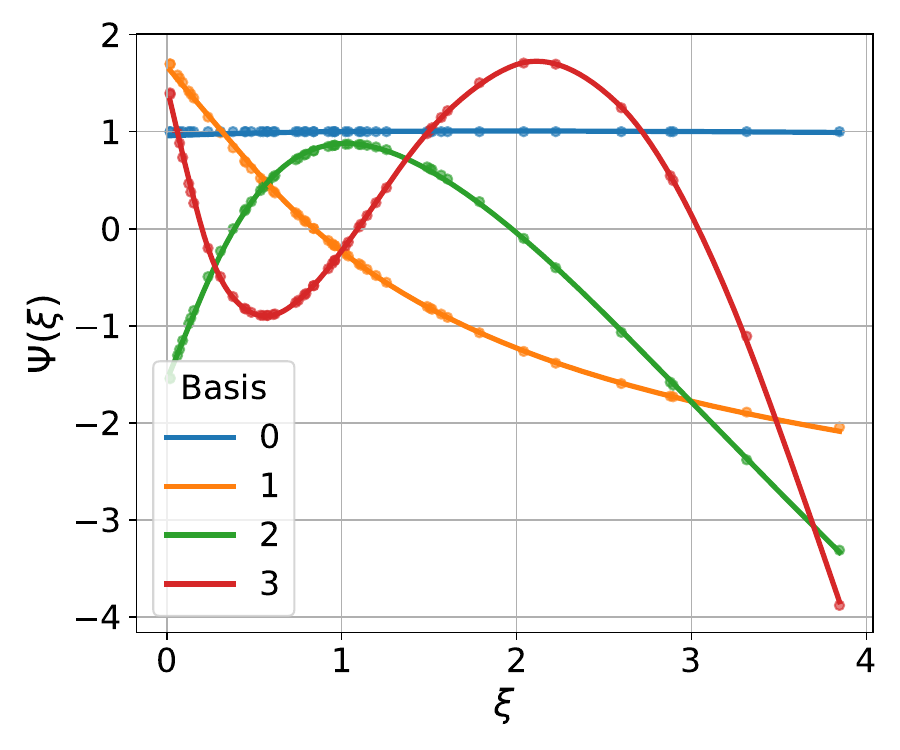}
    \caption{}
  \end{subfigure}
    \begin{subfigure}[b]{0.4\textwidth}
    \includegraphics[width=\textwidth]{figs/pr1_stoch_basis_gamma.pdf}
    \caption{}
  \end{subfigure}
  \caption{
  Learned stochastic basis functions with training data sizes of (a) 50 and (b) 700.
  }
  \label{fig:1d-datasize-comp-stoch-basis}
\end{figure}

\begin{figure}[!ht]
  \centering
  \begin{subfigure}[b]{0.4\textwidth}
    \includegraphics[width=\textwidth]{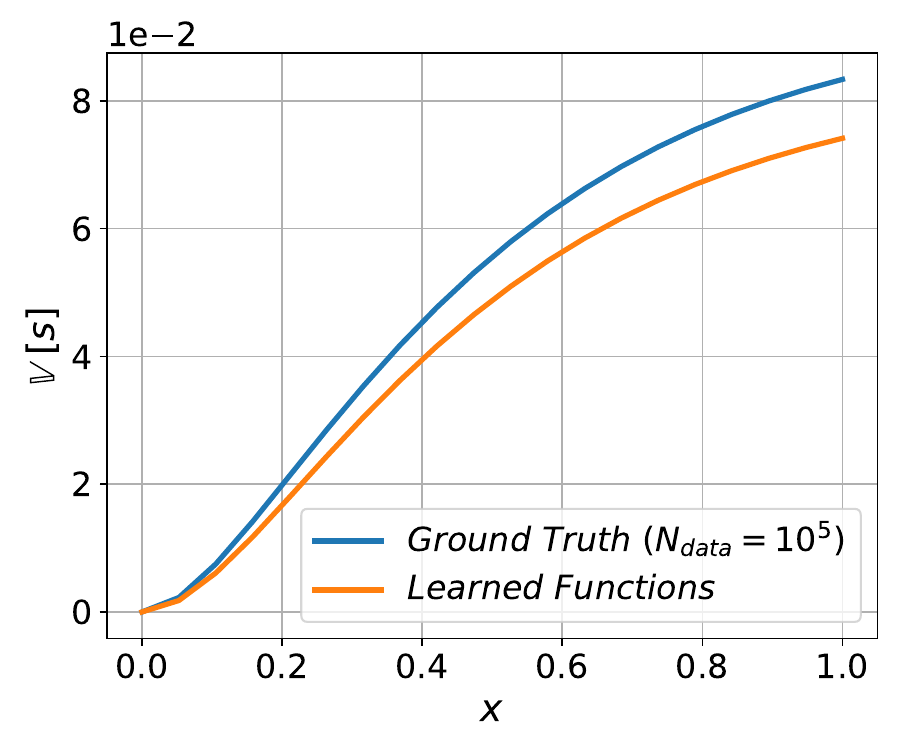}
    \caption{}
  \end{subfigure}
    \begin{subfigure}[b]{0.4\textwidth}
    \includegraphics[width=\textwidth]{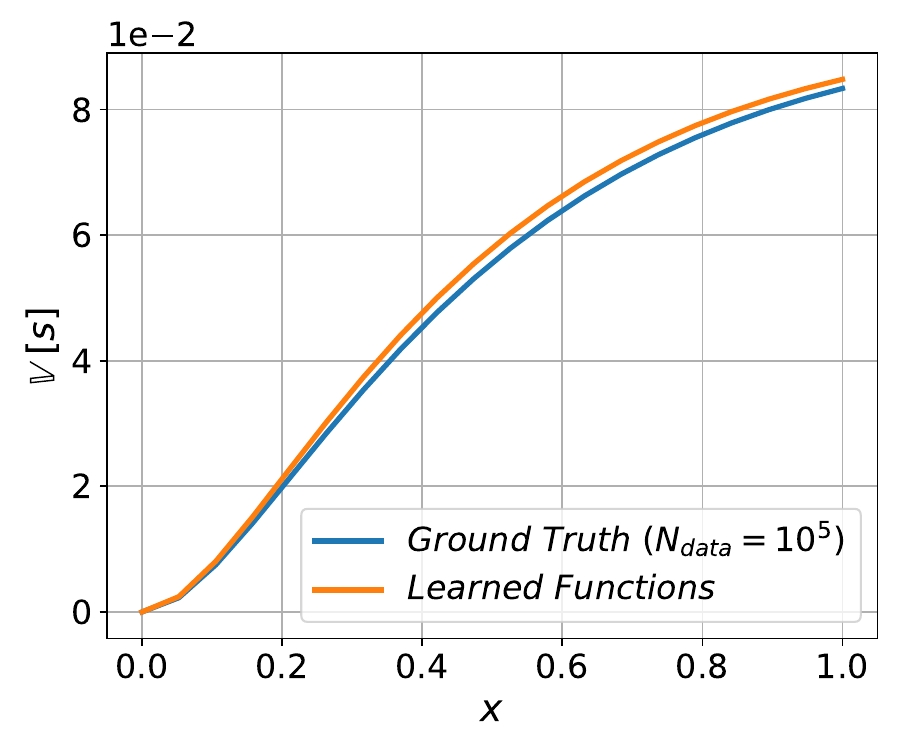}
    \caption{}
  \end{subfigure}
  \caption{
  Variance field estimations using Sobol's method with basis functions learned from (a) 50 and (b) 700 training realizations. The ground truth is based on Monte Carlo estimation with $10^5$ realizations.
  }
  \label{fig:1d-datasize-comp-var}
\end{figure}

\section{Neural Network Hyperparameters}
\label{appx:hyper-params}
In all the numerical examples, the ADAM optimizer \cite{diederik2014adam} from the \texttt{PyTorch} package \cite{paszke2019pytorch} is used.

Section \ref{sec:ex-xui-1d}: Both stochastic and deterministic basis functions are MLPs with two hidden layers of size 20, using \texttt{ELU} activation functions \cite{clevert2015fast}. The learning rate is set to $5 \times 10^{-4}$.

Section \ref{sec:ex-heat1D}: Both stochastic and deterministic basis functions are MLPs with two hidden layers of size 20, using \texttt{ELU} activation functions. The learning rate is set to $5 \times 10^{-4}$.

Section \ref{sec:ex-beam}: Both stochastic and deterministic basis functions are SIRENs \cite{sitzmann2020implicit,bahmani2024resolution} with two hidden layers of size 50 and a frequency parameter $10$. The learning rate is set to $5 \times 10^{-4}$.

Section \ref{sec:ex-heat-2d}: The stochastic basis functions are MLPs with three hidden layers of 200 units each, using \texttt{ReLU} activation functions. The deterministic basis functions are SIREN networks with three hidden layers of 100 units each and a frequency parameter of 15. The learning rate is set to $10^{-3}$.

Section \ref{sec:ex-dependence}: Both stochastic and deterministic basis functions are MLPs with two hidden layers of size 20 and activation function \texttt{ELU}. The learning rate is set to $10^{-3}$.

\end{document}